\documentclass[11pt]{article}

\usepackage[margin=3cm,footskip=1cm]{geometry}

\usepackage{exscale}
\usepackage{amsmath}
\usepackage{amsfonts}
\usepackage{amsthm}
\usepackage{amstext}
\usepackage{enumerate}
\usepackage{graphicx}

\newcommand{\cA}{\mathcal{A}}
\newcommand{\cD}{\mathcal{D}}
\newcommand{\cE}{\mathcal{E}}
\newcommand{\cO}{\mathcal{O}}
\newcommand{\cH}{\mathcal{H}}
\newcommand{\cB}{\mathcal{B}}
\newcommand{\cU}{\mathcal{U}}
\newcommand{\cK}{\mathcal{K}}

\newcommand{\cN}{\mathcal{N}}

\newcommand{\cC}{\mathcal{C}}
\newcommand{\cF}{\mathcal{F}}
\newcommand{\cI}{\mathcal{I}}
\newcommand{\cS}{\mathcal{S}}

\newcommand{\tv}{\tilde{v}}
\newcommand{\gothh}{\mathfrak{h}}
\newcommand{\tcO}{\widetilde{\cO}}

     \newcommand{\supp}{\operatorname{supp}}

\newcommand{\dbar}{{\bar\partial}}

     \newcommand{\ad}{{\operatorname{ad}}}
     \newcommand{\N}{{\mathbb{N}}}
     \newcommand{\NN}{{\mathbb{N}}}
     \newcommand{\R}{{\mathbb{R}}}
     \newcommand{\RR}{{\mathbb{R}}}
     
     \newcommand{\C}{{\mathbb{C}}}
     \newcommand{\CC}{{\mathbb{C}}}

\newcommand{\hM}{\widehat{M}}
\newcommand{\hH}{\widehat{H}}
\newcommand{\tH}{\widetilde{H}}
\newcommand{\tP}{\widetilde{P}}
\newcommand{\tA}{\widetilde{A}}
\newcommand{\tC}{\widetilde{C}}
\newcommand{\tN}{\widetilde{N}}
\newcommand{\tT}{\widetilde{T}}
\newcommand{\hcH}{\widehat{\mathcal{H}}}
\newcommand{\tcH}{\widetilde{\mathcal{H}}}

\newcommand{\slim}{{\rm s-}\lim}

\newcommand{\e}{{\rm e}}
\renewcommand{\i}{{\rm i}}
\renewcommand{\d}{{\rm d}}

\newcommand{\aux}{{\rm aux}}

\renewcommand{\Re}{{\rm Re}\,}
\renewcommand{\Im}{{\rm Im}\,}

\newcommand\inp[2][]{#1 \langle #2#1\rangle}
\newcommand\parb[2][]{#1 \big ( #2#1\big )}

\newcommand\bin[2]{\begin{pmatrix} #1 \\ #2 \end{pmatrix}}

\newcommand{\phy}{{\rm phy}}
\newcommand{\pp}{{\rm pp}}

\renewcommand{\exp}{{\rm exp}}
\newcommand{\mand}{\text{ and }}
\newcommand{\mfor}{\text{ for }}
\newcommand{\rome}{\mathrm{e}}
\newcommand{\fin}{\mathrm{fin}}
\newcommand{\roml}{\mathrm{l}}
\newcommand{\romr}{\mathrm{r}}
\newcommand{\uell}{{\underline{\ell}}}

\newcommand{\Nel}{\mathrm{N}}
\newcommand{\PF}{\mathrm{PF}}
\newcommand{\Mo}{\mathrm{Mo}}
\newcommand{\mo}{\mathrm{Mo}}

\newcommand{\la}{\langle}
\newcommand{\ra}{\rangle}
\newcommand{\dprime}{{\prime\prime}}

\def\bbbone{{\mathchoice {\rm 1\mskip-4mu l} {\rm 1\mskip-4mu l}
{\rm 1\mskip-4.5mu l} {\rm 1\mskip-5mu l}}}
\newcommand{\one}{\bbbone}

\renewcommand{\thesection}
{\arabic{section}}                     
\renewcommand{\theequation}
{\thesection.\arabic{equation}}        

     \theoremstyle{plain}
     \newtheorem{thm}{Theorem}[section]
     \newtheorem{prop}[thm]{Proposition}
     \newtheorem{lemma}[thm]{Lemma}
      \newtheorem{cor}[thm]{Corollary}
     \theoremstyle{definition}
     
     \newtheorem{defn}[thm]{Definition}
     \newtheorem{example}[thm]{Example}
     
     \newtheorem{cond}[thm]{Condition}
     \newtheorem{conds}[thm]{Conditions}

     \newtheorem{remark}[thm]{Remark}
     \newtheorem{remarks}[thm]{Remarks}
\newtheorem*{remarks*}{Remarks}
\newtheorem*{remark*}{Remark}

     \numberwithin{equation}{section}

\begin{document}

\bibliographystyle{amsplain}

\setcounter{page}{0}

\title{Regularity of Bound States}

\author{Jeremy Faupin\thanks{Partially Supported by Center for Theory in Natural Sciences, Aarhus University}
\footnote{email: jeremy.faupin@math.u-bordeaux1.fr}
\\  \hspace{0,2cm} Institut de Math{\'e}matiques de Bordeaux \hspace{0,2cm} \\
    Universit{\'e} de Bordeaux 1\\
    France \\
\and Jacob Schach M{\o}ller\footnote{email: jacob@imf.au.dk}
 \ \   Erik Skibsted\footnote{email: skibsted@imf.au.dk} \\
\hspace{0,2cm} Department of Mathematical Sciences \hspace{0,2cm} \\ Aarhus University \\ Denmark
}
\date{\today}

\maketitle

\begin{abstract} We study regularity of bound states pertaining to embedded eigenvalues
of a self-adjoint operator $H$, with respect to an auxiliary operator $A$ that is conjugate to $H$
in the sense of Mourre. We work within the framework of singular Mourre theory which
enables us to deal with confined massless Pauli-Fierz models, our primary example, and many-body AC-Stark Hamiltonians.
In the simpler context of regular Mourre theory our results boils down to an improvement of
results obtained recently in \cite{CGH}.
\end{abstract}

\thispagestyle{empty}
\newpage
\tableofcontents
\thispagestyle{empty}
\newpage

\section{Introduction}

This paper is the first in a series of two dealing with embedded eigenvalues and their bound states,
in the context of local commutator methods.

 In this paper we study regularity of  bound states with respect to a conjugate operator,
in the context of singular Mourre theory. In the second paper \cite{FMS} we use the results obtained here
to do second order perturbation theory of embedded eigenvalues, in particular
we establish the validity of Fermi's golden rule for an abstract class of Hamiltonians.
We remark that by singular Mourre theory we refer to the situation where the first commutator
is not controlled by the Hamiltonian itself, as in \cite{DJ,Go,GGM1,GGM,MS,Sk}. Regular Mourre theory
refers to the setup considered in \cite{ABG}. See also \cite{AHS,BFSS,DG,FGS,GJ,HS,Mo}.

 Our main motivation is applications to massless models from quantum field theory.
In particular our results apply to the massless confined Nelson model
at arbitrary coupling strength. We can deal with infrared singularities
that are slightly weaker than the physical one, that is we can handle
singularities of the form $|k|^{-\frac12 + \epsilon}$, for some $\epsilon>0$.
As a by-product of our methods we also establish that all bound states are in the domain
of the number operator.

In Section~\ref{sec_models} we in fact deal with a larger class of quantum field theory models, sometimes
called Pauli-Fierz models,
which includes the Nelson model. For simplicity we present our results here in the context of the Nelson model,
which we introduce in Subsection~\ref{subsec-NelsonIntro} below.
The reader can also consult \cite[Subsection~2.3]{GGM} for a discussion of the models
considered in this paper and its sequel.

In Section~\ref{AC--Stark model} we apply the results of this paper to many-body AC-Stark Hamiltonians
where we obtain new regularity results. See Subsection~\ref{The AC--Stark model} below for a formulation of the model
and the result.

\subsection{The Nelson Model}\label{subsec-NelsonIntro}

The model describes a confined atomic system coupled to a massless scalar quantum field.
The Hamiltonian $K$ of the atomic system is
\begin{equation}\label{K}
K = \sum_{i=1}^P \frac1{2m_i}\Delta_i + \sum_{i<j} V_{ij}(x_i-x_j) + W(x_1,\dots,x_P)
\end{equation}
acting on $\cK = L^2(\RR^{3 P})$. Here $m_i>0$ denotes the mass of the $i$'th particle located at $x_i\in \RR^3$.
We write $x = (x_1,\dots,x_P)\in \RR^{3P}$.
The external potential $W$ is the confinement and must satisfy
\begin{itemize}
\item[({\bf W0})] $W\in L^2_{\mathrm{loc}}(\RR^{3P})$ and there exist positive constants $c_0,c_1$ and $\alpha>2$
such that  $W(x)\geq c_0 |x|^{2\alpha} - c_1$.
\end{itemize}
As for the pair potentials $V_{ij}$, they should satisfy
\begin{itemize}
\item[({\bf V0})] The $V_{ij}$'s are $\Delta$-bounded with relative bound $0$.
\end{itemize}

The Hilbert space for the scalar bosons is the symmetric Fock-space $\cF = \Gamma(L^2(\RR^3))$
and the kinetic energy for the massless bosons is $\d \Gamma(|k|)$, the second quantization of
the operator of multiplication with the massless dispersion relation $|k|$. The
uncoupled Hamiltonian, describing the atomic system and the scalar field is
$K\otimes\one_{\cF} + \one_\cK\otimes \d\Gamma(|k|)$, as an operator on the full Hilbert space
\[
\cH = \cK\otimes\cF.
\]

Our next task is to introduce a coupling of the form
\begin{equation}\label{Irho}
I_\rho(x) = \sum_{i=1}^P \phi_\rho(x_i),
\end{equation}
where $\phi_\rho(y)$ is an ultraviolet and infrared regularized field operator
\[
\phi_{\rho}(y) = \frac1{\sqrt{2}}
\int_{\RR^3} \big(\rho(k) \e^{-\i k\cdot y} a^*(k) + \overline{\rho(k)}\e^{\i k\cdot y} a(k)\big) \d k.
\]
We assume purely for simplicity that $\rho$ only depends on $k$ through its modulus.
To conform with the notation used
in \cite{GGM}, we introduce
\[
\tilde{\rho}(r) = r \rho(r,0,0), \ \ \textup{such that} \ \  |k|\rho(k) = \tilde{\rho}(|k|).
\]
For the  interacting Hamiltonian, indexed by the coupling function $\rho$,
\begin{equation}\label{HrhoNel}
H_\rho^\Nel = K\otimes \one_{\Gamma(\gothh)} + \one_\cK \otimes \d\Gamma(|k|) + I_\rho(x)
\end{equation}
to be essentially self-adjoint on $\cD(K)\otimes\Gamma_\fin(C_0^\infty(\RR^3))$,
we need the following basic assumption on $\rho$.
\begin{itemize}
\item[({\bf $\rho$1})] $\int_0^\infty (1+r^{-1})|\tilde{\rho}(r)|^2 \d r <\infty$.
\end{itemize}
Here $\Gamma_\fin(V)$ denotes the subspace of $\cF$
consisting of elements $\eta$ with only finitely many $n$-particle components $\eta^{(n)}$ nonzero,
and those that are nonzero lie in the $n$-fold algebraic tensor product of the subspace $V\subseteq L^2(\RR^3)$.
Note that $\Gamma_\fin(V)$ is dense in $\cF$ if $V$ is dense in $L^2(\RR^3)$.

In order to formulate the remaining assumption on $\rho$ we introduce a function
$d\in C^\infty((0,\infty))$,
which measures the amount of infrared regularization carried by $\rho$.
It should, for some $C_d>0$, satisfy
\begin{equation}\label{d1}
d(r) = 1, \ \textup{for} \ r\geq 1, \ \  -C_d\frac{d(r)}{r} \leq d^\prime(r) < 0 , \ \  \lim_{r\to 0+} d(r) = +\infty.
\end{equation}
Note that the conditions above imply that $1\leq d(r)\leq r^{-C_d}$, for $r\in (0,1]$.
In order to simplify some expressions below we make the additional assumption that
\begin{equation}\label{d2}
\forall r\in (0,1]: \ \ d(r) \leq C_d^\prime r^{-\frac12},
\end{equation}
for some $C_d'>0$. In practice we want to construct
a $d$ with as weak a singularity as possible, so this extra assumption is no restriction.
We formulate the remaining conditions on $\rho$, of which the two first also appeared in \cite{GGM}.
\begin{itemize}
\item[({\bf $\rho$2})] $\int_0^\infty (1+r^{-1}) d(r)^2 [r^{-2}|\tilde{\rho}(r)|^2 +
|\frac{\d \tilde{\rho}}{\d r}(r)|^2]
\d r <\infty$.
\item[({\bf $\rho$3})] $\int_0^\infty |\frac{\d^2 \tilde{\rho}}{\d r^2}(r)|^2 \d r <\infty$.
\item[({\bf $\rho$4})] $\int_0^\infty r^4 |\tilde{\rho}(r)|^2 dr <\infty$.
\end{itemize}
We remark that {\bf ($\rho$2)} and {\bf ($\rho$4)} implies  {\bf ($\rho$1)}.
A typical form of $\rho$, and hence $\tilde{\rho}$, would be
\begin{equation}\label{Nelson-rho}
\rho(k) = \e^{-\frac{|k|^2}{2\Lambda^2}} |k|^{-\frac12 + \epsilon}, \ \ \tilde{\rho}(r)
= \e^{-\frac{r^2}{2\Lambda^2}} r^{\frac12 + \epsilon}.
\end{equation}
One can construct a $d$ by gluing together the functions $1$ and $r^{-\epsilon'}$, with
$0<\epsilon' < \min\{\epsilon,1/2\}$.
The parameters $\Lambda$ and $\epsilon$ are the ultraviolet respectively
infrared regularization parameters.
Ideally we would like to have $\Lambda = \infty$ and $\epsilon=0$. For the conditions
({\bf $\rho$1})--({\bf $\rho$4})
to be satisfied we must have $0< \Lambda<\infty$ and $\epsilon> 1$.
Observe that it is the condition {\bf ($\rho$3)} on the second
derivative of $\tilde{\rho}$ that causes the strongest restriction on $\epsilon$.

Observe that the set of $\rho$'s satisfying  ({\bf $\rho$1}) -- ({\bf $\rho$4})
is a complex vector space $\cI_\Nel(d)\subseteq L^2(\RR^3)$, which can be equipped with a norm
matching the four conditions. That is
\begin{equation}\label{normNel}
\|\rho\|^2_\Nel := \int_0^\infty \Big\{(r^4+d(r)^2 r^{-3})|\tilde{\rho}(r)|^2
+(1+r^{-1})d(r)^2\big|\frac{\d \tilde{\rho}}{\d r}(r)\big|^2 +
\big| \frac{\d^2 \tilde{\rho}}{\d r^2}(r)\big|^2
\Big\} \d r.
\end{equation}

In order to formulate our main theorem, we need to introduce an operator conjugate to $H_\rho^\Nel$.
We use the one constructed in \cite{GGM}, for which a Mourre estimate has been established
under the assumptions above. Let $\chi\in C_0^\infty(\RR)$, with $0\leq \chi\leq 1$,
$\chi(r) = 1$ for  $|r|<1/2$, and $\chi(r)=0$ for $|r|>1$. For $0<\delta\leq 1/2$ we define
a function on $(0,\infty)$ by
\[
s_\delta(r) = \chi(r/\delta) d(\delta)r^{-1} + (1-\chi)(r/\delta)d(r)r^{-1}.
\]
Using this function we construct a vector-field by $\vec{s}_\delta(k) = s_\delta(|k|)k$,
which equals $k/|k|$ for $|k|>1$ and $d(\delta)k/|k|$ for $|k|<\delta/2$.
The conjugate operator on the one-particle sector is
\begin{equation}\label{Eq-adeltaPhys}
a_\delta = \frac12(\vec{s}_\delta \cdot \i \nabla_k + \i\nabla_k \cdot \vec{s}_\delta).
\end{equation}
The operator is symmetric and closable on $\{f\in C_0^\infty(\RR^3) | f(0)=0\}$.
We denote again by $a_\delta$ its closure which is a maximally symmetric operator, but not self-adjoint.
It is a modification, near $k=0$, of the generator of radial translations
$a=(\frac{k}{|k|}\cdot \i\nabla_k +
\i \nabla_k \cdot\frac{k}{|k|})/2$. The conjugate operator is now the maximally symmetric operator
\[
A_\delta = \one_{\cK}\otimes \d\Gamma(a_\delta).
\]
The second quantization $\d\Gamma(a)$ of the generator of radial
translations works as conjugate operator if one stays close to the uncoupled system.
See \cite{DJ,Go,Sk}. It is not known
if one really needs the modified generator of radial translations
$A_\delta$ in order to get a Mourre estimate at arbitrary coupling.

For an eigenvalue $E\in\sigma_\pp(H_{\rho}^\Nel)$ we write $P_{\rho}$ for the associated
eigenprojection. It is known from \cite{GGM} that $P_{\rho}$ has finite dimensional range.
Finally we need the number operator
\[
\cN = \one_{\cK}\otimes \d\Gamma(\one_{L^2(\RR^3)}).
\]
We will make use of the same notation for the (usual) number operator on $\cF$.
Our main result of this paper, formulated in terms of the Nelson model, is

\begin{thm}\label{Main-Nelson} Suppose {\rm (}{\bf W0}{\rm )} and {\rm (}{\bf V0}{\rm)}.
Let $E_0\in\RR$ and $\rho_0\in \cI_\Nel(d)$ be given. There exist $0<\delta\leq 1/2$, $r>0$ and $C>0$
such that for any $\rho\in \cI_\Nel(d)$, with $\|\rho-\rho_0\|_\Nel\leq r$,
and $E\in \sigma_{\pp}(H^\Nel_{\rho})\cap (-\infty,E_0]$ we have
\[
P_{\rho}:\cH\to \cD\big(\cN^\frac12 A_\delta\big)\cap \cD\big(A_\delta\cN^\frac12\big)\cap
\cD\big(\cN\big)
\]
and
\[
\big\|\cN^\frac12 A_\delta P_{\rho}\big\| + \big\|A_\delta\cN^\frac12 P_{\rho}\big\| +
\big\|\cN P_{\rho}\big\| \leq C.
\]
\end{thm}

We remark that for any $\delta>0$ small enough, one can find $r$ and $C$ such that the conclusion of the theorem holds. See Theorem~\ref{Thm-MainPF}. The above suffices for our purpose and is a cleaner statement.

We can implement a unitary transformation, the so-called Pauli-Fierz transform, 
which has the effect of smoothening the infrared
singularity. Let $U_{\rho} = \exp(-\i P\phi_{\i\rho/|k|}(0))$
be the unitary transformation with
\[
 U_{\rho} a(k) U_\rho^* = a(k) - \frac{P\rho(k)}{\sqrt{2}|k|} \ \textup{and} \ 
  U_\rho a^*(k) U_\rho^* = a^*(k) - \frac{P\overline{\rho(k)}}{\sqrt{2}|k|}.
\]
For the transformation $U_\rho$ to be well-defined we must require that $\int_{\RR^3} |k|^{-2}|\rho(k)|^2\d k<\infty$.
To achieve this we strengthen ({\bf$\rho$1}) to read
\begin{itemize}
\item[({\bf $\rho$1'})] $\int_0^\infty (1+r^{-2})|\tilde{\rho}(r)|^2\d r <\infty$.
\end{itemize}

We then get
\begin{equation}\label{HrhoNelPrime}
H^{\Nel'}_\rho = (\one_\cK\otimes U_\rho) H^\Nel_\rho (\one_\cK\otimes  U_\rho)^*  =
K_\rho\otimes\one_\cF + \one_\cK\otimes \d\Gamma(|k|)
+I_\rho(x)-I_\rho(0),
\end{equation}
where
\begin{equation}\label{Krho}
K_\rho = K - \sum_{i=1}^P v_\rho(x_i) + \frac{P^2}2\int_0^\infty r^{-1}|\tilde{\rho}(r)|^2 \d r \one_\cK
\end{equation}
and
\begin{equation}\label{vrho}
v_\rho(y) = P \int_{\RR^3} \frac{|\rho(k)|^2}{|k|}\cos(k\cdot y)\d k.
\end{equation}
Observe that
\[
\phi_\rho(y) - \phi_\rho(0) = \frac1{\sqrt{2}} \int_{\RR^3} \big(\rho(k) (\e^{-\i k\cdot y}-1) a^*(k) +
\overline{\rho(k)}(\e^{\i k\cdot y}-1) a(k)\big) \d k.
\]
The estimate
\begin{equation}\label{PFtransbnd}
|\e^{\pm \i k\cdot y}-1| \leq \max\{2,|k||y|\} \leq 2\frac{|k|}{\la k \ra}\la y\ra,
\end{equation}
with $\la \eta\ra = (1+|\eta|^2)^{1/2}$, enables us to
extract an extra infrared regularization  using the decay in $x$ supplied by the confinement condition ({\bf W0}).
Keeping \eqref{d2} and ({\bf $\rho$1'}) in mind, the remaining two assumptions on $\rho$ now weaken to
\begin{itemize}
\item[({\bf $\rho$2'})]   $\int_0^\infty  |\frac{\d \tilde{\rho}}{\d r}(r)|^2\d r <\infty$.
\item[({\bf $\rho$3'})] $\int_0^\infty r^2
|\frac{\d^2 \tilde{\rho}}{\d r^2}(r)|^2/(1+r^2)\d r<\infty$.
\end{itemize}
The condition {\bf ($\rho$4)}, being an ultraviolet condition, is unchanged.
For the choice \eqref{Nelson-rho} to satisfy  ({\bf $\rho$1'})--({\bf $\rho$3'}) and {\bf ($\rho$4)}
we must have $0< \Lambda<\infty$ and $\epsilon > 0$. Here the first three conditions on $\rho$
all require $\epsilon>0$.

Observe again that the set of $\rho$ satisfying  ({\bf $\rho$1'})--({\bf $\rho$3'}) and {\bf ($\rho$4)}
is a complex vector space $\cI_\Nel^\prime(d)$.
We introduce the natural norm
\[
\|\rho\|^2_{\Nel'} :=
\int_0^\infty \Big\{(r^4+r^{-2})|\tilde{\rho}(r)|^2
+\big|\frac{\d \tilde{\rho}}{\d r}(r)\big|^2 +
\frac{r^2}{1+r^2}
\big| \frac{\d^2 \tilde{\rho}}{\d r^2}(r)\big|^2
\Big\} \d r.
\]
Fix a $\rho_0\in \cI_\Nel^\prime(d)$.
There are now two avenues one can follow. Either one can continue as above,
and for each $\rho$ in a $\|\cdot\|_{\Nel'}$-ball around $\rho_0$
we apply the transformation $U_\rho$ to arrive at the more regular Hamiltonian $H^{\Nel'}_\rho$
that we can fit into our class of Pauli-Fierz models. A second option
would be to apply the same transformation $U_{\rho_0}$ regardless
of $\rho$ chosen near $\rho_0$. The advantage of this is two-fold:
Firstly, we would be working in the same coordinate system for all $\rho$'s,
which in the context of perturbation theory, cf. \cite{FMS}, is the most natural.
Secondly, in this way the Hamiltonian will have a linear dependence on the 'perturbation' $\rho-\rho_0$,
which is a requirement in \cite{FMS}. The drawback is that
$\rho-\rho_0$ has  to
be an element of $\cI_\Nel(d)$, and for example cannot be a small multiple of $\rho_0$.

To implement the latter approach, we now let $\rho = \rho_0+\rho_1$, with $\rho_1\in \cI_\Nel(d)$, the space of regular interactions. We then employ the transformation $U_{\rho_0}$ which yields
the transformed Hamiltonian
\begin{equation}\label{HrhoNelDPrime}
H_\rho^{\Nel''} = (\one_\cK\otimes U_{\rho_0}) H^\Nel_\rho (\one_\cK\otimes  U_{\rho_0})^*
= H^{\Nel'}_{\rho_0} + I_{\rho_1}(x) - \sum_{i=1}^P v_{\rho_0,\rho_1}(x_i), 
\end{equation}
where
\begin{equation}\label{vrho0rho1}
v_{\rho_0,\rho_1}(y) = P\int_{\RR^3}\Re\left\{ \frac{\rho_1(k) \overline{\rho_0(k)}}{|k|} e^{-\i k\cdot y}\right\} \d k.
\end{equation}

For an eigenvalue $E\in\sigma_\pp(H^{\Nel}_{\rho})$ we write 
$P'_{\rho} = (\one_\cK\otimes U_\rho) P_\rho (\one_\cK\otimes U_\rho)^*$ for the associated
eigenprojection for $H_\rho^{\Nel'}$, and 
$P_\rho^{\prime\prime} =(\one_\cK\otimes U_{\rho_0})P_\rho(\one_\cK\otimes  U_{\rho_0})^*$ for the associated 
eigenprojection for $H_\rho^{\Nel''}$. Again $P'_{\rho}$ and $P_\rho^{\prime\prime}$ have finite dimensional ranges.
Theorem~\ref{Thm-MainPF} can be applied to the transformed Hamiltonian and we arrive at
the following theorem.

\begin{thm}\label{Main-NelsonPrime} 
Suppose {\rm (}{\bf W0}{\rm )} and {\rm (}{\bf V0}{\rm)}. Let $E_0\in\RR$ and
 $\rho_0\in\cI_\Nel^\prime(d)$ be given. There exist $0<\delta\leq 1/2$, $r>0$ and $C>0$ such that
\begin{enumerate}[1)]
\item for any $\rho\in \cI_\Nel^\prime(d)$ with $\|\rho-\rho_0\|_{\Nel^\prime}\leq r$
and $E\in \sigma_\pp(H_\rho^{\Nel})\cap (-\infty, E_0]$ we have
\[
P'_{\rho}:\cH\to \cD\big(\cN^\frac12 A_\delta\big)\cap \cD\big(A_\delta\cN^\frac12\big)\cap
\cD\big(\cN\big)
\]
and
\[
\big\|\cN^\frac12 A_\delta P'_{\rho}\big\| + \big\|A_\delta\cN^\frac12 P'_{\rho}\big\| +
\big\|\cN P'_{\rho}\big\| \leq C.
\]
\item for any $\rho_1\in\cI_\Nel(d)$ with $\|\rho_1\|_\Nel \leq r$ and $E\in \sigma_\pp(H_{\rho}^{\Nel})\cap (-\infty,E_0]$, where $\rho = \rho_0+\rho_1$, we have
\[
P''_{\rho}:\cH\to \cD\big(\cN^\frac12 A_\delta\big)\cap \cD\big(A_\delta\cN^\frac12\big)\cap
\cD\big(\cN\big)
\]
and
\[
\big\|\cN^\frac12 A_\delta P''_{\rho}\big\| + \big\|A_\delta\cN^\frac12 P''_{\rho}\big\| +
\big\|\cN P''_{\rho}\big\| \leq C.
\]
\end{enumerate}
\end{thm}

Unfortunately the transformation $U_\rho$, with $\rho\in\cI_\Nel^\prime(d)$, is too singular to allow for a recovery of the full set of regularity
results for the original Hamiltonian $H^\Nel_{\rho}$, as in Theorem~\ref{Main-Nelson}.
 The only thing that remains after undoing the transformation
is the following corollary to Theorem~\ref{Main-NelsonPrime}~1). The same argument using Theorem~\ref{Main-NelsonPrime}~2) would give a weaker result. Theorem~\ref{Main-NelsonPrime}~2)
will however play a role in \cite{FMS}.

\begin{cor} 
Suppose {\rm (}{\bf W0}{\rm )} and {\rm (}{\bf V0}{\rm)}. Let $E_0\in\RR$ and
$\rho_0\in\cI_\Nel^\prime(d)$ be given. There exist $0<\delta\leq 1/2$, $r>0$ and $C>0$ such that
for any $\rho\in \cI_\Nel^\prime(d)$ with $\|\rho-\rho_0\|_{\Nel^\prime}\leq r$
and $E\in \sigma_\pp(H_\rho^{\Nel})\cap (-\infty, E_0]$ we have
\[
P_{\rho}:\cH\to \cD\big(\cN\big) \ \textup{and} \ 
\big\|\cN P_{\rho}\big\| \leq C.
\]
\end{cor}

 We make a number of remarks concerning the results above.

 The domain of $a_\delta$ is independent of $\delta$, and in fact equals the domain of the generator
of radial translations. The same is (presumably) false for the second quantized versions.
This is the reason for
the somewhat unpleasant formulation of the theorems in terms of $A_\delta$. It should be read in the
context of Mourre's commutator method, and in \cite{FMS} we need the regularity
formulated in terms of $A_\delta$.

 The statement that bound states are in the domain of the number operator is new.
Previously it was only known that bound states are in the domain of $\cN^{1/2}$.
See \cite{GGM}.

 The reader should first and foremost read the results above with $\rho=\rho_0$. In the sequel
\cite{FMS} we need the locally uniform version to deduce a Fermi golden rule under minimal assumptions.
In traditional approaches to Fermi's golden rule, one typically require
unperturbed bound states to be in the domain of the square of the conjugate operator.
See \cite{AHS,HS,MS}. In \cite{FMS} we reduce the requirement to bound states $\psi$ being in the domain
of the conjugate operator itself, at the expense of a need for the norm $\|A_\delta\psi\|$
to be bounded uniformly in $\rho$ in a ball around the unperturbed
coupling function $\rho_0$ 
and uniformly in $E$ running  over  eigenvalues of $H_\rho$ in a fixed compact interval. This motivates
the somewhat unorthodox formulation in Theorem~\ref{Main-Nelson}.

 The conditions ({\bf $\rho$3}) and ({\bf $\rho$3'}) come from a need
of handling the double commutator $[[H_\rho,A_\delta],A_\delta]$. It is not a priori obvious that
we should be able to place bound states in the domain of $A_\delta$ with control of just two commutators.
In the context of regular Mourre theory the question is addressed in \cite{CGH} where the authors need
three commutators to conclude a result of this type. In view of the infrared singularity, it is crucial
to minimize
the number of commutators needed. The following example illustrates that
if one desires bound states to be in the domain of the $k$'th power of a conjugate operator,
one needs at least control of $k+1$ commutators.

\begin{example}\label{RegExample}
 Consider the one-dimensional Schr{\"o}dinger
  operator $H=-\triangle +V$ on $\mathcal{H}=L^2(\R)$ given by a
  rank-one potential $V=|\phi\rangle \langle \phi|$ where $\phi \in
  \mathcal{H}$ obeys the following properties: Suppose that  (in momentum
  space) $\hat \phi=\hat \phi_1+\hat \phi_2$, where  $\hat
  \phi_1\in C^\infty_c(]-1,1[)$, $\hat\phi_2\in C^\infty(]1,\infty[)$
  and for some $R>1$ the support $\supp \hat\phi_2\subseteq [1, R]$. Suppose there exists $\epsilon\in
    ]0,\tfrac12[$ such that for all $k\in
    \N\cup\{0\}$
  \begin{equation*}
  \frac{\d^k}{\d \xi^k}\Big\{\hat{\phi}_2(\xi)-(\xi^2-1)^{k_0+\tfrac12+\epsilon}\Big\}
=O\Big((\xi^2-1)^{k_0-k+\tfrac32+\epsilon}\Big)\textup{ as }|\xi|\searrow 1.
  \end{equation*} Finally suppose
  \begin{equation*}
    \int_\RR |\hat \phi(\xi)|^2(\xi^2-1)^{-1}\d \xi=-1.
  \end{equation*}
Then $\psi=(-\triangle-1)^{-1}\phi$ is a
  bound state (with eigenvalue $\lambda=1$). Let $A=\tfrac12(p\cdot
  x+x\cdot p)$ be the generator of dilations. We have  $\phi\in
  \cD(A^{k_0+1})$ and $\psi\in
  \cD(A^{k_0})$,  while indeed  $\psi\notin \cD(A^{k_0+1})$
(intuitively it should be  expected  that in fact $\psi\notin
  \cD(\langle A\rangle ^{k_0+ \epsilon'})$  for
  $\epsilon'\geq\epsilon$).
\end{example}

In Section~\ref{Sec-AbstractResults} our abstract regularity results are formulated.
In the context of regular Mourre theory,
as considered in \cite{CGH}, we need control of one less commutator,
 which given the example above is optimal.

\subsection{Singular Mourre Theory}

 Consider the operator $M_\omega$ of multiplication in the momentum space
 $L^2(\RR^d)$ by  a dispersion relation
$\omega$ assumed to be locally Lipschitz. The connection between dynamics and structure of the spectrum of a self-adjoint
operator is fairly well understood, starting from Kato-smoothness and the RAGE theorem \cite{RS}.
When looking for a conjugate operator, one should study
the dynamics of the operator $M_\omega$. It is natural to
identify what states have (at least) ballistic motion, that is
find states $\psi_0$ satisfying
\[
\la x^2\ra_{\psi_t} \geq c t^2,
\]
for some  $c>0$. Here $\psi_t = \exp(-\i t M_\omega)\psi_0$. The position operator
$x$ is equal to $\i \nabla_k$. We can compute this quantity explicitly and we get
\begin{align*}
\la x^2\ra_{\psi_t} & =  \la x^2\ra_{\psi_0} + \int_0^t
\la x\cdot\nabla\omega + \nabla\omega\cdot x \ra_{\psi_s}\d s\\
& =  \la x^2\ra_{\psi_0} +t\la x\cdot\nabla\omega + \nabla\omega\cdot x \ra_{\psi_0}
+ t^2 \la |\nabla\omega|^2\ra_{\psi_0}.
\end{align*}
We observe that if $\psi_0$ has support away from zeroes of $\nabla\omega$, then the motion is
at least ballistic. More precisely this is the case  if $\mathrm{essinf}_{k\in\supp{\psi_0}} |\nabla \omega(k)|\geq c>0$.

If $\omega = k^2$, the standard non-relativistic dispersion relation, we find that
$\psi_0$ should be localized away from $0$ in momentum space. Since $|\nabla\omega|^2 = 4\omega$,
the requirement on $\psi_0$ can also be expressed as $\psi_0\in E_{M_{\omega}}([c/4,\infty))L^2(\RR^d)$,
where $E_{M_\omega}$ denotes the spectral projections associated with the self-adjoint operator
$M_\omega$.
We observe that the energy $0$ has a special significance for the case $\omega= k^2$
and is called a threshold, in the sense that states localized in energy near a
threshold may not have strict ballistic motion.

A second example is $\omega =|k|$. Here we observe that $|\nabla\omega| = 1$, and hence all states
$\psi_0$ will exhibit ballistic motion. In other words this dispersion relation does not have thresholds.
This of course reflects the constant (momentum independent) speed of light.
See \cite[Subsection~1.2]{GGM1} for a discussion of general dispersion relations.

When picking a conjugate operator in Mourre theory, one is precisely looking for
an observable $a$ with at least ballistic growth. The choice often  used is the Heisenberg
derivative of $x^2$, where $x$ is some suitably chosen position observable.
That is, one would naturally be lead to consider
\[
a = \frac12( x\cdot\nabla\omega + \nabla\omega\cdot x).
\]
This is for example the case for the $N$-body problem,
see e.g. \cite{AHS,Ca,CGH,HS}, and in the case of field theory see \cite{DG,DJ,FGSch1,FGSch2,GGM,Sk},
where the position is the Newton-Wigner position $\d\Gamma(x)$. The free energy is $\d\Gamma(M_\omega)$,
and we get as conjugate operator  $A = \d\Gamma(a)$, where $a$ is as above.

It is often advantageous to modify the so obtained conjugate operator, to simplify proofs,
or circumvent some technical issues. In this paper we need the modified generator of translations
$A_\delta$ from \cite{GGM} in order to deal with the confined massless Nelson model,
and more generally confined massless Pauli-Fierz models.

There are two issues that come up naturally when following the above guidelines
for massless field theory models, like the Nelson model. One is already apparent in the
one-particle setup discussed above. If $\omega(k) = |k|$, the resulting conjugate operator
$a$, the generator of radial translations, does not have a self-adjoint realization.
This appears to be a purely technical complication, that becomes a serious issue
when one is in need of localizations in the operator $a$. The operator is not normal,
so we do not have spectral calculus at hand, only resolvents. This has so far not been a serious issue
when dealing with the limiting absorption principle \cite{DJ,GGM,HuSp,HS,Sk}, and perturbation theory around an
uncoupled system \cite{DJ,Go}. It does however become an obstacle when one tries to apply
the conjugate operator $a$ in the context of  scattering  theory \cite{G}.

In the present paper, non-self-adjointness of $a$ is also
a serious obstacle,
which we overcome, as in \cite{G}, by passing to
a so called expanded Hamiltonian. The idea is to write $L^2(\RR^d) \sim L^2(\RR_+)\otimes L^2(S^{d-1})$
and double the Hilbert space to $L^2(\RR)\otimes L^2(S^{d-1})$. The dispersion relation
in polar coordinates is just multiplication by $r$, which when extended linearly to negative $r$
gives rise to the self-adjoint conjugate operator $\i \partial/\partial r\otimes \one$.
We thus work with an expanded Hamiltonian, and in the end pull our results back to the physical
Hamiltonian.
The reader should keep this in mind when going through the abstract conditions in the following section.

However passing to an expanded Hamiltonian is not a silver bullet, it comes with a price.
The operator of multiplication by $r$ is no longer bounded from below, making it hard to utilize
energy localizations. For this reason we have to develop an abstract theory which does not
demand that any naturally occurring object can be controlled by the (expanded) Hamiltonian.

The second feature we want to discuss does not occur on the one-particle level,
but only after second quantization. The free commutator becomes
\[
\i [\d\Gamma(|k|),A] = \cN,
\]
where $\cN$ is the number operator. In the standard (regular)  commutator based methods,
one typically has the commutator bounded at least as a form on $\cD(H)$.
(This is for example a consequence of a $C^1(A)$ assumption.) This is not the case here and
we call such a situation \emph{singular}. One could of course avoid this issue by observing
that the operators involved conserve particle number, and then rescale $A$ by $1/n$ on the
$n$-particle sector. However, perturbations are typically expressed in terms of field operators,
and straying from second quantized conjugate operators give rise to  terms
from the commutator with the perturbation, that have so far not been controllable.

The $\d\Gamma(|k|)$-unboundedness of the number operator, has led authors to use
a different conjugate operator instead, namely the second quantized generator of dilation
given by $\d\Gamma((x\cdot k+ k\cdot x)/2)$,
normally associated with the dispersion relation $k^2$. Here the commutator with $\d\Gamma(|k|)$
is $\d\Gamma(|k|)$ itself, so the issue disappears. However, this choice induces
an artificial threshold at photon energy $0$, which for a coupled system turns all eigenvalues
of the atomic system into artificial thresholds. In order to circumvent this problem one can
modify the generator of dilation by building the
level shift from Fermi's golden rule into the conjugate operator. This was done in \cite{BFSS}
and gives rise to positive relatively bounded commutators, at weak coupling. There are however
disadvantages to this approach. It does not cover situations where symmetries may
cause embedded eigenvalues to persist to second order in perturbation theory. For the $N$-body problem
in quantum mechanics  one can for example
show that the underlying spectrum is absolutely continuous without a priori imposing Fermi's golden rule,
which can then subsequently be established \cite{AHS,HS}.
Works employing this choice of conjugate operator has, so far, not been able to
address what happens outside the regime of weak coupling, which may be an issue
since coupling constants typically  are explicitly given numbers.
In electron-photon models, the coupling constant involve the feinstructure constant $1/137$ and
in electron-phonon models from solid state physics, the coupling constants occurring
may even be of the order $1$. Effective coupling constants may also depend
on an ultraviolet cutoff, thus imposing apparently artificial limitations on the size of the cutoff.
Finally the restriction on the size of the coupling constant is always locally uniform in energy.
That is, all statements of this type holds only below a fixed $E_0$. Papers employing the generator of
dilation include \cite{BFS,BFSS,FGS}.

We remark that in \cite{Go}, the author modifies the generator of radial translation,
as it was done in \cite{BFSS} for the generator of dilations, in order to establish Fermi's golden rule.
We have no need for this construction since we follow the strategy of \cite{AHS,HS,MS}.

Instead of viewing the unboundedness of the first commutator with respect to $\d\Gamma(|k|)$
as a technical problem, one can also adopt the point of view that it is a feature of the model which can be
exploited. This is most obviously done for small coupling constants, where one gets a  positive
commutator globally in energy, modulo a compact error.
This was done in \cite{DJ,FGSch2,Go,Sk}. In \cite{GGM}
the extra positivity of the commutator is directly utilized to prove a Mourre estimate at
arbitrary coupling constant, the first (and so far only) such result for massless models.
 Another piece of information one can extract is that the number operator
has finite expectation in bound states. This was done in \cite{Sk} for small coupling constants and
generally in \cite{GGM}. A more subtle property is that one can obtain a stronger limiting absorption
principle, see \cite{GGM1,MS}, which has so far not found an application.
Here we prove in particular that bound states are in the domain
of the number operator, not just in its form domain.

We have not discussed  positive temperature models, where one has a
similar situation, except that so far no positive commutator
estimates at arbitrary coupling has been proven, regardless of choice of conjugate operator.
See e.g. \cite{FM} and references therein.



\subsection{The AC--Stark model}\label{The AC--Stark model}

The model describes a system of $N$  charged particles in a nonzero time-periodic
Stark-field with zero mean (AC-Stark field). The particles are here
taken three-dimensional and  we assume that the field is
$1$-periodic  and, for simplicity, that it is continuous i.e. that
$\tilde \cE\in C([0,1];\R^3)$. The Hamiltonian is of the form
\begin{equation}
  \label{eq:57}
  \tilde h(t)=\sum_{i=1}^{N }\parb{ \tfrac{
    p_i^2}{2m_i}-q_i\tilde \cE(t)\cdot x_i}+V;
\end{equation}
here $x_i$, $m_i$ and $q_i$ are the position, the mass and the charge of
the {\it i}'th particle, respectively,  and $p_i=-\i \nabla_{x_i}$ is its momentum.
The potential is of the form
\begin{equation}\label{eq:58}
  V=\sum_{1\leq i<j\leq N}v_{ij}(x_i-x_j),
\end{equation} where the pair-potentials obey

\begin{conds}\label{cond:ac-stark-model23} Let $k_0\in \N$ be given.
  For each pair $(i,j)$
the pair-potential $
\R^3 \ni y\to v_{ij}(y) \in \R$ splits into a sum $v_{ij}=v_{ij}^1+v_{ij}^2$ where
\begin{enumerate}[\quad\normalfont (1)]
\item \label{item:4ib}Differentiability: $v_{ij}^1\in C^{k_0+1}(\R^3)$ and $v_{ij}^2\in C^{k_0+1}(\R^3\setminus\{0\})$.
\item \label{item:5ib} Global bounds: For all  $\alpha$ with  $|\alpha|\leq k_0+1$ there are  bounds
  $|y|^{|\alpha|}\,|\partial^\alpha_yv_{ij}^1(y)|\leq C$.
\item \label{item:6ib} Decay at infinity: $|v_{ij}^1(y)| + |y\cdot\nabla_y
  v_{ij}^1(y)|=o(1)$.
\item \label{item:7ib}Local singularity: $v_{ij}^2$ is compactly supported
  and for all $\alpha$ with
  $|\alpha|\leq k_0+1$ there are  bounds
  $|y|^{|\alpha|+1}\,|\partial^\alpha_yv_{ij}^2(y)|\leq C$; $y\neq 0$.
\end{enumerate}
\end{conds}
In the above conditions, the letter $\alpha$ denotes multiindices.
Note that  \eqref{eq:57} and  \eqref{eq:58} with $v_{ij}(y)=q_iq_j|y|^{-1}$
conform with Condition \ref{cond:ac-stark-model23} for any $k_0$.

Introducing the inner product $x\cdot y=\sum_i2m_ix_i\cdot y_i$ for
$x=(x_1,\dots,x_N)$, $y=(y_1,\dots,y_N)\in \R^{3N}$ we can split
\begin{equation*}
\R^{3N}=X_{\rm CM}\oplus X;\; X_{\rm CM}=\big\{x\in\RR^{3N}\big| \,x_1=\cdots =x_N\big\}.
\end{equation*}

There is a  corresponding  splitting
$$
\tilde{h}(t) = h_{\text{CM}}(t)\otimes I + I\otimes h(t),
\qquad\text{on }L^2(X_{\rm CM})\otimes
L^2(X),
$$
where
$$
h_{\text{CM}}(t) = p^2_{\rm CM} - \cE_{\text{CM}}(t)\cdot x,\mand
h(t) = p^2 - \cE(t)\cdot x + V.
$$
Here
$$
\cE_{\text{CM}} =\frac{Q}{2M}\big(\tilde{\cE},\dots,\tilde{\cE} \big) \mand \cE =
\left(\Big(\frac{q_1}{2m_1}-\frac{Q}{2M}\Big)\tilde{\cE},\dots, \Big(\frac{q_N}{2m_N}-\frac{Q}{2M}\Big)\tilde{\cE}\right),
$$
where $Q=q_1+\cdots+q_N$ and $M=m_1+\cdots + m_N$ are the total
charge and mass of the system, respectively. In the special case where all the
particles have identical charge to mass ratio, we see that the
center of mass Hamiltonian is just an ordinary time-independent
$N$-body Hamiltonian. Otherwise the Hamiltonian $h(t)$ depends
non-trivially on the time-variable $t$. We denote by $\tilde U(t,s)$,
$U_{\text{CM}}(t,s)$ and $U(t,s)$ the dynamics generated by $\tilde
h(t)$, $
h_{\text{CM}}(t)$ and $
h(t)$, respectively, and observe that
$$
\tilde{U}(t,s) = U_{\text{CM}}(t,s)\otimes U(t,s).
$$

We shall address spectral properties of  the {\it monodromy
  operator} U(1,0). Note that this is a unitary operator on $L^2(X)$.
Let $\cA$ be  the set of all cluster partitions $a =
\{C_1,\dots, C_{\#a}\}$, $1\leq \#a\leq N$, each given by
splitting the set of particles $\{1,\dots, N\}$ into non-empty
disjoint clusters $C_i$. The spaces $X_a$, $a\in\cA$, are the
spaces of configurations of the $\#a$ centers of mass of the
clusters $C_i$ (in the center of mass frame). The complement
$$
X^a = X^{C_1}\oplus\cdots\oplus X^{C_{\#a}}
$$
is the space of relative configurations within each of the
clusters $C_i$. More precisely
$$
X^{C_i} = \big\{x\in X \big| x_j=0, j\notin C_i\big\}\mand
X_a = \big\{x\in X \big| k,l\in C_i \Rightarrow  x_k=x_l\big\}.
$$
We will write $x^a$ and $x_a$ for the orthogonal projection of a vector
$x$ onto the subspace $X^a$ and its orthogonal complement respectively.
Notice the natural ordering on $\cA$: $a\subset b$ if and only if any cluster $C\in a$ is contained
in some cluster $C^\prime\in b$.
 Clearly the minimal and maximal elements are $ a_{\rm min}= \{(1),\dots,(N)\}$ and
$ a_{\rm max}=\{(1,\dots,N)\}$, respectively. Any pair $(i,j)$ defines
an $N-1$ cluster decomposition $(ij)\in  \cA$ by letting $C=\{i,j\}$
constitute a cluster and all others being one-particle clusters.

 For each $a\neq a_{\rm max}$ the
  sub-Hamiltonian  monodromy operator is $U^a(1,0)$; it  is defined as the
  monodromy operator on $\cH^a=L^2(X^a)$ constructed for $a\neq
    a_{\rm min}$
  from $h^a=(p^a)^2- \cE(t)^a\cdot x^a +V^a$, $V^a=\sum_{(ij)\subset a}
  v_{ij}(x_i-x_j)$. If $a=
    a_{\mathrm{min}}$ we define $U^a(1,0)=\one$ (implying  $\sigma_{\mathrm{pp}}(U^{a_{\mathrm{min}}}(1,0)) = \{1\}$).
The condition
$\int_0^1\cE(t)\d t=0$ leads to the
existence of a unique $1$-periodic function $b$  such that
\begin{equation*}
  \tfrac{\d}{\d t}b(t)=\cE(t)\mand \int_0^1\,b(t)\d t=0.
\end{equation*}

The set of {\it thresholds} is
\begin{equation}
  \label{eq:52la}
  \cF(U(1,0)) = \bigcup_{a\neq  a_{\rm max}}
\e^{-\i \alpha_a}\sigma_{\rm pp}(U^a(1,0));\; \alpha_a=\int_0^1\,|b(t)_a|^2\d t.
\end{equation}

We recall from \cite{MS} that the set of thresholds  is closed and countable, and
non-threshold eigenvalues, i.e. points  in $\sigma_{\rm pp}(U(1,0))\backslash \cF(U(1,0))$, have finite multiplicity and can only accumulate at
the set of thresholds. Moreover any corresponding bound state is
exponentially decaying, the singular continuous spectrum $\sigma_{\rm
  sc}(U(1,0)) = \emptyset$ and there are integral propagation
estimates for states localized away from the set of eigenvalues  and away from
$\cF(U(1,0))$. These properties  are known under Condition
\ref{cond:ac-stark-model23}  with $k_0=1$. For completeness of
presentation we mention that some of the results of \cite{MS} hold
under more general conditions, in particular the exponential decay
result
does not require that the Coulomb singularity of each pair-potential
(if present) is located at the origin
(this applies to Born-Oppenheimer molecules in an AC-Stark field).

Letting
\begin{equation}\label{eq:59}
  A(t)=\tfrac12 \parb{x\cdot (p-b(t))+(p-b(t))\cdot x},
\end{equation} and using a different frame,
we prove in Section~\ref{AC--Stark model}

\begin{thm}\label{thm:ac-stark-type2dd} Suppose Conditions~\ref{cond:ac-stark-model23},
for some $k_0\in \N$. Let $\phi$ be a bound state for $U(1,0)$
pertaining to an eigenvalue  $\e^{-\i\lambda}\notin \cF(U(1,0))$. Then
  \begin{enumerate}[\quad\normalfont (1)]
  \item\label{item:8} $\phi\in \cD(A(1)^{k_0})$  where $A(t)$ is given by \eqref{eq:59}.
  \item\label{item:9} If for all pairs $(i,j)$ the term $v_{ij}^2=0$ then $\phi\in \cD(|p|^{k_0+1})$.
  \end{enumerate}
\end{thm}

The result (\ref{item:8}) is new for $k_0>1$ while it is
essentially contained in \cite{MS} for $k_0=1$, see \cite[Proposition~8.7~(ii)]{MS}.
We remark that the highest degree of smoothness known in general
in the case
$v_{ij}^2\neq 0$ is $\phi\in\cD (|p|)$, cf.
\cite[Theorem~1.8]{MS}. This  holds without the non-threshold
condition.  The result (\ref{item:9})
overlaps with \cite[Theorem~1.2]{KY}. This is for $N=2$ and
``$k_0=\infty$''.

\section{Assumptions and Statement of Regularity Results} \label{Sec-AbstractResults}

For a self-adjoint operator $A$ on a Hilbert space $\cH$, we will make use of
the $C^1(A)$ class of operators. This class consists a priori of bounded
operators $B$ with the property that
$[B,A]$ extends from a form on $\cD(A)$ to a bounded form on $\cH$.
The class is (consistently) extended to self-adjoint operators $H$,
by requiring that $(H-z)^{-1}$ is of class $C^1(A)$, for some (and hence all)
$z\in\rho(H)$, the resolvent set of $H$. We will use the notation
$H\in C^1(A)$ to indicate that an operator $H$ is of class $C^1(A)$.

If $H$ is of class $C^1(A)$ then $\cD(H)\cap\cD(A)$ is dense in $\cD(H)$
and the form $[H,A]$ extends by continuity from the form domain $\cD(H)\cap\cD(A)$
to a bounded form on $\cD(H)$. The extension is denoted by $[H,A]^0$, and is also interpreted
as an element of $\cB(\cD(H),\cD(H)^*)$.
If in addition $[H,A]^0$ extends by continuity to an element of $\cB(\cD(H),\cH)$,
then we say it is of class $C^1_\Mo(A)$. Note that being of class $C^1_\Mo(A)$
is equivalent to having the conditions of Mourre \cite{Mo} satisfied for the first commutator.
See \cite{GG}.

\begin{conds}\label{cond:condition} Let
  $\cH$ be a complex Hilbert
  space. Suppose  there are given some
  self-adjoint operators $H, A$ and $N$ as well as a symmetric
  operator $H'$ with $\cD(H')=\cD(N)$. Suppose
  $N\geq \one$. Let $R(\eta)=(A-\eta)^{-1}$ for $\eta\in \C\setminus\R$.
  \begin{enumerate}[\quad\normalfont (1)]
\item \label{item:2i0a} The operator $N$ is of class $C^1_\Mo(A)$. We abbreviate $N^\prime=\i [N,A]^0$.
\item \label{item:2} The operator $N$ is of class $C^1(H)$, and there exists
  $0<\kappa\leq \frac12$ such that the commutator obeys
\begin{equation}
  \label{eq:20}
 \i[N,H]^0\in \cB(N^{-\frac12+\kappa}\cH,N^{\frac12-\kappa}\cH).
\end{equation}
\item \label{item:2i0ue}There exists a (large) $\sigma> 0$ such that for all
  $\eta\in \C$ with $|\Im \eta|\geq \sigma$ we have as a form on $\cD(H)\cap \cD(N^{1/2})$
  \begin{equation}
    \label{eq:1commu}
  \i[H, R(\eta)]=-R(\eta)H'R(\eta).
  \end{equation} (Here it should be noticed that
  $N^{-1/2}H'N^{-1/2}$ and  $N^{\mp1/2}R(\eta)N^{\pm1/2}$ are
  bounded if $\sigma$ is large enough,
  cf. Remark~\ref{RemarkTo1}~1).)
\item \label{item:1} The
 commutator form $\i[H',A]$ defined on $\cD(A)\cap\cD(N)$  extends to a bounded operator
\begin{equation}
  \label{eq:5o}
 H'':=\i[H',A]^0\in \cB(N^{-\frac12}\cH,N^{\frac12}\cH).
\end{equation}
  \end{enumerate}
\end{conds}

\begin{cond} \label{cond:virial}
There are constants $C_1,C_2,C_3\in \R$ such that as a form on $\cD(H)\cap \cD(N^{1/2})$
\begin{equation}\label{eq:virial}
 N\leq C_1H+C_2H'+C_3\one.
\end{equation}
\end{cond}

\begin{cond}\label{cond:mourre} For a given
$\lambda\in \R$ there  exist
  $c_0>0$, $C_4\in \R$, $f_\lambda\in C_c^\infty(\R)$ with
$0\leq f_\lambda\leq 1$ and $f_\lambda=1$ in
  a neighborhood of  $\lambda$,
 and a compact operator $K_0$  on
  $\cH$ such that as  a form on  $\cD(H)\cap \cD(N^{1/2})$
\begin{equation} \label{eq:Mourre}
  H'\geq c_0\one - C_4f^{\bot}_\lambda(H)^2\la H\ra - K_0.
\end{equation}
Here $f^{\bot}_\lambda:=1-f_\lambda$.
\end{cond}

\begin{remarks}\label{RemarkTo1}
\begin{enumerate}[1)]
\item It follows from
Condition~\ref{cond:condition}~(\ref{item:2i0a}) and an argument of Mourre \cite[Proposition~II.3]{Mo}, that
there exists $\sigma>0$ such that for $|\Im\eta|\geq \sigma$ we have
$(A-\eta)^{-1}:\cD(N)\subseteq\cD(N)$ and $(A-\eta)^{-1}\cD(N)$ is dense in $\cD(N)$. 
By interpolation the same holds with
$N$ replaced by $N^\alpha$,  $0<\alpha<1$, cf. Lemma
\ref{Lemma-StrongLimits} below.
\item From Condition~\ref{cond:condition}~(\ref{item:2}) and Lemma \ref{HN1} it follows that
$N^{1/2}$ is of class $C^1_\Mo(H)$. In particular
$\cD(H)\cap \cD(N^{1/2})$  is dense in $\cD(N^{1/2})$.
\item Combining the above two remarks with Condition~\ref{cond:condition}~(\ref{item:2i0ue}) and \eqref{strongI3},
we find that given $H$, $A$ and $N$, there can at most be one $H'$ such that
Condition~\ref{cond:condition}~(\ref{item:2i0a}), (\ref{item:2}), and~(\ref{item:2i0ue})
are satisfied.
\item We remark that in practice we work with the weaker commutator estimate
\begin{equation}\label{Eq-WeakMourre}
 H'\geq c_0\one -\Re\{B(H-\lambda)\} - K_0,
\end{equation}
where $B=B(\lambda)$ is a bounded operator, with $B\cD(N^{1/2})\cup B^*\cD(N^{1/2})\subseteq\cD(N^{1/2})$.
The one in Condition~\ref{cond:mourre} is however more standard.
To see that Condition~\ref{cond:mourre} implies the above bound
choose $B = C_4 f_\lambda^\perp(H)^2\langle H \rangle(H-\lambda)^{-1}$ which under our Condition~\ref{cond:condition}
satisfies the requirements on $B$  by Lemma~\ref{lemma:HN}.
\end{enumerate}
\end{remarks}

 We call $H'$ the {\it first
derivative} of $H$. Similarly  $H''$ is the {\it second
derivative}  of $H$. The estimate (\ref{eq:virial}) is called the {\it
virial estimate}, while (\ref{eq:Mourre}) is the  {\it
Mourre estimate} at $\lambda$.

\begin{thm} \label{thm:mainresult} Suppose Conditions~\ref{cond:condition},~\ref{cond:virial}
and~\ref{cond:mourre}, and let $\psi$ be a bound state,
  $(H-\lambda)\psi=0$ (with $\lambda$ as in Condition~\ref{cond:mourre}), obeying
  \begin{equation}  \label{eq:34}
    \psi \in \cD(N^{\frac12}).
  \end{equation}
 Then $\psi\in \cD(A)$ and $A\psi\in \cD( N^{1/2})$.
\end{thm}

By imposing assumptions on higher-order commutators between $H$ and $A$ we obtain
a higher-order regularity result. For this we need
the following condition, which  coincides with   Condition~\ref{cond:condition}
(\ref{item:1}) if $k_0=1$, but  for $k_0\geq 2$ it is  stronger.

\begin{cond}\label{cond:conditioni}
  There exists $k_0\in \N$ such that the
 commutator forms  $\i^\ell\ad _A ^\ell(H')$ defined on $ \cD(A)\cap\cD(N)$, $\ell=0,\dots, k_0$,
 extend to  bounded operators
 \begin{align}
  \label{eq:5oi}
  &\i^\ell\ad_A^\ell(H')\in \cB(N^{-1}\cH,\cH);\;\ell=0,\dots,k_0-1.\\
 &\i^{k_0}\ad_A^{k_0}(H')\in \cB(N^{-\frac12}\cH,N^\frac12\cH). \label{eq:5oim}
 \end{align}
\end{cond}

We have the following extension of Theorem~\ref{thm:mainresult} to include higher orders

\begin{thm}
  \label{thm:mainresulti} Suppose Conditions~\ref{cond:condition}--\ref{cond:mourre} and
Condition~\ref{cond:conditioni}, and let $\psi$ be a bound state,
  $(H-\lambda)\psi=0$ (with $\lambda$ as in Condition~\ref{cond:mourre}), obeying \eqref{eq:34}.  Let $k_0$ be given as
  in Condition~\ref{cond:conditioni}. Then $\psi\in
  \cD(A^{k_0})$, and  for $k=1,\dots,k_0$ the  states $A^k\psi\in \cD(N^{1/2})$.
\end{thm}

It should be noted that under the assumptions imposed in Theorem~\ref{thm:mainresult} and Theorem~\ref{thm:mainresulti},
it is crucial that $N^{1/2}$ is applied \emph{after} the powers of $A$. The following result
requires an additional assumption, and allows for arbitrary placement of $N^{1/2}$ amongst the
at most $k_0$ powers of $A$. The new condition \eqref{eq:5oioo} below is a generalization
of Condition~\ref{cond:condition}~(\ref{item:2i0a}).

\begin{cond}\label{cond:vir21} Let $N'$ be given as in Condition~\ref{cond:condition}~(\ref{item:2i0a}).
There exists $k_0\in \N$ such that the
 commutator forms  $\i^\ell\ad _A ^\ell(N')$ defined on $ \cD(A)\cap\cD(N)$, $\ell=0,\dots, k_0-1$,
 extend to  bounded operators
 \begin{equation}
  \label{eq:5oioo}
  \i^\ell \ad_A^\ell(N') \in \cB(N^{-1}\cH,\cH);\;\ell=0,\dots,k_0-1.
 \end{equation}  Moreover there exists $\kappa_1>0$ such that the commutators (initially defined as  forms on $\cD(N)$)
\begin{equation}
  \label{eq:5oiooa}
  \i\,\ad_N \big( \i^\ell \ad_A^\ell (N')\big)  \in \cB(N^{-1}\cH,N^{1-\kappa_1}\cH );\;\ell=0,\dots,k_0-1.
\end{equation}
\end{cond}

We have

\begin{cor}
  \label{cor:assumpt-stat-regul}Suppose Conditions~\ref{cond:condition}--\ref{cond:mourre},
  \ref{cond:conditioni}~and~\ref{cond:vir21} (with the same $k_0$ in
  Conditions~\ref{cond:conditioni} and~\ref{cond:vir21}). Let $\psi\in\cD(N^{1/2})$ be a bound state,
  $(H-\lambda)\psi=0$ (with $\lambda$ as in Condition~\ref{cond:mourre}).
  For any $k,\ell\geq0$, with $k+\ell \leq  k_0$,
  we have   $\psi\in \cD(A^k N^{1/2} A^\ell)$.
\end{cor}

We end with the following improvement of Theorem~\ref{thm:mainresult}, which concludes in addition
that bound states are in the domain of $N$. It requires the added assumption \eqref{eq:5oiooa}, with $k_0=1$.

\begin{thm} \label{thm:mainresultk=1} Suppose Conditions~\ref{cond:condition}--\ref{cond:mourre}
and \eqref{eq:5oiooa} for $k_0=1$, and let $\psi\in\cD(N^{1/2})$ be a bound state
$(H-\lambda)\psi=0$ (with $\lambda$ as in Condition~\ref{cond:mourre}).  Then $\psi\in \cD(N)$,
the states $\psi, N^{1/2}\psi\in \cD(A)$ and  $A\psi\in \cD( N^{1/2})$.
\end{thm}

In Subsection~\ref{subsec-MoreNReg} we in fact prove an extension of the above theorem, to include
higher order estimates in $N$. These are applied in Section~\ref{AC--Stark model} to
many-body AC-Stark Hamiltonians.

\begin{remarks}\label{remark:assumpt-stat-regul}\begin{enumerate}[1)]
\item \label{item:1a0} The condition that $N\geq \one$ is imposed partly
  for convenience of formulation. Obviously one can obtain a
  version of the above results upon imposing only that $N$
  is bounded from below (upon ``translating'' $N\to N+C\geq \one$ at
  various points in the above conditions).
\item \label{item:1a} The `standard' or 'regular'  Mourre theory, considered for
  example in
  \cite{CGH}, fits in the semi-bounded case into the above scheme so
  that Theorem \ref{thm:mainresulti} holds. In fact (assuming here for
  simplicity that $H$ is
  bounded from below) we have  $N:=H+C\geq \one$ for a sufficiently large
  constant $C$. Use this $N$ and the same 'conjugate operator' $A$
  in Conditions \ref{cond:condition} -- \ref{cond:mourre}, \ref{cond:conditioni} and \ref{cond:vir21}. Note also
  that the standard Mourre
estimate at energy $\lambda$ reads
\begin{equation}\label{eq:Mourrestand}
  f_\lambda(H)\i[H,A]^0 f_\lambda(H)\geq c_0'f_\lambda^2(H)
  - K_0';\;c_0'>0, \;K_0'\textup{ compact}.
\end{equation} 
From \eqref{eq:Mourrestand} we readily conclude
\eqref{eq:Mourre} with $c_0= c_0'/2$,   $K_0=K_0'$ an a suitable
constant $C_4\geq 0$.

Although we shall not elaborate we also remark  that the method of
proof of Theorem \ref{thm:mainresulti} essentially can be adapted under the
conditions of  the standard
Mourre theory, in fact only  a simplified version is  needed.
 Whence although we can not literately conclude from Theorem
 \ref{thm:mainresulti} in the general non-semi-bounded case the result
 $\psi\in  \cD(A^{k_0})$
 is still valid given standard conditions on repeated  commutators
 $\i^k \ad_A^k(H)$ for  $k\leq {k_0+1}$.
\item \label{item:1b} Theorem~\ref{thm:mainresulti} does {\it not } hold  with one less commutator
 in Condition \ref{cond:conditioni}. Alternatively, under the conditions of Theorem
  \ref{thm:mainresulti} it is in general {\it false } that the bound
  state  $\psi\in\cD(A^{k_0+1})$. Based on considerations for discrete eigenvalues this
  statement may at a first thought appear surprising. See Example~\ref{RegExample}.
  Compared to \cite{CGH} our method works with one less commutator,
  cf. \ref{item:1a}),  although the
  overall scheme of ours and the one of  \cite{CGH} are similar.
\item \label{item:1bll} The proofs of Theorems
  \ref{thm:mainresult} and  \ref{thm:mainresulti}, Corollary
  \ref{cor:assumpt-stat-regul} and Theorem  \ref{thm:mainresultk=1} are constructive in that they yield explicit bounds. Precisely, if we
  have a positive lower bound of  the constant $c_0$ in
  \eqref{eq:Mourre} that is uniform in  $\lambda$ belonging to some
  fixed compact  interval $I$ as well as uniform
  bounds of the  absolute value of the constants
  $C_1,\dots,C_4$ of \eqref{eq:virial} and \eqref{eq:Mourre} (uniform
  in  the same sense) and similarly for all possible operator norms
  related to Conditions~\ref{cond:condition},~\ref{cond:conditioni}
   and~\ref{cond:vir21} (and the $B(\lambda)$ in Remark~\ref{RemarkTo1} if it is used) then
  there are bounds of the form, for example,
  \begin{equation*}
   \|N^{\frac12}A^k\psi\|\leq C \|N^{\frac12}\psi\|;\;C=C(k,I, K_0);
  \end{equation*} here $K_0=K_0(\lambda)$ is the compact operator of
  \eqref{eq:Mourre} and $k\leq k_0$.  Similar bounds are valid for the
  states $A^k N^{1/2}A^\ell\psi$ of Corollary
  \ref{cor:assumpt-stat-regul} and for the state   $N\psi$ of  Theorem
  \ref{thm:mainresultk=1}. In
  the context of perturbation theory  typically $I$ will be a small interval
  centered at some (unperturbed) embedded eigenvalue $\lambda_0$ and
  $K_0=K_0(\lambda_0)$. Whence the constant will depend only on the
  interval. For various models one can verify the condition
  \eqref{eq:34} for all bound states $\psi$ by a `virial argument',
 cf. \cite{GGM,MS,Sk}, along with a similar bound
 \begin{equation*}
   \|N^{\frac12}\psi\|\leq C (I)\|\psi\|.
  \end{equation*}
  This virial argument is  in a concrete situation
  related to the virial estimate (\ref{eq:virial}). Clearly the above  bounds can be used in
  combination, and this is precisely how we in Section~\ref{sec_models} arrive at the
Theorems~\ref{Main-Nelson} and~\ref{Main-NelsonPrime}.
In \cite{MW} the case of regular Mourre theory is considered where the derivation of
the bounds is simpler, and care is taken to derive good explicit bounds, which in particular are independent
of any proof technical constructions. The bounds are good enough to formulate a reasonable condition on
the growth of norms of multiple commutators which ensures that bound states are analytic vectors
with respect to $A$.
\end{enumerate}
\end{remarks}


\section{Preliminaries}\label{Preliminaries}


In this section we establish basic consequences of Conditions~\ref{cond:condition},
and introduce a calculus of almost analytic extensions taylored to avoid
issues with $(A-\eta)^{-1}$, when $|\Im\eta|$ is small.

\subsection{Improved Smoothness for Operators of Class $C^1(A)$}

For an operator $N$ of class $C^1(A)$ not much in the way of regularity
can be expected, beyond the $C^1(A)$ property itself, and its equivalent
formulations. See \cite{ABG,GGM1}. Often one requires some additional
smoothness properties to manipulate and estimate expressions in the two operators.
The typical way of achieving improved smoothness is to impose
conditions on $\i[N,A]^0$ stronger than what is implied by the $C^1(A)$ property itself.
This is what is done in Condition~\ref{cond:condition}~(\ref{item:2i0a}) and~(\ref{item:2}).

This subsection is devoted primarily to the extraction of improved smooth\-ness properties of the
pair of operators $N,H$, afforded to us by  Conditions~\ref{cond:condition}.

\begin{lemma}\label{HN0} Let $N\geq \one$ be of class $C^1(H)$
with
\[
[N,H]^0\in \cB(N^{-1/2}\cH,N^{1/2}\cH).
\]
For any $\alpha\in ]0,1[$, the operator $N^\alpha$ is of class $C^1(H)$.
\end{lemma}

\begin{proof} Let $0<\alpha<1$.
It suffices to check for one $\eta\in \rho(N^\alpha)$
that $(N^\alpha-\eta)^{-1}$ is of class $C^1(H)$. To this end we pick $\eta=0$,
and use the representation formula
\begin{equation}\label{eq:26}
N^{-\alpha} = c_\alpha\int_0^\infty t^{-\alpha} (N+t)^{-1}\d t, \ \ c_\alpha = \frac{\sin(\alpha \pi)}{\pi}.
\end{equation}
Since $N\in C^1(H)$ we  have for all
$t>0$ that  the operator $(N+t)^{-1}$ preserves $\cD(H)$. In fact
\begin{equation}
  \label{eq:29}
  [H,(N+t)^{-1}]\phi=(N+t)^{-1} [N,H]^0(N+t)^{-1}\phi;\;\phi\in \cD(H).
\end{equation}

By combining \eqref{eq:26} and \eqref{eq:29} we can compute
$[N^{-\alpha},H]$  considered as a form on $\cD(H)$ as
\begin{equation}
  \label{eq:30}
  [N^{-\alpha},H]=c_\alpha \int_0^\infty t^{-\alpha}( N
  +t)^{-1}[N,H]^0(N+t)^{-1}\d t.
\end{equation}
Notice that the integral is absolutely convergent for any $0<\alpha<1$.
This completes the proof.
\end{proof}

\begin{lemma}\label{HN1} Assume $N\geq \one$ and $H$ satisfy  Condition~\ref{cond:condition}~(\ref{item:2})
and let $\alpha\in ]0,1[$.
Then $N^\alpha\in C^1(H)$ and for $\tau_1,\tau_2\geq 0$,
with
\[
\max\{0,\tfrac12-\kappa-\tau_1\} + \max\{0,\tfrac12-\kappa-\tau_2\} < 1-\alpha,
\]
we have $[N^\alpha,H]^0\in \cB(N^{-\tau_1}\cH,N^{\tau_2}\cH)$. In particular
$N^{1/2}$ is of class $C^1_\Mo(H)$.
\end{lemma}

\begin{proof} That $N^\alpha\in C^1(H)$ follows from Lemma~\ref{HN0}.
We compute as a form on $\cD(N^\alpha)\cap\cD(H)$
\begin{equation}\label{eq:30ny}
[N^\alpha,H] = c_{\alpha} \int_0^\infty t^\alpha(N + t)^{-1}
[N,H]^0 (N+t)^{-1}\d t,
\end{equation}
where we have used the strongly convergent integral representation formula
\begin{equation}\label{Eq-Thomas}
N^\alpha = c_{\alpha} \int_0^\infty t^\alpha\big(t^{-1} - (N+t)^{-1}\big)\d t,
\end{equation}
which follows from (\ref{eq:26}).
We thus get for $\tau_1,\tau_2\geq 0$
\begin{eqnarray*}
|\langle\psi, [N^\alpha,H]\varphi\rangle|& \leq  & C \int_0^\infty t^{\alpha}
\|(N+t)^{-1} N^{\frac12-\kappa-\tau_1}\|\|(N+t)^{-1}
N^{\frac12-\kappa-\tau_2}\|\d t\\
 & &  \ \ \times \|N^{\tau_1}\psi\|
\|N^{\tau_2}\varphi\|.
\end{eqnarray*}
The integrand is of the order $O(t^{\alpha -2 + \theta})$,
with $\theta = \max\{0,\frac12-\kappa-\tau_1\}+\max\{0,\frac12-\kappa-\tau_2\}$.
It   is integrable provided $\theta < 1-\alpha$, which proves the lemma.
\end{proof}

We shall need a  boundedness  result:
\begin{lemma}
   \label{lemma:HN} Assume $N\geq \one$ and $H$ satisfy Condition~\ref{cond:condition}~(\ref{item:2}) and let $\alpha\in ]0,1/2+\kappa[$.
   Suppose  $f\in C^\infty (\R)$ is given such that
   \begin{equation*}
     {\frac{\d^k}{ \d t^k}}f(t)=O\big (\inp{t}^{-k}\big );\;k=0,1,\dots
   \end{equation*}
Then
\begin{equation}
  \label{eq:19}
  N^\alpha f(H) N^{-\alpha}\in\cB(\cH).
\end{equation}
\end{lemma}

\begin{proof} Let $\rho\in ]0,1/2+\kappa[$, where $0<\kappa\leq 1/2$ comes from
Condition~\ref{cond:condition}~(\ref{item:2}).
From Lemma~\ref{HN1} applied with $\tau_1=\max\{0,\rho-\kappa\}$ and $\tau_2 = 0$, we get
\begin{equation}\label{Nalphacomm}
[N^{\rho},H]^0\in \cB(N^{-\max\{0,\rho-\kappa\} }\cH,\cH).
\end{equation}

We recall from \cite[Proposition II.3]{Mo}  that if an operator $\tN$ is of class  $C^1_\Mo(H)$, then
   \begin{align}
    &\exists \sigma>0: \ |\Im \eta|\geq \sigma\Rightarrow
     (H-\eta)^{-1}\textup{ preserves }
     \cD(\tN) \mand\nonumber \\  & \tN (H-\eta)^{-1}\psi=(H-\eta)^{-1}\tN
     \psi\label{eq:32}\\
    &+\i(H-\eta)^{-1}\i[\tN, H]^0
     (H-\eta)^{-1}\psi\textup{ for all }\psi\in \cD(\tN).\nonumber
   \end{align}
We apply this to $\tN= N^\rho$, $0<\rho<1/2+\kappa$.
The assumption is satisfied by (\ref{Nalphacomm}).

We shall show a  representation
formula for the special  case $f(x) =f_\eta(x) =(x-\eta)^{-1}$
with $v=\Im \eta\neq 0$. Now fix $\alpha\in ]0,1/2+\kappa[$.
Using (\ref{Nalphacomm}) and (\ref{eq:32}), multiple times with $\rho=\alpha-j\kappa$, we
obtain for $|\Im \eta|$ sufficiently large and for all $\psi\in \cD(N^{\alpha})$
\begin{align}
 \nonumber  &N^{\alpha}(H-\eta)^{-1}\psi-(H-\eta)^{-1}N^{\alpha}\psi\label{eq:33}\\
  &=\sum^{n}_{j=1} \big((H-\eta)^{-1}B_1\big)\cdots \big( (H-\eta)^{-1}B_j\big)(H-\eta)^{-1}N^{\alpha-j\kappa}\psi\\
 \nonumber &  \ \ +  \big((H-\eta)^{-1}B_1\big)\cdots \big((H-\eta)^{-1}B_n\big)\big((H-\eta)^{-1} B_{n+1}\big)(H-\eta)^{-1}\psi,
\end{align}
where $n$ is the biggest natural number  for which
$\alpha-n\kappa>0$ and the $B_j$'s are  bounded and independent
of $\eta$. Next by analytic continuation we conclude that
 (\ref{eq:33}) is valid for all $\eta\in \C\setminus
\R$. Hence we have verified the adjoint version of (\ref{eq:19})  for
$f =f_\eta$; $v\neq 0$.

 We shall now show \eqref{eq:19} in general. Define a new function by
$h(t) = f(t)(t+\i)^{-1}$, and let $\tilde{h}$ denote an almost analytic extension of $h$
such that (using the notation $\eta = u+ \i v$)
\begin{equation*}
\forall n\in \N:\ |\bar{\partial}\tilde{h}(\eta)|  \leq C_n \inp{\eta}^{-n-2} |v|^n.
\end{equation*}
We shall use the representation
\begin{align}\label{eq:13}
\nonumber  f(H) & = \frac1{\pi}\int_{\C}
(\dbar \tilde{h})(\eta)(H-\eta)^{-1}(H+\i)\d u\, \d v\\
& = \frac1{\pi}\int_{\C}
(\dbar \tilde{h})(\eta)\big(\one +(\eta+\i)(H-\eta)^{-1}\big)\d u\, \d v,
\end{align}
which should be read as a strong integral on $\cD(H)$.
We multiply by $N^{\alpha}$ and  $N^{-\alpha}$
from the left and from the right, respectively. Inserting
(\ref{eq:33}) we conclude the lemma. Observe that $N^{-\alpha}$ being $C^1(H)$ preserves $\cD(H)$.
\end{proof}

It will be important  to work with the following `regularization'
operators, cf. \cite{Mo}: Let  for any given self-adjoint operator
$\tilde A$ and  any positive operator $\tilde N$
\begin{equation}
  \label{eq:5}
  I_n(\tilde A)=-\i n (\tilde A-\i n)^{-1}\mand I_{\i n}(\tilde N)= n
  (\tilde N+n)^{-1};\;n\in \N.
\end{equation} In particular we shall use $I_n(A)$ in conjunction with
\eqref{eq:1commu},  $I_n(H)$ in conjunction with
(\ref{eq:1commu}), (\ref{eq:virial}) and (\ref{eq:Mourre}), while
$I_{\i n}(N)$ will be used in
conjunction with (\ref{eq:20}).

\begin{lemma}\label{Lemma-StrongLimits} Assume  the pairs $N,A$ and $N,H$ satisfy Conditions~\ref{cond:condition}~(\ref{item:2i0a}) and~(\ref{item:2})
respectively. Then
\begin{align}\label{strongI1}
& \slim_{n\to\infty} N^{\frac12} I_n(H) N^{-\frac12}  = \one
\\ \label{strongI2}
& \slim_{n\to\infty} N I_n(A)N^{-1} = \one
\\ \label{strongI3}
& \slim_{n\to\infty} N^\frac12 I_n(A)N^{-\frac12} = \one.
\end{align}
\end{lemma}

\begin{proof} Observe first that $\slim I_n(A) = \one$ and $\slim A(A-\i n)^{-1}=0$,
and similarly with $A$ replaced by $H$.
The statements \eqref{strongI1} and \eqref{strongI2} now follows from (\ref{eq:32})
and boundedness of the operators
$[N^{1/2},H]^0 N^{-1/2}$ and $[N,A]^0 N^{-1}$. This argument appears also in \cite{Mo}.

As for \eqref{strongI3} we observe first that $N(I_n(A)-\one)N^{-1}$ is bounded
uniformly in $n$. By interpolation the same holds true for
$N^{1/2}(I_n(A)-\one)N^{-1/2}$. The result now follows from observing that the result holds true
strongly on the dense set $\cD(N^{1/2})$ by \eqref{strongI2}.
\end{proof}

We end with a small technical remark

\begin{remark}\label{IntersectionDense} Suppose $N$ and $H$ are as in Lemma~\ref{HN0} and $0\leq \alpha<1$.
Then $\cD(H)\cap \cD(N)$ is dense in $\cD(H)\cap \cD(N^\alpha)$ in the intersection topology.

To see this let $\psi\in\cD(H)\cap \cD(N^\alpha)$.
Then $\psi_n= I_{\i n}(N)\psi\in\cD(H)\cap \cD(N)$ since $N$ is of class $C^1(H)$.
We claim that $\psi_n\to \psi$ in $\cD(H)\cap \cD(N^\alpha)$.
Obviously $N^\alpha\psi_n\to N^\alpha\psi$, so it remains
to consider
\[
H\psi_n = I_{\i n}(N) H\psi + \sqrt{\frac{N}{n}}I_{\i n}(N)\big(N^{-\frac12}[N,H]^0N^{-\frac12}\big) \sqrt{\frac{N}{n}} I_{\i n}(N) \psi.
\]
As in the proof above, the last term goes to zero and the first term converges to $H\psi$ proving the claim.
\end{remark}

\subsection{Iterated commutators with $N^{1/2}$}

We address here the following question. Supposing Condition~\ref{cond:condition}~\eqref{item:2i0a}
and \eqref{eq:5oioo} is satisfied for some $k_0\geq 1$. One could reasonably assume that
$N^{1/2}$ is also of class $C^1_\mo(A)$ and admits $k_0$ iterated
$N^{1/2}$-bounded commutators. We have however not been able to establish this,
but making the additional assumption \eqref{eq:5oiooa} we answer the question in the affirmative
below. This permits us to deduce Corollary~\ref{cor:assumpt-stat-regul} from Theorem~\ref{thm:mainresulti}.
The reader primarily interested in Theorem~\ref{thm:mainresulti} may skip this subsection.

We begin with a technical lemma. Let $q\in\NN$ and $\uell\in (\NN\cap\{0\})^q$, with
$0\leq \ell_j < k_0$ for all $j=1,\dots,q$. We abbreviate
$N^\prime_m = \i^{m}\ad_A^{m}(N^\prime)$, which is the iteratively defined $N$-bounded
operator from \eqref{eq:5oioo}. Let for $t\geq 0$ and $q, \uell$ as above
\begin{equation}\label{Bqldef}4
B_q^\uell(t) = t^\frac12\Big(\prod_{j=1}^q (N+t)^{-1} N^\prime_{\ell_j}\Big)(N+t)^{-1}.
\end{equation}
Observe that $B_q^\uell(t)$ is bounded for all $t$. Indeed it satisfies the bound
$B_q^\uell(t) = O(t^{-1/2})$ and is thus not norm integrable.
However if $\varphi\in \cD(N)$ we have $B_q^\uell(t)\varphi = O(t^{-3/2})$.
The extra assumption \eqref{eq:5oiooa} allows us to prove

\begin{lemma}\label{Lemma-N12Int} Suppose Condition~\ref{cond:condition}~\eqref{item:2i0a} and Condition~\ref{cond:vir21}.
For any $q\in \NN$, $\uell\in(\NN\cup\{0\})^q$ (with $0\leq \ell_j <k_0$ as above) and $\varphi\in\cD(N)$ the map
$t\to B_q^\uell(t)\varphi$ is integrable  and
there exist constants $C_q^\uell$ such that
\[
\Big\|\int_0^\infty B_q^\uell(t)\varphi\, \d t \Big\|\leq C_q^\uell \|N^{\frac12}\varphi\|.
\]
\end{lemma}

\begin{proof} 
We only have to prove the bound on the strong integral, since we already discussed strong integrability.
We begin by analyzing the leftmost factors in $B_q^\uell(t)$,
namely the $N$-bounded operator $(N+t)^{-1}N_{\ell_1}^\prime$.

We compute strongly on $\cD(N)$
\begin{align}\label{LeftFactor}
\nonumber (N+t)^{-1}N_{\ell_1}^\prime & = \big(N_{\ell_ 1}^\prime N^{-1}\big) N (N+t)^{-1}\\
\nonumber & \ \ \ -(N+t)^{-1}N\big(N^{-1}[N,N_{\ell_1}^\prime]N^{-1+\kappa_1}\big) N^{1-\kappa_1}(N+t)^{-1}\\
&  = \big(N_{\ell_1}^\prime N^{-1}\big) N (N+t)^{-1} + O(t^{-\kappa_1}).
\end{align}
The contribution to the integral $\int_0^\infty B_q^\uell(t)N^{-1/2}\d t$ coming from  the last term is $O(t^{-1-\kappa_1})$
and hence norm-integrable.

If $q=1$ we can now finish the argument because the contribution to the integral coming from the
first term on the right-hand side of  \eqref{LeftFactor} is
\[
\big(N_{\ell_1}^\prime N^{-1}\big) t^\frac12(N+t)^{-2} N,
\]
which on the domain of $N$ integrates to the $N^{1/2}$-bounded operator $ c N_{\ell_1}^\prime N^{-1/2}$, for some $c\in\RR$.

If $q>1$ we write $N(N+t)^{-1} = \one-t(N+t)^{-1}$. We can now bring out the
next term $N_{\ell_2}^\prime$, and again the commutators with $(N+t)^{-1}$ give norm-integrable contributions. Repeating this procedure
successively until all the terms $N_{\ell_ j}^\prime$ are brought out to the left
yields the formula
\[
B_q^\uell(t) = \Big(\prod_{j=1}^q N_{\ell_j}^\prime N^{-1}\Big) t^\frac12(\one-t(N+t)^{-1})^{q-1}(N+t)^{-2} N + O(t^{-1-\kappa_1})N^\frac12.
\]
We compute, by a  change of variables,
\[
\int_0^\infty t^\frac12(\one-t(N+t)^{-1})^{q-1}(N+t)^{-2}\, \d t = c^\prime N^{-\frac12},
\]
for some $c^\prime\in \RR$. This implies the lemma.
\end{proof}

\begin{prop}\label{Prop-CommN12} Assume Condition~\ref{cond:condition}~\eqref{item:2i0a} and Condition~\ref{cond:vir21}. 
Then $N^{1/2}$ is of class $C^1_\mo(A)$ and
the iterated commutators $\i^p \ad_A^p (N^{1/2})$, $p\leq k_0$,  extends from $\cD(A)\cap \cD(N^{1/2})$ to $N^{1/2}$-bounded operators.
\end{prop}

\begin{proof}
 We already know from Lemma~\ref{HN0} that $N^{1/2}$ is of class $C^1(A)$. Hence we only need
to establish that the iterated commutator forms extend to $N^{1/2}$-bounded operators.
Recall also that $\cD(A)\cap\cD(N)$ is dense in $\cD(A)\cap\cD(N^{1/2})$, cf. Remark~\ref{IntersectionDense}, 
which implies that it suffices to show that the iterated
commutator forms extend from $\cD(A)\cap\cD(N)$ to $N^{1/2}$-bounded operators.

By Lemma~\ref{Lemma-N12Int} and the above remark it suffices to prove, iteratively,  the following representation formula
\begin{equation}\label{N12commrep}
\i^p\ad_A^p(N^{1/2})\varphi = \sum_{q=1}^p \sum_{\ell_1+\cdots +\ell_q = p-q}\alpha^{p,q}_{\uell}\int_0^\infty B_{q}^\uell(t)\varphi\, \d t,
\end{equation}
for $\varphi\in \cD(N)$. Note that the integrals are absolutely convergent. Here $B_q^\uell(t)$ are defined in \eqref{Bqldef}.

For $p=1$ we compute using \eqref{Eq-Thomas}
\[
\i [A_n,N^{\frac12}]\varphi = c_\frac12 \int_0^\infty B_{1,n}^{\underline{0}}(t)\varphi\,\d t,
\]
where the extra subscript $n$ indicates that $N_{0}^\prime = N^\prime$ has been replaced by $I_n(A) N^\prime I_n(A)$.
By \eqref{eq:10} the integrand is $O(t^{-3/2})$ uniformly in large $n$,
and by \eqref{strongI2} and  Lebesgue's  theorem on dominated convergence we can thus compute
\[
\lim_{n\to\infty}\i[A_n,N^{\frac12}]\varphi = c_\frac12 \int_0^\infty B_{1}^{\underline{0}}(t)\varphi\,\d t.
\]
Obviously this together with Lemma~\ref{Lemma-N12Int} implies that the form $\i\, \ad_A(N^{1/2})$ extends from $\cD(A)\cap \cD(N)$ to
an $N^{1/2}$-bounded operator represented on $\cD(N)$ by the strongly convergent integral above.

We can now proceed by induction, assuming that the iterated commutator  $\i^{p-1}\ad_A^{p-1}(N^{1/2})$ exists as
an $N^{1/2}$-bounded operator and is represented on $\cD(N)$  by \eqref{N12commrep}.
Compute first the commutator $\i[A_n,\i^{p-1}\ad_A^{p-1}(N^{1/2})]$ strongly on $\cD(N)$ using that
\[
\i[A_n,N_\ell^\prime] = -I_n(A)N_{\ell+1}^\prime I_n(A) \mand \i[A_n,(N+t)^{-1}] = (N+t)^{-1}N^\prime(N+t)^{-1}.
\]
Subsequently take the limit $n\to\infty$ as above and appeal to Lemma~\ref{Lemma-N12Int} to conclude that
the so computed limit in fact is an $N^{1/2}$-bounded extension of the form $\i[A,\i^{p-1}\ad_A^{p-1}(N^{1/2})]$
from  $\cD(A)\cap \cD(N)$ and represented on $\cD(N)$ as in \eqref{N12commrep}.
\end{proof}

\begin{proof}[Proof of Corollary~\ref{cor:assumpt-stat-regul}:]
We can now argue that Corollary~\ref{cor:assumpt-stat-regul} is
indeed a direct corollary of Theorem~\ref{thm:mainresulti}.

Note that $\psi\in \cD(N^{1/2}A^{k})$ for all $k\leq k_0$
due to  Theorem \ref{thm:mainresulti}.
We can now repeatedly use the fact that $\cD(A)\cap\cD(N^{1/2})$ is dense in $\cD(A)$  and 
Proposition~\ref{Prop-CommN12} to compute for $\varphi\in\cD(A^p)$, with $p+k\leq k_0$,
\[
\big\la A^p\varphi,N^\frac12 A^k \psi\big\ra = 
\sum_{q=0}^p \beta_q\big\la\varphi,\big(\ad_A^{p-q}(N^{\frac12})N^{-\frac12}\big) N^\frac12A^{q+k}\psi\big\ra,
\]
with $\beta_q$ some real combinatorial factors. This completes the proof since the norm of the right-hand side
is bounded by $C\|\varphi\|$.
\end{proof}


\subsection{Approximating $A$ by Regular Bounded Operators}


We recall now a  construction from \cite{MS} (see \cite[p. 203]{MS}).
Consider an odd real-valued function $g\in C^\infty(\R)$ obeying
$g^\prime \geq 0$,  that the function $\R\ni t\to tg^\prime(t)/g(t)$
has a smooth square root, that the function $]0,\infty[\ni t \to g(t)$ is concave   and the properties
\begin{equation*}g(t) =
\begin{cases}
2\qquad & \text{for } t>3\\
t\qquad & \text{for } |t|<1\\
-2\qquad & \text{for } t< -3
\end{cases}\;.
 \end{equation*}

Let $h(t) = g(t)/t$. We pick an almost analytic extension of $h$,
denoted by $\tilde{h}$, such that for some $\rho>0$ (and  using
again the notation $\eta = u+ \i v$)
\begin{align}\label{hbound}
  &\forall N:\ |\bar{\partial}\tilde{h}(\eta)|  \leq C_N
\inp{\eta}^{-N-2} |v|^N,\\
\nonumber &\tilde{h}(\eta) = \begin{cases}2/\eta \qquad
&\text{for } u >6,\ |v|< \rho(u -6)\\
\nonumber -2/\eta \qquad &\text{for } u <-6,\ |v|< \rho(6-u)\end{cases}\;.
\end{align}
We can  choose $\tilde{h}$ such that
$\overline{\tilde{h}(\eta)} = \tilde{h}(\bar{\eta})$.

This gives
the representation
\begin{equation}
  \label{eq:15}
  g(t) = \frac1{\pi}\int_{\C}(\dbar\tilde{h})(\eta)t(t-\eta)^{-1}\d u\,\d v.
\end{equation}

Let $g_m(t) = m g(t/m)$, for $m\geq 1$. Using  the properties of $g$ one
 verifies that for all $t\in \R$ the function
\begin{equation}
  \label{eq:9}
 m\to g_m(t)^2\text { is increasing}.
\end{equation}

We recall that there
exists $\sigma>0$ such that for $|v|\geq\sigma/m$ the operator
\begin{equation}\label{resA}
R_m(\eta):=\left(\frac{A}{m}-\eta\right)^{-1}
\end{equation}
preserves $\cD(N)$. See \eqref{eq:32}. Moreover we have uniformly in
$\alpha\in [0,1]$, $m\in\NN$ and
$\eta$ that
\begin{equation}\label{eq:10}
  \|N^\alpha R_m(\eta)N^{-\alpha}\|
\leq C |v|^{-1};\; \eta\in V_m^>,
\end{equation}
where
\[
V_m^> := \{u+\i v\in \C \, : \,   |v|\geq \sigma/m\} \ \ \mathrm{and}
\ \ V_m^< :=\{u+\i v\in\C \,  :\,  |v|<\sigma/m\}.
\]
 This
motivates the decomposition into smooth bounded real-valued
functions $g_m = g_{1m} + g_{2m}$, where
\begin{align}\label{eq:rep1}
g_{1m}(t) &= \frac{m}{\pi}\int_{V_m^>}
(\dbar\tilde{h})(\eta)\left( 1 +
\eta\left(\frac{t}{m}-\eta\right)^{-1}\right)\d u\, \d v+C_m,\\
 g_{2m}(t) &= \frac{m}{\pi}\int_{V_m^<}
(\dbar\tilde{h})(\eta)\eta\left(\frac{t}{m}-\eta\right)^{-1}\d u\,\d v;\label{eq:rep2}\\
C_m&=\frac m{\pi}\int_{V_m^<}\dbar\tilde{h}(\eta)\, \d u\,\d v.\nonumber
\end{align}

Note that the integral in the expression
for $g_{2m}$ is over a compact set (decreasing with $m$). This implies
the property
\begin{equation}
  \label{eq:21}
  \sup_{m\in\NN,t\in\R}m^n\inp {t}^{k+1}|g^{(k)}_{2m}(t)| \leq C_{n,k} <\infty\mfor
n,k\in \N \cup\{0\}.
\end{equation}
Since $g_m$ and $g_{2m}$ are bounded functions, we conclude the same for $g_{1m}$.

At a key point in the proof we will need a smooth square root
of the function $tg'g$. We pick
\begin{equation}\label{ghat}
\hat{g} = p g \in C_0^\infty(\R),
\end{equation}
where $p(t) = \sqrt{tg'(t)/g(t)}$, which was assumed smooth.
Clearly $\hat{g}^2 = tg'g$.
Let $\tilde{p}\in C_0^\infty(\C)$ be an almost analytic extension
of $p$. It satisfies
\begin{equation}
\forall N:\ |\bar{\partial}\tilde{p}(\eta)|  \leq C_N |v|^N.
\end{equation}
 As above we put $p_m(t) = p(t/m)$ and make the splitting $p_m=p_{1m}+p_{2m}$, where
\begin{align}\label{eq:rep1o}
p_{1m}(t) &= \frac{1}{\pi}\int_{V_m^>}
(\dbar\tilde{p})(\eta)
\left(\frac{t}{m}-\eta\right)^{-1}\d u\, \d v,\\
 p_{2m}(t) &= \frac{1}{\pi}\int_{V_m^<}
(\dbar\tilde{p})(\eta)\left(\frac{t}{m}-\eta\right)^{-1}\d u\, \d v\label{eq:rep20}.
\end{align}
Let $\hat{g}_m = p_m g_m$
and split $\hat{g}_m = \hat{g}_{1m}+\hat{g}_{2m}$ by
\begin{equation}
\hat{g}_{1m} = p_{1m}g_{1m} \ \ \mathrm{and} \ \ \hat{g}_{2m} = p_m g_{2m} + p_{2m}g_{1m}.
\end{equation}
Clearly we can choose $C_{n,k}$ in \eqref{eq:21} possibly larger such that $\hat{g}_{2m}$ satisfies the same estimates. Since $p_m$ and $p_{2m}$
are uniformly bounded in $m$ we get
\begin{equation}\label{punif}
P := \sup_{m\in\N}\sup_{t\in \R}|p_{1m}(t)| < \infty.
\end{equation}

We observe that the operators $g_{2m}(A)$ and $p_{1m}(A)$, $p_{2m}(A)$ are
given by norm convergent integrals, whereas $g_m(A)$ and $g_{1m}(A)$
are given on the domain of  $\langle A\rangle^{s}$, for any $s>0$, as
strongly convergent integrals.

From (\ref{eq:9}) and Lebesgue's theorem on monotone
convergence, we observe that $\psi\in\cD(A^k)$ is equivalent to
$\sup_{m}\|g_m(A)^k\psi\| <\infty$.
Combining this with (\ref{eq:21}) we find that for $k\geq 1$
\begin{equation}\label{domainchar}
\psi \in \cD(A^k) \  \Leftrightarrow \ \psi\in \cD(A^{k-1}) \ \ \mathrm{and} \ \ \sup_{m}\|g_{1m}(A)^k\psi\| <\infty.
\end{equation}

It will be convenient in the following when dealing with $g_{1m}$ to
abbreviate
\[
\d\lambda(\eta) = \frac1{\pi} (\bar{\partial}\tilde{h})(\eta) \, \d u\, \d v.
\]
This is however not a complex measure, just a notation. Similarly we will on one occasion
write $\d\lambda_p(\eta) = \frac1{\pi}(\bar{\partial}\tilde{p})(\eta)\d u \d v$,
which is in fact a complex measure.

We have the following

\begin{lemma}\label{dominv} As a result of the above constructions we have for any $m\geq 1$
and $0\leq \alpha\leq 1$
that the bounded operators $g_{1m}(A)$, $g_{1m}'(A)$, $p_{1m}(A)$ and $Ag_{1m}'(A)$
preserve $\mathcal{D}(N^\alpha)$.
\end{lemma}

\begin{proof}
Let $\psi\in\mathcal{D}(N)$ and $\varphi\in\mathcal{D}(A)$.
 Observe that $N^{-1}\varphi\in\mathcal{D}(A)$,
by the $C^1(A)$ property of $N$, cf. Condition~\ref{cond:condition}~(\ref{item:2i0a}).
 We can thus compute using the strongly convergent
integral representation for $g_{1m}(A)$, and the notation introduced in (\ref{resA}),
\begin{eqnarray}\label{inlem24}
\lefteqn{\langle N\psi, g_{1m}(A) N^{-1}\varphi\rangle}\\
\nonumber & = & m\int_{V_m^>}\big\langle N\psi,\left( 1 +
\eta R_m(\eta)\right)N^{-1}\varphi\big\rangle \d\lambda(\eta) + C_m\langle\psi,\varphi\rangle\\
\nonumber& = & \langle \psi,g_{1m}(A)\varphi\rangle + \i  \int_{V_m^>}
\eta \big\langle \psi, R_m(\eta) N' R_m(\eta) N^{-1}\varphi\big\rangle\, \d\lambda(\eta).
\end{eqnarray}
By Condition~\ref{cond:condition}~(\ref{item:2i0a}), (\ref{hbound}) and~(\ref{eq:10})
 we find that for some constant $K_m$ we have
\begin{equation}\label{inlem24a}
|\langle N\psi, g_{1m}(A) N^{-1}\varphi\rangle| \leq K_m\|\psi\|\|\varphi\|.
\end{equation}
This together with an interpolation argument concludes the proof.

The cases $g_{1m}'(A)$ and  $p_{1m}(A)$ are done the same way. As for $Ag_{1m}'(A)$ we write
$A_j = A I_j(A)$ and compute
\[
NA_jg_{1m}'(A)N^{-1} = A_j N g_{1m}'(A) N^{-1}- \i I_j(A) N'N^{-1} N I_j(A) g_{1m}'(A)N^{-1}.
\]
To complete the proof by taking $j\to\infty$ we need to argue
that
\[
N g_{1m}'(A)\cD(N)\subseteq \cD(A).
\]
 To achieve this we repeat the computation
(\ref{inlem24}), with $\psi$ replaced by $A\psi$, $\psi\in \cD(A)$, and $g_{1m}$ replaced
by $g_{1m}'$. We get
\begin{eqnarray*}
\lefteqn{\langle A\psi, N g_{1m}'(A) N^{-1}\varphi\rangle
= \langle \psi, A g_{1m}'(A)\varphi\rangle}\\
& &  +\int_{V_m^>}\eta
\Big\langle \frac{A}{m}\psi, \big\{R_m(\eta)N'R_m(\eta)^2
+R_m(\eta)^2 N' R_m(\eta)\big\}N^{-1}\varphi\Big\rangle\, \d\lambda(\eta).
\end{eqnarray*}
The result now follows from writing
$\frac{A}{m}R_m(\eta) = \one+\eta R_m(\eta)$and appealing to
(\ref{hbound}) and (\ref{eq:10}) as above.
\end{proof}


\section{Proof of the Abstract Results}\label{Proof of Theorem}


In this section we prove the abstract theorems formulated in
Section~\ref{Sec-AbstractResults} as well as an extended version of
Theorem \ref{thm:mainresultk=1}.
The proofs are given in  separate subsections.


\subsection{Proof of Theorem \ref{thm:mainresulti}}


Let
\[
\cD_k = \{\varphi\in\cD(A^{k})| \forall 0\leq j \leq k: A^j\varphi\in \cD(N^\frac12)\}.
\]
Using Conditions \ref{cond:condition} -- \ref{cond:mourre} and 
  \ref{cond:conditioni} we shall prove Theorem~\ref{thm:mainresulti} by
induction in $ k=0,\dots,k_0$ that $\psi\in\cD_k$. We can assume without
loss of generality that $\lambda=0$.

The proof relies on three estimates which we state first in the form of
three propositions. After giving the proof of Theorem~\ref{thm:mainresulti},
we then proceed to verify the propositions.

We begin with some abbreviations and a definition.
For a state $\psi$ we introduce the notation
\[
\psi_m = g_{1m}(A)^k\psi, \ \ \mathrm{and} \ \
\hat{\psi}_m = \hat{g}_{1m}(A)g_{1m}(A)^{k-1}\psi=p_{1m}(A)\psi_m.
\]
Let $\sigma>0$ be fixed as in Remark~\ref{RemarkTo1}~1), applied with $N^{1/2}$
in place of $N$.

\begin{defn}\label{Def-k-remainder} Let $k\geq 1$.
A family of forms $\{R_m\}_{m=1}^\infty$ on $\cD_{k-1}$ will be called a
$k$-remainder if for all $\epsilon>0$
there exists $C_\epsilon>0$ such that
\begin{equation}\label{kremainder}
|\langle\psi,R_m\psi\rangle| \leq
\epsilon \| N^\frac12 \psi_m\|^2
+ C_\epsilon\| N^\frac12 (A -\i\sigma)^{k-1}\psi\|^2,
\end{equation}
for any $\psi\in\cD_{k-1}$ and $m\in\N$.
\end{defn}

Lemma~\ref{dominv} is repeatedly used below, mostly without comment, to justify
manipulations.
The first proposition is a virial result, to be proved by  a symmetrization
of a commutator between $H$ and a regularized version of $A^{2k+1}$.

\begin{prop}\label{input1} Let $0<k\leq k_0$ and
$\psi\in\cD_{k-1}$ be a bound state for $H$.
There exists a $k$-remainder $R_m$, such that
\[
\langle\psi_m,H'\psi_m\rangle +2k\langle \hat{\psi}_m,H'\hat{\psi}_m\rangle
=\langle\psi,R_m\psi\rangle.
\]
\end{prop}

The second result is an implementation of the virial bound (\ref{eq:virial}) in
Condition~\ref{cond:virial}, which together
with Proposition~\ref{input1} makes it possible to deal with
$N^{1/2}\psi_m$. This is reminiscent of what was done in the proof of \cite[Proposition 8.2]{MS}.
The constant $C_2$ appearing in the proposition comes from Condition~\ref{cond:virial}.

\begin{prop}\label{input2} Let $\psi\in\cD_{k-1}$ be a bound state. There exists $C$
independent of $m$ such that
\[
\|N^\frac12 \psi_m\|^2 \leq 2 C_2 \langle \psi_m, H' \psi_m\rangle
+ C\big(\|\psi_m\|^2 + \|N^\frac12(A-\i\sigma)^{k-1}\psi\|^2\big)
\]
and
\[
\|N^\frac12 \hat{\psi}_m\|^2 \leq
2 C_2 \langle \hat{\psi}_m,H' \hat{\psi}_m\rangle
+ C\big(\|\hat{\psi}_m\|^2 + \|N^\frac12(A-\i\sigma)^{k-1}\psi\|^2\big).
\]
\end{prop}

The third and final input is an implementation of the positive
commutator estimate in Condition~\ref{cond:mourre}.
The constant $c_0$ and the compact operator $K_0$ appearing in the
proposition come from Condition~\ref{cond:mourre}.

\begin{prop}\label{input3}
Let $\psi\in\cD_{k-1}$ be a bound state. There exist constants  $C,\tC>0$
independent of $m$ such that
\[
\langle \psi_m, H' \psi_m\rangle \geq \frac{c_0}2\|\psi_m\|^2
-\tC \langle\psi_m,K_0\psi_m\rangle -
C\|N^\frac12(A-\i\sigma)^{k-1}\psi\|^2
\]
and
\[
\langle \hat{\psi}_m, H' \hat{\psi}_m\rangle
\geq \frac{c_0}2\|\hat{\psi}_m\|^2
- \tC \langle \hat{\psi}_m,K_0\hat{\psi}_m\rangle
- C\|N^\frac12(A-\i\sigma)^{k-1}\psi\|^2.
\]
\end{prop}

\noindent\emph{Proof of Theorem~\ref{thm:mainresulti}:}
Let $\psi$ be the bound state, which we take to be normalized.
By assumption $\psi\in\cD_0$. Assume by induction that $\psi\in\cD_{k-1}$,
for some $k\leq k_0$. We proceed to show that $\psi\in\cD_k$:

From Proposition~\ref{input1} we get the existence of a
$k$-remainder $R_m$ such that
\[
\langle \psi_m H' \psi_m\rangle + 2k\langle \hat{\psi}_m,H'\hat{\psi}_m\rangle
=\langle\psi,R_m\psi\rangle.
\]

Estimating the right-hand side using (\ref{kremainder}) and
Proposition~\ref{input2} we find a $C>0$ such that
\[
\langle \psi_m, H' \psi_m\rangle + 2k\langle \hat{\psi}_m,H'\hat{\psi}_m\rangle
\leq \frac{c_0}4 \|\psi_m\|^2
+ C \|N^\frac12 (A-\i\sigma)^{k-1}\psi\|^2.
\]

Finally, we appeal to Proposition~\ref{input3} to derive the bound
\begin{equation}\label{bndonpsim}
\frac{c_0}4\|\psi_m\|^2 \leq  C  \|N^\frac12 (A-\i\sigma)^{k-1}\psi\|^2
+\tC\langle \psi_m, K_0 \psi_m\rangle
+ 2k\tC\langle \hat{\psi}_m, K_0 \hat{\psi}_m\rangle.
\end{equation}
Pick $\Lambda>0$ large enough such that
\[
2\tC\|K_0\one_{[|A|>\Lambda]}\|\leq \frac{c_0}{12(1+2kP^2)},
\]
where $P$ is given by (\ref{punif}).
Write $\one_{[|A|\leq \Lambda]} \psi_m = [\one_{[|A|\leq \Lambda]}(g_m(A)-g_{2m}(A))]^k\psi$ and estimate using \eqref{eq:21}
\begin{align*}
2\tC |\langle \one_{[|A|\leq \Lambda]} \psi_m,K_0\psi_m\rangle|
& \leq 2\tC(\Lambda+C_{0,0})^k \|K_0\|\|\psi\|\|\psi_m\|\\
& \leq \frac{c_0}{12}\|\psi_m\|^2 + \frac{12\tC^2(\Lambda+C_{0,0})^{2k}\|K_0\|^2}{c_0}\|\psi\|^2
\end{align*}
and similarly
\[
2\tC|\langle \one_{[|A|\leq \Lambda]} \hat{\psi}_m,K_0\hat{\psi}_m\rangle|
\leq  \frac{c_0}{24 k}\|\psi_m\|^2 + \frac{24 k \tC^2  (\Lambda+C_{0,0})^{2k}\|K_0\|^2 P^4}{c_0}\|\psi\|^2.
\]
Inserting $\one=\one_{[|A|\leq \lambda]} + \one_{[|A|>\lambda]}$ ahead of the $K_0$'s in
(\ref{bndonpsim}) and appealing to the bounds above we get
\[
\frac{c_0}8\|\psi_m\|^2\leq  C \big(\|N^\frac12 (A-\i\sigma)^{k-1}\psi\|^2 + \|\psi\|^2\big),
\]
for a suitable $m$-independent $C$.
Recalling \eqref{domainchar} we conclude that $\psi\in\cD(A^k)$.

It remains to prove that $A^k\psi\in \cD(N^{1/2})$.

Note that what we just established implies that
$\psi_m\to A^k\psi$ in norm, cf. \eqref{eq:9} and \eqref{eq:21}.
We can now compute
\[
\langle A^k\psi, NI_{\i n}(N) A^k\psi\rangle =
\lim_{m\to\infty}\langle \psi_m, N I_{\i n}(N)\psi_m\rangle.
\]
But by Propositions~\ref{input1} and~\ref{input2} we have
\begin{eqnarray*}
\langle \psi_m, N I_{\i n}(N)\psi_m\rangle &\leq & \|N^\frac12\psi_m\|^2 \\
& \leq & \|N^\frac12\psi_m\|^2 + 2k\|N^\frac12\hat{\psi}_m\|^2\\
& \leq & 2 C_2\big(\langle \psi_m, H'\psi_m\rangle +
2k\langle \hat{\psi}_m,H'\hat{\psi}_m\rangle \big) + C\\
& = & \langle\psi, R_m\psi\rangle + C,
\end{eqnarray*}
where $C>0$ is constant independent of $m$. The result now follows from
(\ref{kremainder}) by first taking the limit $m\to\infty$, and subsequently $n\to\infty$.
Notice that Lebesgue's
theorem on monotone convergence applies, since $I_{\i n}(N) = n(N+n)^{-1} \to \one$
monotonously.
$\hfill\square$

The rest of the section is devoted to establishing
Propositions~\ref{input1}--\ref{input3}.

We begin with a definition and a series of lemmata. The $\sigma$ in the definition below is the same $\sigma$
that entered into Definition~\ref{Def-k-remainder}.

\begin{defn}\label{lrerror} Let $E_m^\roml$ and $E^\romr_m$ be families of forms on $\cD_{k-1}\times\cD(N^{1/2})$
 and $\cD(N^{1/2})\times \cD_{k-1}$ respectively. We say that $E^\roml_m$ is a left-error
if
\[
|\langle\psi, E_m^{\roml}\varphi\rangle| \leq  C\|N^\frac12(A-\i \sigma)^{k-1}\psi\|\|N^\frac12\varphi\|.
\]
We say that $E^\romr_m$ is a right-error if
\[
|\langle\psi, E_m^{\romr}\varphi\rangle| \leq  C\|N^\frac12\psi\|\|N^\frac12(A-\i \sigma)^{k-1}\varphi\|.
\]
\end{defn}

\begin{remark}\label{Rem-RightError} An example of a right-error that we will encounter below are
forms  
\[
N^{1/2}B_m N^{1/2}g_{1m}(A)^\ell (A-\i \sigma)^{-j},
\]
with $\ell-j\leq k-1$ and $\sup_{m}\|B_m\|<\infty$.
To see that this is a right-error 
observe that it suffices to prove that $N g_{1m}(A)^\ell (A-\i \sigma)^{-j-k+1} N^{-1}$ is
uniformly bounded in $m$.
The result then follows from interpolation.
Since $j+k-1\geq \ell$, recalling that $\sigma$ was chosen according to  \eqref{eq:32}, we reduce the problem to showing that
$Ng_{1m}(A)(A-\i\sigma)^{-1}N^{-1}$ is bounded uniformly in $m$.
But this follows by a computation similar to \eqref{inlem24}, where the extra resolvent
produces a bound which is uniform in $m$ compared with the point wise bound \eqref{inlem24a}.
\end{remark}

We introduce the notation
\begin{equation}\label{Hn}
H_n := H I_n(H) =\i n(I_n(H)-\one),
\end{equation}
which plays the role of a regularized Hamiltonian. See \eqref{eq:5} for the definition of $I_n(H)$.

\begin{lemma}\label{Lemma-RemoveHcutoff} We have the following limit in the sense of forms on $\cD(N^{1/2})$
\[
\lim_{n\to \infty} \i[H_n,g_{1m}(A)] = -\int_{V_m^>}
\eta R_m(\eta)H' R_m(\eta)\, \d\lambda(\eta).
\]
\end{lemma}

\begin{proof} Observe first that the integral on the right-hand side in the lemma is norm convergent.

Compute as a form on $\cD(A)$ using that the
integral representation for $g_{1m}(A)$ is strongly convergent on $\cD(A)$
\[
\i[H_n,g_{1m}(A)] =  \int_{V_m^>} \eta \i [H_n,R_m(\eta)]\,\d\lambda(\eta).
\]
Recalling \eqref{Hn} we arrive at
\begin{align*}
\i[H_n,g_{1m}(A)] & = \i n \int_{V_m^>} \eta \i[I_n(H),R_m(\eta)]\,\d\lambda(\eta)\\
&  =  \int_{V_m^>} \eta I_n(H) \i[H,R_m(\eta)] I_n(H)\,\d\lambda(\eta).
\end{align*}
Finally we employ Condition~\ref{cond:condition}~\ref{item:2i0ue}) to conclude that for each $n$, the following holds as a form identity
on $\cD(A)\cap\cD(N^{1/2})$
\[
\i[H_n,g_{1m}(A)] = - \int_{V_m^>} \eta I_n(H) R_m(\eta) H' R_m(\eta) I_n(H)\,\d\lambda(\eta).
\]
The integral on the right-hand side of the above identity is absolutely convergent in $\cB(N^{-1/2}\cH; \linebreak[1] N^{1/2}\cH)$.
By density of $\cD(A)\cap \cD(N^{1/2})$ in  $\cD(N^{1/2})$, see Remark~\ref{RemarkTo1}~2), the identity therefore extends to a form identity on $\cD(N^{1/2})$.
The lemma now follows from \eqref{strongI1}.
\end{proof}

\begin{lemma}\label{symm1} Let $1\leq k\leq k_0$.
\begin{enumerate}[\quad\normalfont (1)]
\item
There exist right-errors $E_m^{\romr}$,
$\hat{E}_m^{\romr}$ such that, as forms on $\cD(N^{1/2})\times\cD_{k-1}$,
\begin{align*}
&\lim_{n\to\infty} \i [H_n, g_{1m}(A)^k]  = E_m^{\romr}\\
&\lim_{n\to\infty} \i [H_n, \hat{g}_{1m}(A)g_{1m}(A)^{k-1}]   =
\hat{E}_m^{\romr}.
\end{align*}
\item There exist a left-error $E^\roml_m$ and a right-error $E_m^{\romr}$ such that, as forms on
$\cD_{k-1}\times\cD(N^{1/2})$ and
$\cD(N^{1/2})\times\cD_{k-1}$ respectively,
\begin{align*}
& \lim_{j\to\infty}\lim_{n\to\infty} \i[H_n,g_{1m}(A)^k]A_j
 =    k g_{1m}(A)^{k-1} A g_{1m}'(A) H' + E_m^\roml  \\
 & \lim_{j\to\infty}\lim_{n\to\infty} A_j\i[H_n,g_{1m}(A)^k]
 =    k H' A g_{1m}'(A)g_{1m}(A)^{k-1} + E_m^\romr.
\end{align*}
\end{enumerate}
\end{lemma}

\begin{proof} (1) also holds if we take the limit in the sense of forms on
$\cD_{k-1}\times\cD(N^{1/2})$ and replace the right-error by a left-error.
We will however not need that statement. One does however need its proof for the
left-error part of (2).

In the proof we will only work with right-errors. The other case is similar.
We begin with (1) and prove only the first statement leaving the second
to the reader.

 We first compute as a form on $\cD(N^{1/2})$.
\begin{eqnarray}\label{symm1_1}
\nonumber \i [H_n, g_{1m}(A)^k] & = &  k\i[H_n,g_{1m}(A)]g_{1m}(A)^{k-1}\\
& &  +\sum_{\ell=2}^k  (-1)^{\ell+1}\bin{k}{\ell}
\i \, \ad_{g_{1m}(A)}^\ell(H_n) g_{1m}(A)^{k-\ell}.
\end{eqnarray}

We now analyze the large $n$ limit. The first term on the right-hand side of \eqref{symm1_1}
can be dealt with using Lemma~\ref{Lemma-RemoveHcutoff} directly, observing that by Lemma~\ref{dominv} $g_{1m}(A)$ preserves the domain of $N^{1/2}$.
As for the terms involving higher order commutators,
we again use Lemma~\ref{Lemma-RemoveHcutoff} to compute
\[
\lim_{n\to\infty} \i \,\ad_{g_{1m}(A)}^\ell(H_n) = - \int_{V_m^>} \eta R_m(\eta)\ad_{g_{1m}(A)}^{\ell-1}(H')R_m(\eta)\, \d\lambda(\eta)
\]
in the sense of forms on $\cD(N^{1/2})$.

We can now employ Condition~\ref{cond:conditioni} to compute as forms on $\cD(N^{1/2})$
\begin{equation}\label{symm1_4}
 \lim_{n\to\infty} \i\, \ad_{g_{1m}(A)}^\ell(H_n) = (-1)^\ell N^\frac12 B^{(\ell)}_m  N^\frac12 ,
\end{equation}
where $B^{(\ell)}_m$ is a family of bounded operators with $\sup_{m} \|B_m^{(\ell)}\|<\infty$, for all
$\ell$. They are given by
\begin{eqnarray}\label{symm1_3}
\nonumber B^{(\ell)}_m  & = &
 \int_{(V_m^>)^\ell} \eta_1\cdots\eta_\ell N^{-\frac12} R_m(\eta_1)\cdots R_m(\eta_\ell)\ad_{A}^{\ell-1}(H')\\
& & \ \ \ \times
R_m(\eta_\ell)\cdots R_m(\eta_1)N^{-\frac12} \d\lambda(\eta_1)\cdots \d\lambda(\eta_\ell).
\end{eqnarray}
From \eqref{symm1_1}, \eqref{symm1_4} and Lemma~\ref{Lemma-RemoveHcutoff} we thus obtain
\begin{align}\label{symm1-Step1}
\nonumber \lim_{n\to\infty}\i [H_n, g_{1m}(A)^k] & = -k\int_{V_m^>}
\eta R_m(\eta)H' R_m(\eta)\, \d\lambda(\eta)g_{1m}(A)^{k-1}\\
& \ \ \ -  \sum_{\ell=2}^k \bin{k}{\ell}
 N^\frac12 B^{(\ell)}_m  N^\frac12  g_{1m}(A)^{k-\ell}.
\end{align}
Combining this computation with Remark~\ref{Rem-RightError} yields (1).

We now turn to part (2) of the lemma. In view of \eqref{symm1-Step1} we begin by computing
as a form on $\cD(N^{1/2})$,
using  Condition~\ref{cond:condition}~(\ref{item:1})
\begin{align}\label{symm1_2}
\nonumber & -k\int_{V_m^>} \eta R_m(\eta)H' R_m(\eta)\d\lambda(\eta) \\
& =   kH' g_{1m}'(A)  - \frac{\i k}{m}  \int_{V_m^>}\eta R_m(\eta) H'' R_m(\eta)^2\d\lambda(\eta).
\end{align}
We remark that the identity $\i [H',R_m(\eta)] = - m^{-1} R_m(\eta) H'' R_m(\eta)$
holds a priori as a form identity on $\cD(N)$.
It extends by continuity to a form identity on $\cD(N^{1/2})$,
which is what is used in the above computation.
Note that the integral on the right-hand side is convergent as a form on $\cD(N^{1/2})$.

From \eqref{symm1-Step1}, \eqref{symm1_2} and Remark~\ref{Rem-RightError} we find that
\[
\lim_{n\to\infty} \i [H_n, g_{1m}(A)^k] A_j = \big( kH' g_{1m}'(A)A g_{1m}(A)^{k-1} + E_m^\romr \big)I_j(A)
\]
and hence by \eqref{strongI3} we conclude the following identity as forms on $\cD(N^{1/2})$
\[
\lim_{j\to\infty}\lim_{n\to\infty} \i [H_n, g_{1m}(A)^k] A_j =  kH' g_{1m}'(A)A g_{1m}(A)^{k-1} + E_m^\romr.
\]
To prove the second statement in (2) it remains to show that the commutator between  $A_j$
and $\i[H_n,g_{1m}(A)^k]$ converges to a right-error.

From (\ref{symm1_2}) we get,
as a form on $\cD(N^{1/2})$,
\begin{align*}
&  \Big[-k\int_{V_m^>} \eta R_m(\eta)H' R_m(\eta)\d\lambda(\eta),A_j\Big]\\
 & = kI_j(A)H''I_j(A) g_{1m}'(A) -\frac{\i k}{m}\int_{V_m^>}
\eta R_m(\eta)(H''A_j - A_jH'' ) R_m(\eta)^2\d\lambda(\eta).
\end{align*}
We can now take the limit $j\to\infty$ and obtain
\begin{equation}\label{symm1_5}
 \lim_{j\to\infty}\Big[-k \int_{V_m^>} \eta R_m(\eta)H' R_m(\eta)\d\lambda(\eta),A_j\Big] =  N^\frac12 B^{(1)}_m N^\frac12,
\end{equation}
where $B^{(1)}_m$, is a family of bounded operators with $\sup_m\|B^{(1)}_m\|<\infty$. It is given by
\[
B^{(1)}_m  = k N^{-\frac12}\Big\{H''g_{1m}'(A) -\i \int_{V_m^>}
\eta \big(R_m(\eta)H'' -  H''R_m(\eta)\big)R_m(\eta) \d\lambda(\eta)\Big\}N^{-\frac12}.
\]
Here we used \eqref{strongI3}, that
$A_j R_m(\eta) = R_m(\eta)A_j = m(\one +\eta R_m(\eta))I_j(A)$, as well as Lebesgue's theorem on
dominated convergence.

For the commutator between  $A_j$ and the second term on the right-hand side of (\ref{symm1-Step1})
we compute
\[
[N^\frac12 B^{(\ell)}_m N^\frac12,A_j] = I_j(A)N^\frac12\tilde{B}_m^{(\ell)} N^\frac12 I_j(A),
\]
where $\tilde{B}_m^{(\ell)}$ are bounded operators with $\sup_{m\in\NN}\|B_m^{(\ell)}\|<\infty$, for all $\ell$.
They are given by
\begin{eqnarray*}
\lefteqn{ \tilde{B}_m^{(\ell)} = \int_{(V_m^>)^\ell}  N^{-\frac12} R_m(\eta_1)\cdots R_m(\eta_\ell)
\ad_{A}^{\ell}(H')}\\
& & \ \ \ \times
R_m(\eta_\ell)\cdots R_m(\eta_1) N^{-\frac12} \d\lambda(\eta_1)\cdots \d\lambda(\eta_\ell) .
\end{eqnarray*}
We can now take the limit $j\to\infty$ using \eqref{strongI3}, and the resulting expression together with \eqref{symm1_5},
the formula \eqref{symm1-Step1} and Remark~\ref{Rem-RightError} yields that
\[
\lim_{j\to\infty}\lim_{n\to\infty} [\i [H_n, g_{1m}(A)^k], A_j]  = E_m^\romr.
\]
\end{proof}

\begin{lemma}\label{symm2} There exists a $k$-remainder $R_m$ such that
\begin{eqnarray*}
\lefteqn{\lim_{j\to\infty}\lim_{n\to\infty}\i [ H_n, g_{1m}(A)^k A_j g_{1m}(A)^k]}\\
&= & g_{1m}(A)^kH' g_{1m}(A)^k + 2k\Re\{g_{1m}(A)^{k-1}A g_{1m}'(A) H' g_{1m}(A)^k\}
+R_m,
\end{eqnarray*}
in the sense of forms on $\cD_{k-1}$.
\end{lemma}

\begin{proof}
We compute as a form on $\cD_{k-1}$
\begin{eqnarray*}
\lefteqn{\i [ H_n, g_{1m}(A)^k A_j g_{1m}(A)^k]  =  \i[H_n,g_{1m}(A)^k]A_j g_{1m}(A)^k}\\
& & + g_{1m}(A)^k \i[H_n,A_j] g_{1m}(A)^k +  g_{1m}(A)^k A_j\i[ H_n,g_{1m}(A)^k].
\end{eqnarray*}
Using that $\lim_{n\to\infty}\i[H_n,A_j] = I_j(A)H'I_j(A)$,
$\lim_{j\to\infty} I_j(A) H' I_j(A)= H'$ (in the sense of forms on $\cD(N^{1/2})$),
and Lemma~\ref{symm1}~(2), we conclude the result, with
\[
R_m = E_m^\roml g_{1m}(A)^k + g_{1m}(A)^k E_m^\romr.
\]
Note that $R_{m}$ is a $k$-remainder, in the sense of Definition~\ref{Def-k-remainder}.
\end{proof}

We now  symmetrize the form
$g_{1m}(A)^{k-1} Ag_{1m}'(A) H' g_{1m}(A)^k$, defined on $\cD(N^{1/2})$.

\begin{lemma}\label{lastsymm} There exists a $k$-remainder $R_m$ such that
\begin{align*}
& \Re\{g_{1m}(A)^{k-1}A g_{1m}'(A) H' g_{1m}(A)^k\}\\
& \ \ \  = g_{1m}(A)^{k}p_{1m}(A) H' p_{1m}(A) g_{1m}(A)^{k} + R_m,
\end{align*}
in the sense of forms on $\cD_{k-1}$.
\end{lemma}

\begin{proof} {\bf Step I:}
From the proof of Lemma~\ref{dominv} it follows that
\begin{align}\label{stepIIbnd}
& [N,Ag_{1m}'(A)]N^{-1}, \ \ \ \ [N,p_{1m}^2(A)g_{1m}(A)]N^{-1},\\
& \label{stepIIIbnd}
\mathrm{and} \ \  N^{-1}p_{1m}(A)N
\end{align}
extend as forms from $\cD(N)$ to bounded operators with norm bounded uniformly in $m$.

\noindent{\bf Step II:} Boundedness of the forms in (\ref{stepIIbnd}), together with the observation that $\|tg_{1m}' - p_{1m}^2g_{1m}\|_\infty$
is bounded uniformly in $m$, implies after an interpolation argument that
\[
N^\frac12\big(Ag_{1m}'(A)-p_{1m}(A)^2 g_{1m}(A)\big)N^{-\frac12}
\]
is bounded uniformly in $m$. Hence
\begin{eqnarray*}
\lefteqn{\Re\{g_{1m}(A)^{k-1}A g_{1m}'(A)H'g_{1m}(A)^{k}\}}\\
& = & g_{1m}(A)^k \Re\{p_{1m}(A)^2 H'\}g_{1m}(A)^k + R_m^{(1)},
\end{eqnarray*}
where $R_m^{(1)}$ is a $k$-remainder.

\noindent {\bf Step III:} We compute as a form on $\cD(N^{1/2})$
\[
(A+\i\sigma)[p_{1m}(A),H'] = -\i \int_{V_m^>}\frac{A+\i\sigma}{m}R_m(\eta) H'' R_m(\eta)\,\d\lambda_p(\eta),
\]
which is bounded uniformly in $m$ as a form on $\cD(N^{1/2})$.
This together with \eqref{stepIIIbnd} and a interpolation argument as in step II, shows that
\[
g_{1m}(A)^k\Re\{p_{1m}(A)^2 H'\}g_{1m}(A)^k = g_{1m}(A)^k p_{1m}(A)H' p_{1m}(A) g_{1m}(A)^k + R_m^{(2)},
\]
where $R_m^{(2)}$ is a $k$-remainder. Here we used again Remark~\ref{Rem-RightError}. This proves the lemma with $R_m= R_m^{(1)}+R_m^{(2)}$.
\end{proof}

\vspace{5mm}

\begin{proof}[Proof of Proposition~\ref{input1}]
Combine Lemmas \ref{symm2} and \ref{lastsymm}.
\end{proof}

\vspace{5mm}

\begin{proof}[Proof of Proposition~\ref{input2}]
We only prove the first estimate. The second is verified the same way.
We can assume that $\lambda=0$.

We estimate using Condition~\ref{cond:virial}
\begin{eqnarray}\label{virhelp1}
\nonumber\lefteqn{\|N^\frac12 I_n(H)\psi_m\|^2 \leq
C_1\langle I_n(H)\psi_m,H I_n(H)\psi_m\rangle}\\
& & + C_2 \langle I_n(H)\psi_m,H' I_n(H)\psi_m\rangle
+ C_3\|I_n(H)\psi_m\|^2.
\end{eqnarray}
Note that $HI_n(H)\psi_m  = H_n\psi_m = [H_n,g_{1m}(A)^k]\psi$.

By Lemma~\ref{symm1}~(1) we find that for any $\varphi\in\cD(N^{1/2})$
we have
\begin{equation}\label{virhelp}
\lim_{n\to\infty} \langle\varphi,H_n\psi_m\rangle = \langle \varphi,E_m^\romr\psi\rangle.
\end{equation}
By this observation  and the uniform boundedness principle there exists $C = C(m)$ such that
$|\langle\varphi, H_n\psi_m\rangle| \leq C \|N^{1/2}\varphi\|$ uniformly in $n$,
for $\varphi\in\cD(N^{1/2})$. Applying this to $\varphi= (I_n(H)-I)\psi_m$, together
with (\ref{virhelp}), now applied with $\varphi=\psi_m$, we get
\begin{equation}\label{remH}
\lim_{n\to\infty} \langle I_n(H)\psi_m,H I_n(H)\psi_m\rangle =
\langle\psi_m,E_m^\romr\psi\rangle.
\end{equation}
Here $E_m^\romr$ is a right-error.

We can now take the limit $n\to\infty$ in (\ref{virhelp1}), and the result follows
from Definition~\ref{lrerror}.
\end{proof}

\vspace{5mm}

\begin{proof}[Proof of Proposition~\ref{input3}]
As above we assume $\lambda=0$ and prove only the first bound.

By Remark~\ref{RemarkTo1}~4) it suffices to estimate using the bound \eqref{Eq-WeakMourre}
instead of the one in Condition~\ref{cond:mourre}. We get
\begin{align}\label{input3main}
\nonumber & \langle I_n(H)\psi_m,H' I_n(H)\psi_m\rangle   \geq  c_0 \|I_n(H)\psi_m\|^2\\
  & \ \ + \Re\langle I_n(H)\psi_m,BH I_n(H)\psi_m\rangle - \langle I_n(H)\psi_m,K_0 I_n(H)\psi_m\rangle.
\end{align}
%

Arguing as in the part of the proof of  Proposition~\ref{input2} pertaining to (\ref{remH}),
we find that
\[
\lim_{n\to\infty} \Re\langle I_n(H)\psi_m, B H I_n(H)\psi_m\rangle = \Re\langle \psi_m,E_m^\romr\psi\rangle.
\]
where $E_m^{\romr}$ is a right-error. Here \eqref{virhelp} was used (twice) with $\varphi$ replaced by $B\varphi$
and $B^*\varphi$,
where we used the assumption on $B$ in Remark~\ref{RemarkTo1}~4) to argue that $B\varphi,B^* \varphi\in\cD(N^{1/2})$ in \eqref{virhelp}.

Inserting this limit into \eqref{input3main} yields
\begin{eqnarray*}
\langle\psi_m, H'\psi_m\rangle & = &
\lim_{n\to\infty} \langle I_n(H)\psi_m,H' I_n(H)\psi_m\\
& \geq & c_0\|\psi_m\|^2 - \langle\psi_m,K_0 \psi_m\rangle + \Re\langle \psi_m,E_m^\romr\psi\rangle,
\end{eqnarray*}
with $E_m^\romr$ being a right-error.
Using Definition~\ref{lrerror} and Proposition~\ref{input2} we conclude the first estimate.
\end{proof}



\subsection{Proof of Theorem \ref{thm:mainresultk=1}}

We shall show Theorem~\ref{thm:mainresultk=1},
which is an extension of Corollary \ref{cor:assumpt-stat-regul} under the
minimal condition $k_0=1$.

\begin{proof}[Proof of Theorem~\ref{thm:mainresultk=1}:]
 We can without loss of generality take $\lambda=0$.
 Due to  Corollary~\ref{cor:assumpt-stat-regul} only the first
  statement needs elaboration. The idea of the proof is to apply a
  virial argument for the commutator $\i [H,A]$ and the state
  $N^{1/2}\psi$. We divide the proof into three steps. Let $N_n^{(1/2)}=N^{1/2}I_{\i n}(N)$.

\noindent {\bf Step I:} Due to Lemma~\ref{HN0} we have
 $N^{(1/2)}_n\psi\in\cD(H)$. We shall show that
  \begin{equation}  \label{eq:4}
    \sup _{n\in \N}\| H N_{n}^{(1/2)}\psi\|<\infty.
  \end{equation}

We can use the
representation formula \eqref{Eq-Thomas} with $\alpha= 1/2$ and commute $H$ through $N^{1/2}$,
cf. (\ref{eq:30ny}). Whence it suffices to bound
\begin{equation*}
  \int_0^\infty t^{\frac12} (N+t)^{-1}[H,N]^0(N+t)^{-1}I_{\i n}(N)N^{-\frac12}\d t
\end{equation*}
independently of $n$. (Note that the contribution from
commuting through the second factor $I_{\i n}(N)$ indeed is bounded
independently of $n$.)  By (\ref{eq:20}) we have
\begin{equation*}
  [H,N]^0= N^{\frac12-\kappa}B N^{\frac12-\kappa}\text{ for }B\text{ bounded},
\end{equation*} and we can estimate
\begin{equation*}
 \|(N+t)^{-1} \i [H,N]^0( N+t)^{-1}I_{\i n}(N)N^{-\frac12}\|\leq \|B\| \inp{t}^{-\frac32-\kappa}\text{ uniformly in } n.
\end{equation*} Hence the integrand is $O(t^{-1-\kappa})$ uniformly in $n$, and (\ref{eq:4}) follows.

\noindent {\bf Step II:} We shall show that
  \begin{equation}
    \label{eq:4i}
    \sup _{n\in \N}\| AN_n^{(1/2)}\psi\|<\infty.
  \end{equation}
Since $\phi:=N^{1/2}\psi\in \cD(A)$ due to  Corollary~\ref{cor:assumpt-stat-regul} it suffices to bound
 the state $[A,I_{\i n}(N)]\phi$
 independently of $n$. This is obvious from the representation
 \begin{equation*}
   [A,I_{\i n}(N)]\phi=-\i (N+n)^{-1}N'I_{\i n}(N)\phi,
 \end{equation*}
and whence (\ref{eq:4i}) follows.

\noindent {\bf Step III:} We look at
\begin{equation*}
  \inp{\i [H,A]}_{N^{(1/2)}_n\psi}=-2\Re { \inp{\i H N^{(1/2)}_n\psi,AN^{(1/2)}_n\psi}}.
\end{equation*} Due to (\ref{eq:4}) and (\ref{eq:4i}) the right hand
side is bounded independently  of $n$. We compute using
Condition~\ref{cond:condition}~(\ref{item:2i0a}) and~(\ref{item:2i0ue})
\begin{equation*}
  \inp{\i
    [H,A]}_{N^{(1/2)}_n\psi}=\lim _{\tilde n\to\infty}\inp{\i
    [H,AI_{\tilde n}(A)]}_{N^{(1/2)}_n\psi}=\inp{H'}_{N^{(1/2)}_n\psi}.
\end{equation*}
Whence using the virial estimate Condition~\ref{cond:virial}
(and also Step I again) we conclude that
\begin{equation*}
  \inp{N}_{N^{(1/2)}_n\psi}\leq C\textup{ uniformly in } n.
\end{equation*} Taking $n\to \infty$ we obtain that indeed
$\psi\in \cD(N)$.
\end{proof}

\subsection{Theorem on more $N$--Regularity}\label{subsec-MoreNReg}

 We formulate and prove an extended version of Theorem
 \ref{thm:mainresultk=1}.

Notice that under
Condition~\ref{cond:condition}~(\ref{item:2i0a}) and~(\ref{item:2}), and the additional
condition \eqref{eq:5oiooa} for $k_0=1$,
\begin{equation}\label{eq:1}
  N^{\frac12}\in C^1_{\mo}( A)\cap C^1_{\mo}(H),
\end{equation}
cf. Lemma~\ref{HN1} and Proposition~\ref{Prop-CommN12}.

We impose the conditions of Corollary~\ref{cor:assumpt-stat-regul} and aim at an improvement of
Corollary~\ref{cor:assumpt-stat-regul} and Theorem~\ref{thm:mainresultk=1}
in the case  $k_0\geq 2$. Let $M_0=\i[N^{1/2},A]^0$. Then, cf. Proposition~\ref{Prop-CommN12},
\begin{equation}
  \label{eq:2}
  \i ^m\ad_{A}^m(M_0) \text{ is }N^{\frac12}\text{--bounded for }m=0,\dots,k_0-1.
\end{equation} Here the commutators are defined iteratively as
extensions of forms on $\cD(N^{1/2})\cap\cD(A)$ and  they are considered
as symmetric $N^{1/2}$--bounded operators. We introduce the following $N^{1/2}$--bounded operators:
\begin{align*}
  M_1&=\i [N^{\frac12},H]^0=c_{\frac12}\int_0^\infty t^{\frac12} (N+t)^{-1}\i[N,H]^0(N+t)^{-1}\, \d t,\\
M_2&=H'N^{-\frac12} \ \ \textup{and} \ \ M_3=N^{-\frac12}H'.
\end{align*} Notice that
\begin{align}
  \label{eq:3}
  M_3\subseteq M_2^*\mand M_2\subseteq
M_3^*.
\end{align}

 We need to consider repeated commutation of $M_j$, $j=1,
\dots,3$,  with factors
of $T= A$ or $T=N^{1/2}$.
\begin{cond}\label{cond:more-n-regularity}
For all $j=1,
\dots,3$, $m=1,\dots,k_0-1$  and all possible combinations of
factors $T_n\in \{A,N^{1/2}\}$ where $n=1,\dots,m$
 \begin{equation}
  \label{eq:2b}   \i ^m\ad_{T_m}\cdots \ad_{T_1}(M_j) \text{ is }N^{\frac12}\text{--bounded}.
\end{equation}
\end{cond}
Notice that in \eqref{eq:2b} the commutators are defined iteratively as
extensions of forms on $\cD(N^{1/2})\cap\cD(A)$ using \eqref{eq:3} and
the analogue properties for $m\geq  2$
\begin{align*}
  (-1)^{m-1}\ad_{T_{m-1}}\cdots \ad_{T_1}(M_3) &\subseteq \big(\ad_{T_{m-1}}\cdots \ad_{T_1}(M_2)\big)^*,\\
 (-1)^{m-1}\ad_{T_{m-1}}\cdots \ad_{T_1}(M_2) &\subseteq \big(\ad_{T_{m-1}}\cdots \ad_{T_1}(M_3)\big)^*.
\end{align*}

We shall prove the following extension of Corollary~\ref{cor:assumpt-stat-regul}
and Theorem~\ref{thm:mainresultk=1}.
\begin{thm}\label{thm:more-n-regularity} Suppose the conditions of
  Corollary~\ref{cor:assumpt-stat-regul} and for $k_0\geq 2$ also
  Condition~\ref{cond:more-n-regularity}. Let $\psi\in\cD(N^{1/2})$ be a bound state
$(H-\lambda)\psi=0$ (with $\lambda$ as in Condition~\ref{cond:mourre}).
 Then $\psi\in \cD(T_{k_0+1}\cdots T_1)$ where
  $T_n\in \{A,N^{1/2},\one\}$ for $n=1,\dots,k_0+1$ and at least for one such
  $n$,  $T_n\neq A$.
\end{thm}
\begin{proof}
  We proceed by induction in $k_0$. The case $k_0=1$ is the content
  of Theorem~\ref{thm:mainresultk=1}. So suppose $k_0\geq2$ and that the statement
holds for $k_0\to k_0-1$. Consider any product $S=T_{k_0+1}\cdots T_1$
not all factors being given by $A$. We shall show that $\psi\in
\cD(S)$. By Corollary~\ref{cor:assumpt-stat-regul} and the induction hypothesis we can
assume that  the factors $T_n\in \{A,N^{1/2}\}$ and that for at least two $n$'s
$T_n=N^{1/2}$. By using \eqref{eq:2} and the induction hypothesis we
can assume that
$T_{k_0+1}=N^{1/2}$. Whence we can assume  $S=N^{1/2}S_{k,\ell}^\alpha$
with  $k=k_0$ introducing here  the following   notation for
$k=1,\dots,k_0$, $\ell=0,\dots, k$ and $\alpha$ being a multiindex
$\alpha\in \{0,1\}^{k}$ with $\sum_{j\leq k}\alpha_{j}=\ell$,
\begin{equation*}
  S_{k,\ell}^\alpha=S_{\alpha_k}\cdots S_{\alpha_1}=:\prod_{j=1}^{k}S_{\alpha_j}\textup{ where
  }S_0=A\mand S_1=N^{\frac12}.
\end{equation*}

 Partly motivated by the above considerations we introduce the following
 quantity for $n\in \N$ large and $\epsilon\in ]0,1[$ small
 \begin{equation*}
   f(n,\epsilon)=\sum _{\ell=0}^{k_0}\epsilon^{-2\ell^2}\;g(n,\ell);\ \ g(n,\ell)
   :=\sum _{\stackrel{\alpha\in \{0,1\}^{k_0}}{\alpha_1+\cdots +\alpha_{k_0}=\ell}}\,\big\|N^{\frac12}I_{\i n}(N)S_{k_0,\ell}^\alpha\psi\big\|^2.
 \end{equation*} We claim that for some constants $K_1,K_2(\epsilon)>0$ independent of $n$
 \begin{equation}
   \label{eq:6}
   f(n,\epsilon)\leq \epsilon^2K_1f(n,\epsilon)+K_2(\epsilon).
 \end{equation} The theorem follows from \eqref{eq:6} by
 first
 choosing  $\epsilon$ so small that $\epsilon^2K_1\leq 1/2$,
 subtraction of the first term on the right-hand side and then letting
 $n\to \infty$. By  Corollary~\ref{cor:assumpt-stat-regul} (or Theorem~\ref{thm:mainresulti}), $\sup_n g(n,\ell=0)<\infty$, in agreement with
\eqref{eq:6}.

To see how the factor $\epsilon^2$ comes about let us note that
\begin{equation*}
  -2\ell^2=-(\ell-1)^2-(\ell+1)^2+2,
\end{equation*} whence (to be used later) we can for $\ell=1,\dots,k_0-1$ bound the  expression
\begin{equation}\label{eq:7}
 \epsilon^{-2\ell^2} \sqrt{ g(n,\ell-1)}\;\sqrt{g(n,\ell+1)}\leq \epsilon^2f(n,\epsilon).
\end{equation}

To show \eqref{eq:6} we mimic the proof of Theorem~\ref{thm:mainresultk=1}. Again this is in three steps and we assume
that $\lambda=0$. We need to bound each term of  $g(n,\ell)$ for $\ell\geq 1$.

\noindent {\bf Step I:} Bounding $\|HI_{\i
  n}(N)S_{k_0,\ell}^\alpha\psi\|$. We expand into terms; some can be
bounded independently of $n$ (using the induction hypothesis)  while
others will be estimated as
$C\sqrt{g(n,\ell+1)}$ (assuming here that $\ell\leq k_0-1$). We compute
 formally
\begin{equation}\label{eq:23}
  \i\big[H,I_{\i n}(N)S_{k_0,\ell}^\alpha\big]
   = \i\big[H,I_{\i n}(N)\big]S_{k_0,\ell}^\alpha+I_{\i n}(N)\i\big[H,\prod_{j=1}^{k_0}S_{\alpha_j}\big],
\end{equation}
where the second commutator is  expanded as
\begin{equation}\label{eq:24}
  \i\big[H,\prod_{j=1}^{k_0}S_{\alpha_j}\big]=\sum^{m=k_0}_{m=1}\big(\prod_{j=m+1}^{k_0}S_{\alpha_j}\big)
  \i\big[H,S_{\alpha_m}\big]\big(\prod_{j=1}^{m-1}S_{\alpha_j}\big).
\end{equation}
In turn we have the  expressions
\begin{subequations}
\begin{align}
  \label{eq:8} \i[H,I_{\i n}(N)]
  &=n^{-1}I_{\i n}(N)\i[N,H]^0 I_{\i n}(N),\\\label{eq:11}
 \i[H,S_{\alpha_m}]&=-M_1\;\;\;\,\text{ if }\alpha_m=1,\\
\i[H,S_{\alpha_m}]&=M_2S_1 \;\;\;\text{ if }\alpha_m=0.\label{eq:22}
\end{align}
\end{subequations}
We plug \eqref{eq:8}--\eqref{eq:22} into
\eqref{eq:23} and \eqref{eq:24} and look at each term
separately. Before embarking on a such  examination  we need to
``fix'' the above formal
computation. This is done in terms of  multiple  approximation
somewhat similar
to the one of the proof of Theorem~\ref{thm:mainresulti}. We replace
$H\to H_p$ and the factors $A\to A_q$ and $N^{1/2}\to
N^{1/2}_{\i q}=(N^{1/2})_{\i q}$. More precisely  it  is convenient to introduce
$k_0$ different $q$'s,  say $q_1,\dots, q_{k_0}$; the $q$ used for the
$j$'th
factor $S_{\alpha_j}$ is $q_j$.  For fixed $p$ and $q$'s  the product rule applies
for computing the commutator of the  product and the analogues
of \eqref{eq:23} and \eqref{eq:24} hold true. Now we can take the
limit $p\to \infty$. We can plug the modified expressions of \eqref{eq:8}--\eqref{eq:22} into
 (modified) \eqref{eq:23} and \eqref{eq:24}. Actually \eqref{eq:8} is
 the same, but \eqref{eq:11} and \eqref{eq:22} are changed as
\begin{subequations}
\begin{align}
  \label{eq:11b}
 \i\big[H,N^{\frac12}_{\i q_j}\big]&=-I_{\i q_j}(N^{\frac12})M_1I_{\i q_j}(N^{\frac12}),\\
\i\big[H,A_{q_j}\big]&=I_{q_j}(A)M_2S_1I_{q_j}(A).\label{eq:22b}
\end{align}
\end{subequations} Of course we have a $q$--dependence of the
various  factors of either
$S_1\to N^{1/2}_{\i q_j}$ or $S_0\to A_{q_j}$. Eventually we take the
limits in  the $q$'s done in increasing order starting by taking  $q_1\to
\infty$ and ending by taking  $q_{k_0}\to
\infty$. Before taking these limits  we need to do some further commutation using
Condition~\ref{cond:more-n-regularity}.   For simplicity of presentation we
ignore below in this process commutation with the regularizing factors of
$I_{\i q_j}(N^{1/2})$ or $I_{q_j}(A)$ since in the limit they will
disappear (a manifestation of this occurred also in the proof of
Lemma~\ref{Lemma-StrongLimits}). In other words we proceed now
slightly formally
using  \eqref{eq:23} and \eqref{eq:24} with the plugged in expressions \eqref{eq:8}--\eqref{eq:22}:

From \eqref{eq:8} we obtain that $\|\i[H,I_{\i n}(N)]\|\leq C$ so the contribution from the first term of
\eqref{eq:23} can be estimated (uniformly in $n$) as
\begin{equation}\label{eq:25}
  \|\i[H,I_{\i
  n}(N)]S_{k_0,\ell}^\alpha\psi\|\leq C\|S_{k_0,\ell}^\alpha\psi\| \leq
\tC.
\end{equation}
As for the contribution from \eqref{eq:11}  we compute
\begin{equation*}
-I_{\i  n}(N)\big(\prod_{j=m+1}^{k_0}S_{\alpha_j}\big)M_1\big(\prod_{j=1}^{m-1}S_{\alpha_j}\big)
=\tT_1\big (N^{\frac12}\prod_{\stackrel{1\leq j \leq k_0}{j\neq m}}S_{\alpha_j}\big)+\tT_2,
\end{equation*}
where
\[
\tT_1=-I_{\i n}(N)M_1N^{-\frac12}.
\]
Here $\tT_2$ is given by repeated commutation using
Condition~\ref{cond:more-n-regularity}. We apply this identity to the bound state
$\psi$. Since $\|\tT_1\|\leq C$ the induction
hypothesis gives  similar bounds  as \eqref{eq:25} for the contribution
from \eqref{eq:11}.

It remains to look at  the contribution from
\eqref{eq:22}: We commute the factor $M_2$ to the left and get
similarly
\begin{align*}
& I_{\i n}(N)\big(\prod_{j=m+1}^{k_0}S_{\alpha_j}\big)M_2S_1\big(\prod_{j=1}^{m-1}S_{\alpha_j}\big)\\
& \ \ =\tT_1 N^{\frac12}I_{\i n}(N)
\big(\prod_{j=m+1}^{k_0}S_{\alpha_j}\big) S_1 \big(\prod_{j=1}^{m-1}S_{\alpha_j}\big)+\tT_2,
\end{align*}
where
\[
\tT_1=I_{\i n}(N)M_2(N^{\frac12}I_{\i n}(N)\big)^{-1}.
\]
As before $\|\tT_1\|\leq C$ (here we use that $H'$ is
$N$--bounded) and the contribution from
$\tT_2$ is treated by using Condition~\ref{cond:more-n-regularity} and the induction
hypothesis. Consequently we get for $\ell\leq k_0-1$ the total bound
\begin{equation}
  \label{eq:28}
  \|HI_{\i n}(N)S_{k_0,\ell}^\alpha\psi\|\leq \tC_1\sqrt{g(n,\ell+1)}+\tC_2,
\end{equation} where $\tC_1$ and $\tC_2$ are independent of
$n$, and for $\ell= k_0$ this bound without the first term to the right.

\noindent {\bf Step II:} Bounding $\|AI_{\i n}(N)S_{k_0,\ell}^\alpha\psi\|$.
We claim that (recall   $\ell\geq 1$)
\begin{equation}
  \label{eq:28b}
  \|AI_{\i n}(N)S_{k_0,\ell}^\alpha\psi\|\leq \tC_3\sqrt{g(n,\ell-1)}+\tC_4,
\end{equation}
where $\tC_3$ and $\tC_4$ are independent of $n$.

To prove  \eqref{eq:28b}
 we observe that it suffices  by the induction
hypothesis  to bound $\|I_{\i n}(N)A \linebreak[1] S_{k_0,\ell}^\alpha\psi\|$.
Since $\ell\geq 1$ there is a nearest
factor of $N^{1/2}$ in the product $S_{k_0,\ell}^\alpha$ that we move to
the left in front of the factor $A$:
\[
I_{\i n}(N)AS_{k_0,\ell}^\alpha=N^{\frac12}I_{\i n}(N)AS_{k_0-1,\ell-1}^\beta+T.
\]
We apply this identity to the bound state
$\psi$. The contribution from
$T$ is treated by using \eqref{eq:2} and the induction
hypothesis. This proves \eqref{eq:28b}.

\noindent {\bf Step III:} We repeat Step III of the proof of
Theorem~\ref{thm:mainresultk=1} using now the proven estimates \eqref{eq:28}
and \eqref{eq:28b} to bound any term of $g(n,\ell)$  for $\ell\geq 1$. In
combination with \eqref{eq:7}  these bounds yield  \eqref{eq:6} with
\begin{equation*}
  K_1=2C_2\tC_1\tC_3(2^{k_0}-1)+1;
\end{equation*}
 here the constant $C_2$ comes from \eqref{eq:virial} while $\tC_1$ and $\tC_3$ come from \eqref{eq:28}
and \eqref{eq:28b}, respectively. Notice that the cardinality of set
$\{0,1\}^{k_0}$ is $2^{k_0}$, so the factor $2^{k_0}-1$ arises by
counting only those indices $\alpha\in \{0,1\}^{k_0}$ with $\sum \alpha_{j}\geq 1$.
\end{proof}

\begin{cor}\label{cor:more-n-regularity} Suppose the conditions of
Corollary~\ref{cor:assumpt-stat-regul} and for $k_0\geq 2$ also Condition~\ref{cond:more-n-regularity}.
Let $\psi\in\cD(N^{1/2})$ be a bound state
$(H-\lambda)\psi=0$ (with $\lambda$ as in Condition~\ref{cond:mourre}).
Then $\psi\in \cD(N^{(k_0+1)/2})$.
\end{cor}


\section{A Class of Massless Linearly Coupled Models}\label{sec_models}


In this section we introduce a class of massless linearly coupled Hamiltonians,
sometimes referred to as Pauli-Fierz Hamiltonians \cite{BD,DG,DJ,GGM}. The bulk of this section
is spent on checking that an expanded version of the Hamiltonian does indeed
satisfy the abstract assumptions of Section~\ref{Sec-AbstractResults}.
In Subsection~\ref{subsec-NelsonApp} we verify that the Nelson model described in Subsection~\ref{subsec-NelsonIntro}
is indeed an example of the type of models discussed here.

 \subsection{The Model and the Result}\label{Sec-PF}

Consider the Hilbert space $\cH_\PF = \cK \otimes \Gamma (\gothh)$,
where $\cK$ is the Hilbert space for a ``small'' quantum system, and $\Gamma(\mathfrak{h} )$
is the symmetric Fock space over $\gothh = L^2( \RR^d , {\rm d}k )$,
describing a field of massless scalar bosons. The Pauli-Fierz Hamiltonian $H^\PF_v$ acting on
$\cH_\PF$ is defined by
\begin{equation}
H^\PF_v = K \otimes \one_{\Gamma( \gothh )} +
\one_\cK \otimes {\rm d} \Gamma( |k| ) + \phi(v),
\end{equation}
where $K$ is a Hamiltonian on $\cK$ describing the dynamics of the small system.
We assume that $K$ is bounded from below, and for convenience we require furthermore that
\[
K\geq 0.
\]
The term ${\rm d}\Gamma(|k|)$ is the second quantization of the operator of multiplication by
$|k|$, and $\phi(v) = ( a^*(v) + a(v) )/ \sqrt{2}$.
The form factor $v$ is an operator from $\cK$ to
$\cK \otimes \gothh$, and $a^*(v)$, $a(v)$
are the usual creation and annihilation operators associated to $v$.
See \cite{BD,GGM}.
The hypotheses we make are slightly stronger than the ones considered in \cite{GGM}.
The first one, Hypothesis ${\bf (H0)}$, expresses  the assumption that
the small system is  confined:
\begin{itemize}
\item[{\bf (H0)}] $(K+\i)^{-1}$ is compact on $\cK$.
\end{itemize}
Let $0\leq\tau<1/2$ be fixed. We will introduce a class of
interactions which increase with $\tau$.
In order to formulate our assumption on the form factor $v$
we introduce the subspace $\cO_\tau$ of $\cB( \cD( K^{\tau} ) ; \cK \otimes \gothh )$
consisting of those operators which extend by continuity from $\cD(K^\tau)$ to an element
of $\cB( \cK ; \cD( K^{\tau} )^* \otimes \gothh )$. In other words
\begin{align*}
\cO_\tau & := \big\{v\in\cB( \cD( K^{\tau} ) ; \cK \otimes \gothh )\, \big|\\
&  \ \ \ \ \ \exists C>0,\,
\forall \psi\in \cD(K^\tau): \
\|[(K+1)^{-\tau}\otimes\one_{\gothh}] v \psi\|_{\cK\otimes\gothh}\leq C\|\psi\|_{\cK}\big\}.
\end{align*}
We also write $v$ for the extension. It is natural to introduce a norm on $\cO_\tau$ by
\[
\|v\|_{\tau} = \|v(K+1)^{-\tau}\|_{\cB(\cK;\cK\otimes\gothh)} +
\|[(K+1)^{-\tau}\otimes\one_{\gothh}] v\|_{\cB(\cK;\cK\otimes\gothh)}.
\]

Our first assumption on the form factor interaction is the following:
\begin{itemize}
\item[{\bf (I1)}] $v, [\one_\cK\otimes |k|^{-1/2}]v\in \cO_\tau$.
\end{itemize}
It is proved in \cite{GGM} that if $\bf{ (I1) }$ holds, $H^\PF_v$
is self-adjoint with domain $\cD(H^\PF_v) = \cD(K \otimes \one_{\Gamma( \gothh )}
+ \one_\cK \otimes {\rm d} \Gamma(|k|))$.

The unitary operator $T: L^2(\RR^d)\to L^2(\RR_+)\otimes L^2(S^{d-1})=:\tilde{\gothh}$ defined by
$(T u)(\omega,\theta) = \omega^{(d-1)/2}u(\omega\theta)$ allows us to pass to polar coordinates.
Lifting $T$ to the full Hilbert space as $\one_\cK\otimes \Gamma(T)$ gives a unitary map
from $\cH_\PF$ to  $\tcH_\PF := \cK \otimes \Gamma( \tilde{\gothh} )$.
The Hamiltonian $H^\PF_v$ is unitarily equivalent to
\begin{equation}
\tH^\PF_v := K \otimes \one_{\Gamma( \tilde{\gothh})} + \one_{\cK} \otimes \d\Gamma(\omega) + \phi ( \tilde{v} ),
\end{equation}
where $\tilde{v} = [\one_\cK\otimes T] v \in \cB ( \cK ; \cK \otimes \tilde{\gothh} )$.

In polar coordinates the space of couplings consists of operators of the form $[\one_\cK\otimes T]v : \cK\to \cK\otimes\tilde{\gothh}$,
where $v\in \cO_\tau$. We write $\tcO_\tau = [\one_\cK\otimes T]\cO_\tau$ and equip
it with the obvious norm $\|\tilde{v}\|_\tau^{\tilde{\,}} =\|[\one_\cK\otimes T^*] \tilde{v}\|_\tau$.
Observe $\|\tv\|_\tau^{\tilde{\,}}= \|v\|_\tau$, when $\tilde{v} = [\one_\cK\otimes T] v$.

Let $d$ be as in \eqref{d1} and \eqref{d2}. We recall that $d$ expresses the
least amount of infrared regularization carried by a $v$ satisfying {\bf (I2)} below.
The following further assumptions on the interaction are made:
\begin{itemize}
\item[{\bf (I2)}]  The following holds
\begin{align*}
& [\one_\cK\otimes ( 1+\omega^{-1/2} )\omega^{-1} d(\omega)\otimes \one_{L^2(S^{d-1})}] \tilde{v}\in\tcO_\tau,\\
& [\one_{\cK}\otimes ( 1+\omega^{-1/2} ) d(\omega) \partial_\omega\otimes \one_{L^2(S^{d-1})}] \tilde{v} \in \tcO_\tau,
\end{align*}
\end{itemize}
\begin{itemize}
\item[{\bf (I3)}] $[\one_\cK\otimes \partial^2_\omega\otimes \one_{L^2(S^{d-1})}] \tilde{v} \in \cB(\cD(K^\tau);\cK\otimes\tilde{\gothh})$.
\end{itemize}

In this paper we need an additional assumption compared to \cite{GGM}.
For bounded $K$, it is implied by {\bf(I1)}. Its presence is motivated by a desire to deal
effectively with infrared singularities.
\begin{itemize}
\item[{\bf (I4)}] The form $[K\otimes \one_{\tilde{\gothh}}]\tv- \tv K$
extends from $[\cD(K)\otimes \tilde{\gothh}]\times \cD(K)$ to an element of $\tcO_\frac12$.
\end{itemize}
Here $\tcO_\frac12$ is defined as $\tcO_\tau$.
Supposing {\bf (I1)}, the statement above is meaningful. See also Remark~\ref{Rem-RelaxI4} below.
\begin{remark} We remark that for separable Hilbert spaces $\cK_1$ and $\cK_2$ there are two
natural subspaces of $\cB(\cK_1;\cK_2\otimes \gothh)$. Namely
\begin{align*}
L^2\big(\RR^d;\cB(\cK_1;\cK_2)\big) & = \Big\{ v:\RR^d\to \cB(\cK_1;\cK_2) \, \Big| \,
\int_{\RR^d} \|v(k)\|^2_{\cB(\cK_1;\cK_2)}\d k<\infty\Big\}\\
L^2_\mathrm{w}\big(\RR^d;\cB(\cK_1;\cK_2)\big) & = \Big\{ v:\RR^d\to \cB(\cK_1;\cK_2) \, \Big| \,
\sup_{\|\psi\|_{1}\leq 1 }\int_{\RR^d} \|v(k)\psi\|_{2}^2\d k<\infty\Big\}.
\end{align*}
The functions $v$ should be weakly measurable, to ensure that $\|v(k)\|_{\cB(\cK_1,\cK_2)}$ and $\|v(k)\psi\|_2$ are measurable.
Here $\|\cdot\|_j$ denotes the norm on $\cK_j$.
We have the obvious inclusions
\[
L^2\big(\RR^d:\cB(\cK_1;\cK_2)\big) \subseteq L^2_\mathrm{w}\big(\RR^d;\cB(\cK_1;\cK_2)\big) \subseteq
\cB(\cK_1;\cK_2\otimes \mathfrak{h}).
\]
The first inclusion is a contraction and the second an isometry.
Both inclusions are strict as exemplified by choosing $\cK_1=\cK_2 = \mathfrak{h}= L^2(\RR^3)$ and
$v(k) = \e^{-|x-k|}$ for the first inclusion and $v(k) = |x-k|^{-1}\e^{-|x-k|}$ for the second.
(In \cite[Subsection ~2.16]{DG} and \cite[Subsection~3.4]{GGM} the second inclusion is claimed to be an equality.)
\end{remark}

We denote by $\cI_{\PF}(d)$ the vector space of interactions $v$ satisfying {\bf (I1)}--{\bf (I4)}
and turn it into a normed vector space by equipping it (in polar coordinates) with the norm
\begin{align}\label{normPF}
\nonumber \|v\|_{\PF} := &
\ \big\|[\one_\cK \otimes (1+\omega^{-3/2}d(\omega))\otimes \one_{L^2(S^{d-1})}]\tv \big\|_\tau^{\tilde{\,}}\\
\nonumber & + \big\|\one_\cK \otimes (1+\omega^{-1})d(\omega)\partial_\omega)\otimes \one_{L^2(S^{d-1})}]\tv\big\|_\tau^{\tilde{\,}}\\
\nonumber & + \big\|[(K+1)^{-1/2}\otimes \partial_\omega^2\otimes \one_{L^2(S^{d-1})}] \tv\big\|_{\cB(\cK;\cK\otimes\tilde{\gothh})}\\
& + \big\|[K\otimes\one_\gothh]\tv - \tv K \big\|_\frac12^{\tilde{\,}},
\end{align}
 For any $v_0\in\cI_\PF(d)$ and $r>0$ write
\begin{equation}\label{PFBall}
\cB_r(v_0) = \big\{v\in \cI_\PF(d) \, \big| \, \|v-v_0\|_\PF\leq r\big\}
\end{equation}
for the closed ball in $\cI_\PF(d)$ with radius $r$ around $v_0$.

Let us recall the definition of the conjugate operator on $\tcH_\PF$ used in \cite{GGM}.
Let $\chi \in C_0^\infty( [0 , \infty) )$ be such that $\chi( \omega )=0$ if
$\omega \ge 1$ and $\chi( \omega )=1$ if $\omega \le 1/2$. For $0<\delta\le 1/2$,
the function $m_\delta \in C^\infty( [0,\infty) )$ is defined by
\begin{equation*}
m_\delta( \omega ) = \chi( \frac{ \omega }{ \delta } ) d( \delta ) + (1-\chi)( \frac{ \omega }{ \delta } ) d( \omega ),
\end{equation*}
On $\tilde{\gothh}$, the operator $\tilde{a}_\delta$ is defined in the same way as in \cite{GGM}, that is
\begin{equation}\label{adeltaPF}
\tilde{a}_\delta := \i m_\delta( \omega ) \frac{ \partial }{ \partial \omega }
+ \frac{ \i }{ 2} \frac{ \d m_\delta }{ \d \omega }( \omega ), \quad \cD( \tilde{a}_\delta )
= H_0^1( \mathbb{R}^+ ) \otimes L^2( S^{ d-1}).
\end{equation}
Its adjoint is given by
\begin{equation}
\tilde{a}_\delta^* := \i m_\delta( \omega ) \frac{ \partial }{ \partial \omega }
- \frac{\i}{2} \frac{ \d m_\delta }{ \d \omega }( \omega ), \quad \cD( \tilde{a}_\delta^* )
= H^1( \RR^+ ) \otimes L^2( S^{ d-1}).
\end{equation}
We recall that $H_0^1(\RR^+)$ is the closure of $C_0^\infty((0,\infty))$ in $H^1(\RR^+)$.
The conjugate operator $\tA_\delta$ on $\tcH_\PF$ is defined by
$\tA_\delta := \one_\cK \otimes {\rm d}\Gamma( \tilde{a}_\delta )$.
Going back to $\cH_\PF$ we get $a_\delta= T^{-1}\tilde{a}_\delta T$ and
\[
A_\delta =\d \Gamma(a_\delta) = \big[\one_\cK\otimes \Gamma(T^{-1})\big]\,\tA_\delta \,\big[\one_\cK\otimes \Gamma(T)\big].
\]
The operator $a_\delta$ takes the form \eqref{Eq-adeltaPhys} when written in the original coordinates.

We write $\cN$ for the number operator $\one_\cK\otimes \d\Gamma(\one_{\gothh})$ on $\cH_\PF$.
For $E\in \sigma_{\mathrm{pp}}(H_v^\PF)$, we write $P_v$ for the
corresponding eigenprojection.
Recall from \cite[Theorem~2.4]{GGM} that the range of $P_v$  is finite dimensional
under the assumptions ({\bf H0}), ({\bf I1}) and ({\bf I2}).

\begin{thm}\label{Thm-MainPF} Suppose  {\bf (}${\bf H0}${\bf )}.
Let $v_0\in\cI_\PF(d)$ and $J\subseteq \RR$ be a compact interval. There exists $0<\delta_0\leq 1/2$
such that for all $0<\delta\leq \delta_0$ the following holds: There exist $\gamma>0$ and  $C>0$ such that for any $v\in \cB_\gamma(v_0)$
and $E\in \sigma_\mathrm{pp}(H^\PF_v)\cap J$ we have
\[
P_v:\cH_\PF \to \cD\big(\cN^\frac12 A_\delta\big)\cap \cD\big(A_\delta\cN^\frac12\big)\cap \cD\big(\cN\big)
\]
and
\[
\big\|\cN^\frac12 A_\delta P_v\big\| + \big\|A_\delta\cN^\frac12 P_v\big\| + \big\|\cN P_v\big\| \leq C.
\]
\end{thm}

Unfortunately we cannot employ our theory directly to conclude the above theorem, due to $A_\delta$ not
being self-adjoint. Instead we use a trick of passing to an 'expanded' model, for which we can
use our abstract theory. The theorem above will then be a consequence of a corresponding theorem
in the expanded picture.

\begin{remark}
Under the hypotheses of Theorem~\ref{Thm-MainPF}, we also have that $P_v :\cH_\PF \to \cD( A_\delta^*\cN^{1/2})
\linebreak[1] \cap \cD(\cN^{1/2} A_\delta^*)$. This follows from  $A_\delta \subseteq A_\delta^*$.
In particular this implies that $P_v A_\delta$ extends from $\cD(A_\delta)$ to a bounded operator on $\cH_\PF$.
Similar statements hold also  for $P_v A_\delta \cN^{1/2}$ and $P_v \cN^{1/2}A_\delta$.
\end{remark}


\subsection{Application to the Nelson Model}\label{subsec-NelsonApp}

In this subsection we check the conditions {\bf (H0)} and {\bf (I1)}--{\bf (I4)} for the
Nelson model introduced in the introduction. After possibly adding a constant to $W$, we
can assume that $K\geq 0$. See \eqref{K} and {\bf (W0)}.

We begin by remarking that it follows from {\bf (W0)} and {\bf (V0)} that
\begin{align}
\label{xalpha-K-bnd}
& |x|^\alpha (K+1)^{-\frac12} \in \cB(\cK) \\
\label{p-K-bnd}
& |p|(K+1)^{-\frac12} \in\cB(\cK).
\end{align}
Here $\alpha>2$ is coming from {\bf (W0)}, $|x| = |x_1|+\cdots |x_P|$ and $|p| = |p_1|+\cdots +|p_P|$, where $p_\ell = -\i\nabla_{x_\ell}$.
These bounds imply in particular {\bf (H0)}.

Let $\Psi_\Nel: \cI_\Nel(d)\to \cB(\cK;\cK\otimes\gothh)$ be defined by
\[
\Psi_\Nel(\rho) = \sum_{\ell = 1}^P \e^{-\i k\cdot x_\ell}\rho.
\]
Clearly $\Psi_\Nel$ is a linear map and $\phi(\Psi_\Nel(\rho)) = I_\rho(x)$ such that
\[
H_\rho^\Nel = K\otimes\one_\cF + \one_\cK\otimes \d \Gamma(|k|) + \phi\big(\Psi_N(\rho)\big),
\]
is a Pauli-Fierz Hamiltonian, cf. \eqref{Irho} and \eqref{HrhoNel}.
Verifying the conditions {\bf (I1)}--{\bf (I4)} will be achieved if we can show that $\Psi_\Nel$ is a bounded operator from $\cI_\Nel(d)$ to $\cI_\PF(d)$. This also implies that results valid uniformly for $v$ in a ball in $\cI_\PF(d)$
will translate into results holding uniformly for $\rho$ in a sufficiently small ball in $\cI_\Nel(d)$.
See Remark~\ref{remark:assumpt-stat-regul}~\ref{item:1bll}).

That the terms in the norm $\|\Psi_\Nel(\rho)\|_\PF$, cf. \eqref{normPF}, pertaining to the conditions {\bf (I1)}--{\bf (I3)} can be
bounded by $\|\cdot\|_\Nel$ (or rather terms in $\|\cdot\|_\Nel$ pertaining to {\bf ($\rho$1)}--{\bf ($\rho$3)}), follows
as in \cite{GGM} after we have checked that $|x|^2(K+1)^{-\tau}$ is bounded for some positive $\tau<1/2$.

To produce such a $\tau$ we invoke Hadamard's three-line theorem.
Consider the function $z\to |x|^{-\i \alpha z} (K+1)^{\i z/2}\in \cB(\cK)$.
Observe that this function is bounded when $\Im{z}=0$ or $\Im{z}=1$, cf. \eqref{xalpha-K-bnd}.
It now follows, cf. \cite{RS}, that $|x|^{s\alpha}(K+1)^{-s/2}$ is bounded for $0\leq s\leq 1$.
Choosing $s= 2/\alpha$ implies the desired bound with $\tau = \alpha^{-1}< 1/2$.
This will be the $\tau$ used in the conditions {\bf (I1)}--{\bf (I3)}.

It remains to verify {\bf (I4)}. For this we compute
\begin{align}\label{I4comp}
\nonumber [K\otimes\one_{\gothh}]\e^{-\i k\cdot x_j}\rho - \e^{-\i k\cdot x_j}\rho K
\nonumber & = -\sum_{\ell=1}^P \Big[\frac1{2m_\ell}\Delta_\ell \, \e^{-\i k\cdot x_j}\rho 
 - \rho \e^{-\i k\cdot x_j} \frac1{2m_\ell}\Delta_\ell\Big]\\
\nonumber & = -\Big[\frac1{2m_j}\Delta_j \, \e^{-\i k\cdot x_j}\rho 
 - \rho \e^{-\i k\cdot x_j} \frac1{2m_j}\Delta_j\Big]\\
\nonumber & = \frac{\e^{-\i k\cdot x_j}}{2m_\ell}\big[-2k\cdot p_j + k^2\big]\rho\\
& =  \big[-2k\cdot p_j - k^2\big]\rho\frac{\e^{-\i k\cdot x_j}}{2m_j}
\end{align}
From this computation and \eqref{p-K-bnd} we conclude that $[K\otimes\one]\Psi_\Nel(\rho)- \Psi_\Nel(\rho)K\in\cO_{1/2}$ as required by {\bf (I4)}
and the $\|\cdot\|_{1/2}$-norm of the difference is bounded by a constant times $\|\rho\|_\Nel$.
Here we need the term in $\|\cdot\|_\Nel$ coming from {\bf ($\rho$4)}.

We can thus conclude Theorem~\ref{Main-Nelson} from Theorem~\ref{Thm-MainPF}.

It remains to discuss the Nelson model after a Pauli-Fierz transformation.
We recall that we have two transformations to consider, one giving rise
to $H_\rho^{\Nel'}$ and one to $H_\rho^{\Nel''}$. See \eqref{HrhoNelPrime} and \eqref{HrhoNelDPrime}.
To identify these Hamiltonians as Pauli-Fierz Hamiltonians, we introduce a
linear map $\Psi_\Nel' : \cI_\Nel'(d) \to \cB(\cK;\cK\otimes\gothh)$ by
\[
\Psi_\Nel'(\rho) = \sum_{\ell=1}^P (e^{-\i k\cdot x_\ell} - 1)\rho.
\]
With this notation we find for $\rho\in \cI_\Nel'(d)$
\[
H_\rho^{\Nel'} = K_\rho \otimes \one_{\Gamma(\gothh)} + \one_\cK\otimes \d\Gamma(|k|) + \phi\big(\Psi_\Nel'(\rho)\big)
\]
and, specializing to  $\rho = \rho_0 + \rho_1$ with $\rho_0\in\cI_\Nel'(d)$ and $\rho_1\in\cI_\Nel(d)$, 
\[
H_\rho^{\Nel''} = \big(K_{\rho_0} - \sum_{\ell=1}^P v_{\rho_0,\rho_1}(x_\ell)\big) \otimes \one_{\Gamma(\gothh)} + \one_\cK\otimes \d\Gamma(|k|) 
+ \phi\big(\Psi_\Nel'(\rho_0)+ \Psi_\Nel(\rho_1)\big).
\]
See \eqref{Krho} for $K_\rho$ and  \eqref{vrho0rho1} for $v_{\rho_0,\rho_1}$.

In order to apply Theorem~\ref{Thm-MainPF} one should first observe that
$\Psi_\Nel'$ is a bounded map from $\cI_\Nel'(d)$ to $\cI_\PF(d)$. 
We leave it to the reader to establish this following the arguments in \cite{GGM},
using the key estimate \eqref{PFtransbnd}. As for {\bf (I4)}, observe that the extra  $-\rho$ from
$(e^{-\i k\cdot x_j} -1)\rho$ drops out when repeating \eqref{I4comp}
for $\Psi_\Nel'(\rho)$. In particular we do not need \eqref{PFtransbnd} for {\bf (I4)}.

Observe that for both the transformed Hamiltonians, the Hamiltonian for the confined
quantum system $K$ is altered by the transformation, to obtain e.g. $K_\rho$ in the case of $H_\Nel'$. A priori the norm $\|\cdot\|_\PF$
is however defined in terms of the operator $K$, and this definition we retain.

However, when verifying the Mourre estimate in Subsection~\ref{Subsec-ME}
and our abstract assumptions for Pauli-Fierz Hamiltonians in Subsection~\ref{Subsec-AbstractAss},
we will naturally meet norms with the modified $\rho$-dependent $K$'s, and not the original $K$.
We proceed to argue that the $\|\cdot\|_\PF$ norms arising in this way are equivalent,
locally uniformly in $\rho$, with respect to the appropriate normed
space. Let for $\rho\in\cI_\Nel'(d)$
\[
 B_\rho' = K_\rho -K= -\sum_{\ell=1}^P v_{\rho}(x_\ell) 
 + \frac{P^2}2\int_0^\infty r^{-1}|\tilde{\rho}(r)|^2 \d r \one_\cK
\]
and for $\rho= \rho_0+\rho_1$ as above
\[
B_\rho'' = -\sum_{\ell=1}^P v_{\rho_0,\rho_1}(x_\ell)-\sum_{\ell=1}^P v_{\rho_0}(x_\ell) 
 + \frac{P^2}2\int_0^\infty r^{-1}|\tilde{\rho_0}(r)|^2 \d r \one_\cK.
\]
We observe the bounds 
\[
\|B_\rho'\|\leq C \|\rho\|_\Nel'^2  \ \textup{and} \ 
\|B_\rho''\| \leq C\big(\|\rho_0\|_\Nel'^2 + \|\rho_1\|_\Nel^2\big),
\]
 for some $\rho$-independent constant $C$. 
In particular both $B_\rho'$ and $B_\rho''$ can be bounded locally uniformly in $\rho$,
with respect to the appropriate norm.
By yet another interpolation argument this implies that we can pass between $\|\cdot\|_\PF$
norms defined with either $K$, $K_{\rho}$, or $K_{\rho_0} + B_\rho''$,
and still retain bounds that are locally uniform in $\rho$. 

Finally we note that the above bounds also imply 
that by possibly adding to $W$ a positive constant we still have $K_\rho \geq 0$ and $K_{\rho_0}+B_\rho''\geq 0$
locally in $\rho$. This ensures that {\bf (H0)} is satisfied  also for transformed Nelson Hamiltonians.
In particular we still have e.g. $|x|^2 (K_\rho+1)^{-\tau}$ bounded.

In conclusion, Theorem~\ref{Main-NelsonPrime} also follows from Theorem~\ref{Thm-MainPF}.


\subsection{Expanded Objects}

Let us now define the expanded operator $\hH^\rome_v$ on
$\hcH^\rome := \tcH_\PF \otimes \Gamma( \tilde{\gothh} )$ by
\begin{equation}\label{hatHexpanded}
\hH^\rome_v :=
\tH_v^\PF \otimes \one_{\Gamma( \tilde{\gothh} )}- \one_{\tcH_\PF} \otimes \d\Gamma(\hat{h}),
\end{equation}
where $\hat{h}$ is the operator of multiplication by
\begin{equation}\label{Eq-hath}
\hat{h}( \omega ) = \e^\omega - 1 - \frac{\omega^2}2.
\end{equation}
From the bound $\omega \leq \frac12 + \omega^2/2$ we find that for $\omega\ge 0$
\begin{equation}\label{eq:h'_ge_h}
\frac{\d}{\d \omega}\hat{h} (\omega) \geq \hat{h}(\omega)+ \frac12.
\end{equation}

Since $ L^2( \RR^+ ) \oplus L^2( \RR^+ )\simeq L^2( \RR )$,
it is known (see e.g. \cite{DJ}) that there exists a unitary operator
\begin{equation}
\cU :\Gamma( \tilde{\gothh} ) \otimes \Gamma( \tilde{\gothh} ) \to  \Gamma( \gothh^\rome ),
\end{equation}
where $\gothh^\rome := L^2(\RR) \otimes L^2( S^{d-1} )$.
On $\cK \otimes \Gamma( \tilde{\gothh})\otimes \Gamma(\tilde{\gothh})$, the unitary operator $\one_\cK \otimes \cU$
is still denoted by $\cU$. It maps into $\cH^\rome = \cK\otimes \Gamma(\gothh^\rome)$. In this representation, the operator $\hH^\rome_v$ is unitary equivalent
to the `expanded Pauli-Fierz Hamiltonian' $H^\rome_v$ defined as an operator on $\cH^\rome$ by
\begin{equation}
H^\rome_v := \cU \hH^\rome_v \cU^{-1} = K \otimes \one_{\Gamma( \gothh^\rome )}
+ \one_\cK \otimes \d\Gamma(h) + \phi( v^\rome ),
\end{equation}
where $v^\rome \in \cB( \cK , \cK \otimes \gothh^\rome )$, and  $v^\rome$ and $h$ are
defined by
\begin{equation}\label{Eq-h-and-ve}
h( \omega ) :=
\begin{cases}
 \omega  &\text{if } \omega \ge 0, \\
- \hat{h}( - \omega ) &\text{if } \omega \le 0,
\end{cases}
\quad\qquad v^\rome( \omega ) :=
\begin{cases}
\tilde{v}( \omega )  &\text{if } \omega \ge 0, \\
0 &\text{if } \omega \le 0.
\end{cases}
\end{equation}
Note that $h \in C^2( \RR )$. The idea of expanding the Hilbert space in the above fashion
has been used previously in \cite{DJ,DJ2,G,JP1}. Our choice of expansion
for the boson dispersion relation to the unphysical negative $\omega$ appears to be new.
Previous implementations of the expansion all used the obvious linear expansion $h(\omega)= \omega$.

We remark that if $\cC_K\subseteq \cK$ is a core for $K$, $\cC\subseteq\Gamma(\tilde{\gothh})$ is a core for $\d\Gamma(\omega)$,
then the algebraic tensor product $\cC_K\otimes\cC$ is a core for $\tH_0^\PF$, hence for $\tH_v^\PF$,
and finally $\cC_K\otimes \cC\otimes\cC$ is a core for $\hH^\rome_v$ for any $v\in\cI_\PF(d)$. The domain
$\cD(H_v^\rome)$ itself may however be $v$ dependent. (The argument for the contrary in \cite[Section~5.2]{DJ} seems wrong.)
We have however set up our analysis such that knowledge of $H_v^\rome$'s domain is not needed.
See also Lemma~\ref{Lemma-ExplDomain} where an intersection domain is computed.

\begin{remark}
We remark that if one is going for higher order results, i.e. $\psi\in\cD(A^{k_0})$ for $k_0\geq 2$,
one should use a different $\hat{h}$. The choice
\[
\hat{h}_{k_0}(\omega) = \e^\omega - 1 -\sum_{\ell=2}^{k_0+1} \frac{\omega^\ell}{\ell !}
\]
will work since the corresponding $h_{k_0}$ is in $C^{k_0+1}(\RR)$ and the  bound
\[
\frac{\d}{\d \omega}\hat{h}_{k_0}(\omega)\geq \frac{\hat{h}_{k_0}(\omega)}{(k_0-1)!} + \frac12
\]
holds for $\omega\geq 0$ and $k_0\geq 1$. For $k_0 = 1$ this reduces to \eqref{eq:h'_ge_h}.
\end{remark}

Before introducing the conjugate operator on $\cH^\rome$ that we shall use,
let $m_\delta^\rome \in C^\infty( \RR )$ be defined by
\begin{equation*}
m_\delta^\rome( \omega ) :=
\begin{cases}
m_\delta( \omega ) &\text{if } \omega \ge 0, \\
d( \delta ) &\text{if } \omega \le 0.
\end{cases}
\end{equation*}
We set
\begin{equation}
a_\delta^\rome := \i m_\delta^\rome( \omega ) \frac{ \partial }{ \partial \omega }
+ \frac{ \i }{ 2} \frac{ \d m_\delta^\rome }{ \d \omega }( \omega ), \quad \cD( a_\delta^\rome )
= H^1( \RR ) \otimes L^2( S^{ d-1}),
\end{equation}
and $A_\delta^\rome := \one_\cK \otimes {\rm d} \Gamma( a_\delta^\rome )$ as an operator on
$\cH^\rome$. Note that both $a_\delta^\rome$ and $A_\delta^\rome$ are self-adjoint.

We can now formulated the expanded version of our regularity theorem

Let
\[
\cN^\rome = \one_{\cK}\otimes \d\Gamma(\one_{\gothh^\rome}) = \cU\big(\cN \otimes\one_{\Gamma(\tilde{\gothh})}
+\one_{\tcH_\PF}\otimes \d\Gamma(\one_{\tilde{\gothh}})\big)\cU^{-1}
\]
denote the expanded number operator. For $E\in \sigma_\mathrm{pp}(H_v^\rome)$ we write $P_v^\rome$ for the
associated eigenprojection.

\begin{thm}\label{Thm-MainPFE} Suppose {\bf (}${\bf H0}${\bf )}.
Let $v_0\in\cI_\PF(d)$ and $J\subseteq \RR$ be a compact interval. There exists a $0<\delta_0\leq 1/2$ such that
for any $0<\delta\leq \delta_0$ the following holds: There exist $\gamma>0$ and $C>0$ such that for any $v\in\cB_\gamma(v_0)$ and
$E\in \sigma_\mathrm{pp}(H_v^\rome)\cap J$ we have
\[
P_v^\rome:\cH^\rome\to \cD\big((\cN^\rome)^\frac12 A^\rome_\delta\big))
\cap \cD\big((A^\rome_\delta (\cN^\rome)^\frac12\big)
\cap \cD\big(\cN^\rome\big)
\]
and
\[
\big\|(\cN^\rome)^\frac12 A^\rome_\delta P_v^\rome\big\| + \big\|A^\rome_\delta(\cN^\rome)^\frac12 P_v^\rome\big\|
+ \big\|\cN^\rome P_v^\rome\big\| \leq C.
\]
\end{thm}

In the next two subsections we verify that our abstract theory applies to the expanded model, but before doing so
we pause to check that Theorem~\ref{Thm-MainPF} does indeed follow from  Theorem~\ref{Thm-MainPFE}.
For that we need a lemma.

Let $W_{\delta,t}$, $t\geq 0$, denote the contraction semigroup on $\tcH_\PF$ generated by $\tA_\delta$.

\begin{lemma}\label{lm:Adelta-Adelta-e}
For any state $\varphi\in\tcH_\PF$ we have for $t\geq 0$
\[
\e^{-\i t A^\rome_\delta} \cU(\varphi\otimes\Omega) = \cU(W_{\delta,t}\varphi\otimes\Omega).
\]
In particular, $\varphi \in \cD(\tA_\delta^k)$ if and only if
$\cU(\varphi\otimes\Omega)\in \cD((A^\rome_\delta)^k)$.
\end{lemma}
\begin{proof}
It suffices to check the identity on a dense set of $\varphi$'s.
Let $\varphi \in \cK\otimes \Gamma_\fin(H^1_0(\RR^+)\otimes L^2(S^{d-1}))\subseteq \cD(\tilde{A}_\delta)$.
Then
$\cU(\varphi\otimes\Omega) \in \cK\otimes \Gamma_\fin(H^1(\RR)\otimes L^2(S^{d-1}))
\subseteq \cD(A_\delta^\rome)$.
The identify now follows by differentiating both sides of the equation and
observing they satisfy the same differential
equation, with the same initial condition. Here we made use of the
equality $A_\delta^\rome\cU(\varphi\otimes\Omega) = \cU(\tA_\delta\varphi\otimes \Omega)$ valid for
 $\varphi \in \cK\otimes \Gamma_\fin(H^1_0(\RR^+)\otimes L^2(S^{d-1}))$.
\end{proof}

\begin{proof}[Proof of Theorem~\ref{Thm-MainPF}]
We only have to recall that bound states of $H_v^\rome$ are precisely states
on the form $\cU(\varphi\otimes\Omega)$, where $\varphi$ is a bound state for $\tH_v^\PF$,
with the same eigenvalue. This implies that eigenprojections for $H_v^\rome$
are on the form $\cU[\tP\otimes |\Omega\ra \la\Omega|]\cU^{-1}$ where $\tP$ is an eigenprojection
for $\tH_v^\PF$. Theorem~\ref{Thm-MainPFE}, together with Lemma~\ref{lm:Adelta-Adelta-e},
now implies Theorem~\ref{Thm-MainPF}.
\end{proof}


\subsection{Mourre Estimates}\label{Subsec-ME}


We begin by establishing a Mourre estimate for $H_v^\PF$ and $A_\delta$ in a form appropriate for use in this paper.
At the end of the subsection we derive a Mourre estimate for $H^\rome_v$ and $A_\delta^\rome$.

Let
\[
M_\delta := \one_\cK \otimes {\rm d}\Gamma( m_\delta ) \ \  \textup{and} \ \
R_\delta =R_\delta(v) := -\phi( \i a_\delta v )
\]
as operators on $\cH_\PF$. Let $H'$ be the closure of
$M_\delta + R_\delta$ with domain $\cD(H_v^\PF) \cap \cD( M_\delta )$.
Recall from \cite{GGM} that $H' = [ H_v^\PF , \i A_\delta ]^0$. Let $f \in C_0^\infty( \RR )$
be such that $0 \le f \le 1$, $f(\lambda) = 1$ if $|\lambda| \le 1/2$ and $f(\lambda) = 0$
if $|\lambda| \ge 1$. In addition we choose $f$ to be monotonously decreasing away from $0$, i.e. $\lambda f^\prime(\lambda)\leq 0$.
 For $E \in \RR$ and $\kappa>0$ we set
\begin{equation*}
f_{E,\kappa}( \lambda ) := f\big( \frac{\lambda-E}{\kappa} \big).
\end{equation*}
The following `Mourre estimate' for $H^\PF_v$ is proved in \cite{GGM}:
\begin{thm} \label{MourreGGM} \emph{\textbf{\cite[Theorem 7.12]{GGM}}}
Assume that Hypotheses {\bf (}${\bf H0}${\bf )}, {\bf (}${\bf I1}${\bf )} and {\bf (}${\bf I2}${\bf )} hold.
Let  $E_0\in \RR$. There exists $\delta_0 \in ]0,1/2]$ such that: For all
$E \le E_0$, $0<\delta \le \delta_0$ and $\varepsilon_0>0$,
there exist $C>0$, $\kappa>0$, and a compact operator $K_0$ on $\cH_\PF$ such that the estimate
\begin{equation}\label{ME-PF}
M_\delta + f_{E,\kappa}( H_v^\PF ) R_\delta f_{E,\kappa}( H_v^\PF )
\ge (1-\varepsilon_0) \one_{\cH_\PF} - C f^\perp_{E,\kappa}( H_v^\PF )^2 - K_0
\end{equation}
holds as a form on $\cD(\cN^{1/2})$.
\end{thm}

The following lemma is just a reformulation of \cite[Proposition 4.1~i),  Lemma 4.7 and Lemma 6.2~iv)]{GGM}.
We leave the proof to the reader.

\begin{lemma}\label{Lem-Pointinv} Let $v_0\in\cI_\PF(d)$. There exists $c_0,c_1,c_2>0$, depending on $v_0$, such that
$H_{v_0}^\PF+c_0\geq 0$ and the following  holds: for all $w\in\cI_\PF(d)$ and $0<\delta\leq 1/2$
\[
\pm \phi(w) \leq c_1\|w\|_\PF (H^\PF_{v_0} + c_0) \ \ \textup{and} \ \ \pm R_\delta(w) \leq c_1\|w\|_\PF (H^\PF_{v_0} + c_0).
\]
\[
\|\phi(w)(H^\PF_{v_0}+\i)^{-1}\|\leq c_2\|w\|_\PF \ \ \textup{and} \ \ \|R_\delta(w)(H^\PF_{v_0}+\i)^{-1}\|\leq c_2\|w\|_\PF.
\]
\end{lemma}

The first step we take is to translate the commutator estimate above into the form
used in this paper, see Condition~\ref{cond:mourre}. In anticipation of the need for local uniformity
of constants, we need to already at this step ensure that $B = C_B \one$ can be chosen
uniformly in $E\in J$, where $J$ is compact interval.

\begin{cor}\label{MourreCor1}
Let $J\subseteq\RR$ be a compact interval and $v_0\in \cI_\PF(d)$.
There exists $\delta_0\in ]0,1/2]$ and $C_B>0$ such that
for any  $E\in  J$, $\epsilon_0>0$ and  $0<\delta<\delta_0$
the following holds.
There exists $\kappa>0$, $C_4>0$ and a compact operator $K_0$ such that
the form inequality
\begin{equation}\label{ME-PF2}
M_\delta + R_\delta(v_0) \geq (1-\epsilon_0)\one_{\cH_\PF} - C_4 f^\perp_{E,\kappa}( H^\PF_{v_0} )^2 - C_B(H^\PF_{v_0}-E) - K_0
\end{equation}
holds on $\cD(\cN^{1/2})\cap \cD(H_{v_0}^\PF)$.
\end{cor}

\begin{proof}
Let $E_0$ be an upper bound for the interval $J$ and take $\delta_0$ to be the one coming from
Theorem~\ref{MourreGGM}, applied with $v=v_0$.

Fix  $E\in J$, $0<\delta<\delta_0$ and $\epsilon_0>0$.  Apply Theorem~\ref{MourreGGM} with $\epsilon_0/2$ in place of
$\epsilon_0$.

Compute as a form on $\cD(H^\PF_{v_0})$
\begin{align*}
R_\delta(v_0)  & =  f_{E,\kappa}(H^\PF_{v_0}) R_\delta(v_0) f_{E,\kappa}(H^\PF_{v_0}) +
f^\perp_{E,\kappa}(H^\PF_{v_0})R_\delta(v_0) f^\perp_{E,\kappa}(H^\PF_{v_0}) \\
& \ \ \  +2\Re\{f_{E,\kappa}(H^\PF_{v_0})R_\delta(v_0) f^\perp_{E,\kappa}(H^\PF_{v_0})\}.
\end{align*}
Using Lemma~\ref{Lem-Pointinv} with $w=v_0$ and abbreviating $C_B=c_1\|v_0\|_\PF$ we estimate
\begin{align*}
& f^\perp_{E,\kappa}(H^\PF_{v_0})R_\delta(v_0) f^\perp_{E,\kappa}(H_{v_0}^\PF)\\
& \ \  \geq - c_1\|v_0\|_\PF (H_{v_0}^\PF+c_0)f^\perp_{E,\kappa}(H^\PF_{v_0})^2\\
& \ \ = -C_B(H^\PF_{v_0}-E)f^\perp_{E,\kappa}(H^\PF_{v_0})^2 - C_B(c_0 + E)f^\perp_{E,\kappa}(H^\PF_{v_0})^2\\
& \ \ \geq -C_B(H^\PF_{v_0}-E) - 3C_B\kappa  - C_B(c_0 + E)f^\perp_{E,\kappa}(H^\PF_{v_0})^2.
\end{align*}

Using Lemma~\ref{Lem-Pointinv} again we get
\begin{align*}
& 2\Re\{f_{E,\kappa}(H^\PF_{v_0})R_\delta(v_0) f^\perp_{E,\kappa}(H^\PF_{v_0})\}\\
& \ \  \geq -\frac{\epsilon_0}4 -
\frac4{\epsilon_0}\|R_\delta(v_0) f_{E,\kappa}(H^\PF_{v_0})\|^2 f^\perp_{E,\kappa}(H^\PF_{v_0})^2\\
& \ \ \geq -\frac{\epsilon_0}4 - \frac{4 c_2^2\|v_0\|_\PF^2 (|E|+\kappa+1)^2}{4\zeta}f^\perp_{E,\kappa}(H^\PF_{v_0})^2.
\end{align*}
Combining the equations above with Theorem~\ref{MourreGGM} yields \eqref{ME-PF2}
with $C_B$ only depending on $v_0$.
\end{proof}

The above corollary  suffices to prove Theorem~\ref{Thm-MainPFE} without local uniformity in $v$ and $E$.

The following lemma is designed to deal with uniformity of estimates
in a small ball of interactions $v$ around a fixed (unperturbed) interaction $v_0$.
Technically it replaces \cite[Lemma~6.2~iv)]{GGM}.

\begin{lemma}\label{Lem-Unifinv} Let $v_0\in \cI_\PF(d)$. There exists $\gamma_0>0$,
$C'_B>0$ and $c_0',c_1',c_2'>0$, only depending on $v_0$, such that
\begin{enumerate}[\quad\normalfont (1)]
\item $\forall v\in \cB_{\gamma_0}(v_0):$ $H^\PF_v\geq -c_0'$.
\item $\forall v\in \cB_{\gamma_0}(v_0):$  $\pm \phi(v) \leq c_1' (H^\PF_v+c_0')$ and $\|\phi(v)(H^\PF_v-\i)^{-1}\|\leq c_2'$.
\item $\forall v\in \cB_{\gamma_0}(v_0)$ and $0<\delta\leq 1/2$:  $\pm R_\delta(v) \leq C'_B(H^\PF_v+c_0')$
 and $\|R_\delta(v)(H^\PF_v-\i )^{-1}\|\leq c_2'$.
\end{enumerate}
\end{lemma}

\begin{proof} Let $v_0\in \cI_\PF(d)$ be given.
Let $C_1(r,v) = \|[\one_\cK \otimes\omega^{-1/2}] \tv (K+r)^{-1/2}\|$, for $v\in \cI_\PF(d)$ and $r>0$.

We begin with (1). Fix $r = r(v_0)\geq 1$ such that
$\sqrt{2}C_1(r,v_0)\leq 1/3$.  This is possible due to {\bf (I1)}.
Using \cite[Proposition 4.1~i)]{GGM} we get
\begin{align*}
H_v^\PF &  = H_0^\PF + \phi(v)  = H_0^\PF +  \phi(v_0) + \phi(v-v_0)\\
& \geq  H_0^\PF -\frac13 (H_0^\PF+r) - \sqrt{2}C_1(1,v-v_0)(H^\PF_0+1)\\
& = \big(1-\frac13 - \sqrt{2} C_1(1,v-v_0)\big)H^\PF_0 - \frac{r}3 - \sqrt{2}C_1(1,v-v_0).
\end{align*}
Using that $\omega^{-1/2} \leq 2/3 + \omega^{-3/2}/3\leq 2/3( 1+ \omega^{-3/2}d(\omega))$
we get $C_1(r,v)\leq 2\|v\|_\PF/3$ for any $v\in \cI_\PF(d)$ and $r\geq 1$. This implies
\[
H^\PF_v \geq \big(\frac23 - \frac{2\sqrt{2}}3\|v-v_0\|_\PF\big) H^\PF_0 -\frac{r}3 - \frac{2\sqrt{2}}3 \|v-v_0\|_\PF.
\]
Observe that the choice $\gamma_0= 1/(2\sqrt{2})$ ensures that we arrive at the bound
\[
H^\PF_v \geq -  \frac{r+1}{3}.
\]
Choose $c_0' = 1+ (r+1)/3$ such that $H_v^\PF + c_0'\geq 1$.
This proves (1).

As for (2) we observe first that $\phi(v) = H^\PF_v - H^\PF_0\leq H^\PF_v$.
Next let $r=r(v_0)$ and $\gamma_0= 1/(2\sqrt{2})$ be as in the proof of (1) and estimate
\[
-\phi(v) = -\phi(v_0) + \phi(v_0-v) \leq \frac13(H^\PF_0+r) + \frac13(H_0^\PF+1) = \frac23 H^\PF_0 + \frac{r+1}3.
\]
Writing $H^\PF_0 = H^\PF_v- \Phi(v)$ we arrive at
\[
-\phi(v) \leq 2H^\PF_v + r+1.
\]
Combining with the choice of $c_0'$  in the proof of (1) now yields the first estimate in (2),
for a sufficiently large $c_1'$.

As for the second part of (2) one can employ \cite[Proposition 4.1~ii)]{GGM} in place of \cite[Proposition 4.1~i)]{GGM}
and argue as above.
This  gives a bound of the desired type for $\gamma_0$ small enough. The choice $\gamma_0=1/8$ works.
Here one should observe that the constants $C_j(r,v)$, $j=0,1,2$, in \cite{GGM} are all related to
the norm $\|\cdot\|_\PF$ by
$C_j(1,v)\leq 2 \|v\|_\PF/3$ as argued above for $C_1$.

The statement in (3) now follows by appealing to \cite[Proposition 4.1~i)]{GGM} again
\begin{align*}
\pm R_\delta(v) & \leq \sqrt{2} C_1(1,[\one_\cK\otimes a_\delta] v)(H^\PF_0+1)\\
& \leq  \sqrt{2} C_1(1,[\one_\cK\otimes a_\delta] v)((c_1'+1)H^\PF_v+ c_1'c_0' + 1).
\end{align*}
From \eqref{adeltaPF} and \eqref{normPF} we conclude the existence of a $C'_B$ for which the first estimate in (3)
is satisfied.

Similarly for the second part of (3), where, as in the discussion of the second part of (2),
one can make  use of \cite[Proposition 4.1~ii)]{GGM}.
\end{proof}

We can now state and prove a commutator estimate that is uniform with respect to
$v$ from a small ball around $v_0$, and $E$ in a compact interval.
Given $v_0$, let $\gamma_0$ denote the radius  coming from Lemma~\ref{Lem-Unifinv}.

\begin{cor}\label{Cor-UnifMEPF}
Let $J\subseteq\RR$ be a compact interval, $v_0\in\cI_\PF(d)$, and $\epsilon_0>0$. There exist a
$\delta_0\in ]0,1/2]$ such that
for any $0<\delta<\delta_0$ the following holds.
There exists $0<\gamma<\gamma_0$, $\kappa>0$, $C_4>0$ and a compact operator $K_0$, with $\gamma$ only depending on $\delta,\epsilon_0,J$ and $v_0$, such that
the form inequality
\begin{equation}\label{ME-PF3}
M_\delta + R_\delta(v) \geq (1-\epsilon_0)\one_{\cH_\PF} - C_4 f^\perp_{E,\kappa}( H^\PF_v )^2\la H^\PF_v\ra - K_0
\end{equation}
holds on $\cD(\cN^{1/2})\cap \cD(H^\PF_v)$, for all $E\in J$ and $v\in \cB_\gamma(v_0)$.
\end{cor}

\begin{remark*} We note that the constant $C_4$ in Corollary~\ref{MourreCor1} can, on inspection
of the proof of \cite[Theorem 7.12]{GGM}, be chosen uniformly in $0<\delta\leq \delta_0$.
Making use of this would allow us to choose $\gamma$ independent of $\delta\leq \delta_0$ here,
which would slightly simplify the exposition. We however choose not to test the readers patience
on this issue. See Step~II in the proof below.
\end{remark*}

\begin{proof} Given $J$, $v_0$ and $\epsilon_0$, let $\gamma_0$ be given by Lemma~\ref{Lem-Unifinv} and let
$C_B>0$ $\delta_0>0$ be the constants coming from
Corollary~\ref{MourreCor1}. For $E\in J$ we apply Corollary~\ref{MourreCor1}, with $\epsilon_0$ replaced by $\epsilon_0/3$,
and get the form estimate
\begin{align}\label{UnifBndInt}
\nonumber  M_\delta + R_\delta(v_0)  & \geq (1-\epsilon_0/3)\one_{\cH_\PF}- C_4(v_0,E) f^\perp_{E,\kappa(v_0,E)}( H^\PF_{v_0} )^2 \\
 & \ \ \  - C_B(H^\PF_{v_0}-E) - K_0(v_0,E).
\end{align}
The constants $C_4$, $\kappa$ and the operator $K_0$ also depend on $\delta$,
but this dependence does not concern us. We can assume that $K_0\geq 0$.
The key observation is that the constants $C_4$ and $\kappa$, and
the operator $K_0$ above can be chosen independently of $E\in J$ and $v\in B_\gamma(v_0)$, for some sufficiently small $\gamma$
which does not depend on $\delta\leq \delta_0$.

We divide the proof of the corollary into three steps, the two first establish the observation mentioned in the previous paragraph.

\noindent{\bf Step I:}
We begin by arguing that $C_4$, $\kappa$ and $K_0$ can be chosen independently of $E\in J$.
By a covering argument it suffices to show that they can be chosen independently of $E'$ in a small neighborhood of $E\in J$.
For the compact error, we remark that one should replace $K_0$ by a
finite sum $K_0(v_0) = K_0(v_0,E_1) + \cdots + K_0(v_0,E_m)$ of non-negative compact operators, which is again compact.

Let $E\in J$ be fixed. Pick $\zeta_1 = \epsilon_0/(6C_B)$ such that for $|E-E'|< \zeta_1$ we have
\begin{equation}\label{zeta1}
C_B E \geq C_B E' - \epsilon_0/6.
\end{equation}
As for the term involving $f_{E,\kappa}^\perp$ we observe that for any self-adjoint operator $S$
we have
\begin{align*}
& f^\perp_{E,\kappa}(S) - f^\perp_{E',\kappa}(S)  = f_{E',\kappa}(S) - f_{E,\kappa}(S)\\
& \ \ \ =\frac1{\pi}\int_\CC (\bar{\partial} \tilde{f})(z)\Big[ \Big(\frac{S-E'}{\kappa}-z\Big)^{-1}-
\Big(\frac{S-E}{\kappa}-z\Big)^{-1}\Big] \d u \d v.
\end{align*}
Here $z = u+\i v$.
Estimating this we find that
\[
\|f^\perp_{E,\kappa}(S) - f^\perp_{E',\kappa}(S)\|\leq C \frac{|E-E'|}{\kappa}.
\]
Writing $a^2-b^2= (a-b)(a+b)$ we observe a similar bound for $f^\perp_{E,\kappa}(S)^2- f^\perp_{E',\kappa}(S)^2$.
Again we conclude that for $\zeta_2= \kappa(v_0,E) \epsilon_0/(6 C C_4(v_0,E))$ we find
that for $|E-E'|<\zeta_2$:
\begin{equation}\label{zeta2}
 - C_4 f^\perp_{E,\kappa}( H^\PF_{v_0} )^2 \geq -C_4(v_0,E) f^\perp_{E',\kappa}(H^\PF_{v_0})^2 - \epsilon_0/6.
\end{equation}

The estimates \eqref{zeta1} and \eqref{zeta2} plus the aforementioned covering argument implies the form estimate
\begin{align}\label{Eq-UnifinE}
\nonumber M_\delta + R_\delta(v_0) & \geq (1-2\epsilon_0/3)\one_{\cH_\PF}  - C_4(v_0) f^\perp_{E,\kappa(v_0)}( H^\PF_v )^2\\
& \ \ \ - C_B(H^\PF_v-E) - K_0(v_0),
\end{align}
for all $E\in J$.

\noindent{\bf Step II:}
Secondly we argue that one can use the same constants $C_4$, $\kappa$, and compact operator $K_0$
for $v\in \cB_\gamma(v_0)$, if $\gamma$ is small enough.

Using Lemma~\ref{Lem-Unifinv} we estimate
\[
R_\delta(v_0) = R_\delta(v) + R_\delta(v_0-v) \leq R_\delta(v) + C_1\|v-v_0\|_\PF (H_v^\PF+C_2).
\]
Writing
\[
C_1\|v-v_0\|_\PF (H_v^\PF+C_2) =  C_1\|v-v_0\|_\PF (H_v^\PF-E)+ C_1\|v-v_0\|_\PF(C_2+E),
\]
We see that choosing $\gamma_1=\gamma_1(\epsilon_0,J,v_0)$ small enough we arrive at the following bound
\begin{equation}\label{gamma1}
R_\delta(v_0)\leq R_\delta(v) + C(H_v^\PF-E) + \frac{\epsilon_0}{9}\one_{\cH_\PF},
\end{equation}
which holds for all $v\in \cB_{\gamma_1}(v_0)$ and $E\in J$.

For the $f^\perp_{E,\kappa}$ contribution we compute
\begin{align*}
& f_{E,\kappa}(H^\PF_{v_0}) - f_{E,\kappa}(H^\PF_{v})\\
 & \ \  = \frac1{\pi}\int_\CC (\bar{\partial}\tilde{f})(z)
\Big[ \Big(\frac{H^\PF_{v_0}-E}{\kappa}-z\Big)^{-1}-
\Big(\frac{H^\PF_v-E'}{\kappa}-z\Big)^{-1}\d u \d v\\
& \ \ = \frac1{\kappa \pi}\int_\CC (\bar{\partial}\tilde{f})(z)
 \Big(\frac{H^\PF_{v}-E}{\kappa}-z\Big)^{-1}\phi(v-v_0)
\Big(\frac{H^\PF_{v_0}-E'}{\kappa}-z\Big)^{-1}\d u \d v.
\end{align*}
From Lemma~\ref{Lem-Pointinv} and the representation formula above we find that
\[
\|f^\perp_{E,\kappa}(H^\PF_{v_0})^2 - f^\perp_{E,\kappa}(H^\PF_{v})^2\| \leq C \|v-v_0\|_\PF.
\]
uniformly in $E\in J$.
Arguing as above we thus find a $\gamma_2=\gamma_2(\epsilon_0,J,v_0,\delta)>0$ such that
\begin{equation}\label{gamma2}
-C_4(v_0)f^\perp_{E,\kappa}(H^\PF_{v_0})^2 \geq - C_4(v_0)f^\perp_{E,\kappa}(H^\PF_v)^2 - \frac{\epsilon_0}{9}\one_{\cH_\PF}
\end{equation}
for all $v\in \cB_{\gamma_2}(v_0)$. This is where the $\delta$-dependence enters into the choice of $\gamma$ through $C_4$.
See the remark to the corollary.

Using Lemma~\ref{Lem-Unifinv} we also get a $\gamma_3=\gamma_3(\epsilon_0,v_0)>0$ such that
\begin{equation}\label{gamma3}
-C_B(H^\PF_{v_0}-E) \geq -C_B(H^\PF_{v}-E) - \frac{\epsilon_0}{9}\one_{\cH_\PF},
\end{equation}
for all $v\in \cB_{\gamma_3}(v_0)$.

Combining \eqref{Eq-UnifinE} with \eqref{gamma1}--\eqref{gamma3} we conclude that the estimate \eqref{UnifBndInt}
holds with the same $C_4$, $\kappa$ and $K_0$, for all $E\in J$ and $v\in B_\gamma(v_0)$,
with $\gamma=\min\{\gamma_1,\gamma_2,\gamma_3\}$ only depending on $\epsilon_0,J,v_0$ and $\delta$.

\noindent{\bf Step III:} To conclude the proof we let $\gamma$, $C_4$, $\kappa$ and $K_0$ be fixed by Steps~I and~II.
Pick $\kappa^\prime$ smaller than $\kappa$ such that
$\kappa^\prime C_B(1+\max_{E\in J}|E|)|\leq \epsilon_0$. The Corollary now follows from \eqref{UnifBndInt}
and the estimate
\[
-C_B(H_v^\PF-E) \geq -C_B(1+\max_{E\in J}|E|)f_{E,\kappa}^\perp(H_v^\PF)^2\la H_v^\PF\ra.
\]
Observe that \eqref{UnifBndInt} holds with $\kappa$ replaced by $\kappa^\prime$ as well.
\end{proof}

The corresponding objects in the expanded Hilbert space are defined as follows:
We set
\[
M_\delta^\rome := \one_\cK \otimes {\rm d}\Gamma( m_\delta^\rome h' )
\ \ \textup{and} \ \  R_\delta^\rome = R_\delta^\rome(v) := - \phi( \i a_\delta^\rome v^\rome ).
\]
 Note that
\begin{equation}
\cU^{-1} M_\delta^\rome ~ \cU = M_\delta \otimes \one_{\Gamma( \tilde{\gothh} )}
+ \one_\cH \otimes \hM_\delta, \quad \cU^{-1} R_\delta^\rome ~ \cU
= R_\delta \otimes \one_{\Gamma( \tilde{\gothh} )},
\end{equation}
where $\hM_\delta :=  {\rm d}\Gamma( d(\delta) \hat{h}' )$
as an operator on $\Gamma( \tilde{\gothh} )$. From \eqref{eq:h'_ge_h}, we get
\begin{equation}\label{eq:M-geH-}
\hM_\delta \ge d(\delta) \left [\d\Gamma(\hat{h}) + \tfrac12 \cN \right ],
\end{equation}
The Mourre estimate for $H^\rome_v$ is stated in the following theorem.

\begin{thm}\label{ME-PFE}
Assume that Hypotheses {\bf (}${\bf H0}${\bf )}, {\bf (}${\bf I1}${\bf )} and
{\bf (}${\bf I2}${\bf )} hold. Let $v_0\in\cI_\PF(d)$, $J$ a compact interval, and $\epsilon_0>0$.
There exists $\delta_0 \in ]0,1/2]$ such that for all
$0<\delta \le \delta_0$,
there exist $0<\gamma<\gamma_0$, $C_4>0$, $\kappa>0$,
and a compact operator $K_0$ on $\cH^\rome$ such that
\begin{equation}
M_\delta^\rome + R_\delta^\rome
\ge (1-\epsilon_0) \one_{ \cH^\rome } - C f^\perp_{E,\kappa}( H^\rome_v )^2\la H^\rome_v\ra -  K_0
\end{equation}
for all $E \in J$ and $v\in \cB_\gamma(v_0)$, as a form on $\cD((M_\delta^\rome)^{1/2})\cap \cD(H^\rome_v)$.
\end{thm}
\begin{remark*} As in Corollary~\ref{Cor-UnifMEPF} , the constant $\gamma$
can be chosen to only depend on $\epsilon_0,J,v_0$ and $\delta$, and as in the associated remark
one can in fact choose it uniformly in  $0<\delta\leq \delta_0$.
\end{remark*}
\begin{proof} We fix $v_0$, $J$ and $\epsilon_0$ as in the statement of the the theorem.

We begin by taking $\delta_0^\prime$ to be the $\delta_0$ coming from Corollary~\ref{Cor-UnifMEPF}.
Secondly we fix $C_B^\prime$ and $c_0^\prime$ to be the two constants from Lemma~\ref{Lem-Unifinv}~(3).

We can now choose $0<\delta_0\leq\delta_0^\prime$ such that
\begin{equation}\label{ChoiceOfDelta}
d(\delta_0)\geq \max\{C_B^\prime+2, \max_{E\in J} 2C_B^\prime(E+c_0^\prime)\}.
\end{equation}
Here we used that $\lim_{t\to 0+} d(t) = +\infty$.
Fix now a $0<\delta\leq\delta_0$ and denote by $\gamma$ the radius
coming from Corollary~\ref{Cor-UnifMEPF}.

The above choices anticipates the proof below, but we make them here
to make it evident that we pick the constants in the right order.

We begin the verification of the commutator estimate for $v\in B_\gamma(v_0)$ by  computing as a form on $\cD((M_\delta^\rome)^{1/2}) \cap \cD(H^\rome)$
\begin{equation}\label{eq:MourreH^e_1}
\cU^{-1} \left [ M_\delta^\rome  +  R_\delta^\rome \right ] \cU
= \left [ M_\delta +  R_\delta \right ] \otimes P_\Omega
 + \big [ M_\delta \otimes \one + \one \otimes \hM_\delta
+ R_\delta \otimes \one  \big ] \one \otimes \bar{P}_\Omega.
\end{equation}

We apply Corollary~\ref{Cor-UnifMEPF} to the first term in the r.h.s. of \eqref{eq:MourreH^e_1},
with the given $\delta$ (apart from $v_0$, $J$ and $\epsilon_0$).
This yields a $C_4^\prime$, a $\kappa^\prime>0$, and a compact operator $K_0^\prime$ (apart from $\gamma$)
such that the following bound holds
\begin{equation}\label{eq:MourreH^e_2}
\left[ M_\delta  + R_\delta \right ] \otimes P_\Omega
 \ge [( 1 - \epsilon_0 ) \one  -
C_4^\prime f_{E,\kappa^\prime}^\perp(H_v^\PF)^2\la H_v^\PF\ra - K_0^\prime] \otimes P_\Omega.
\end{equation}
Observe that the bound above also holds with $\kappa^\prime$ replaced by any $0<\kappa\leq\kappa^\prime$.

To bound from below the second term on the r.h.s. of \eqref{eq:MourreH^e_1}, we use Lemma~\ref{Lem-Unifinv}.
Together with \eqref{eq:M-geH-} and \eqref{ChoiceOfDelta}, this implies
\begin{align}\label{eq:MourreHe3}
\nonumber & \left[ \one \otimes \hM_\delta
+ R_\delta \otimes \one \right]\one \otimes \bar{P}_\Omega \\
\nonumber & \ge  \bigg[ \one \otimes d(\delta) \left( \d\Gamma(\hat{h}) +  \tfrac12\right)  - C_B^\prime
\left( H_v^\PF\otimes\one + c_0^\prime\right) \otimes \one \bigg ] \one \otimes \bar{P}_\Omega \\
\nonumber & \ge \left[ ( d(\delta) - C_B^\prime ) \one \otimes \d\Gamma(\hat{h}) -  C_B^\prime(\hH^\rome_v -E)
+  \frac{d(\delta)}2  - C_B^\prime ( E+c_0^\prime) \right]  \one \otimes \bar{P}_\Omega \\
& \ge \left[2-  C_B^\prime(\hH_v^\rome -E)\right] \one \otimes \bar{P}_\Omega,
\end{align}
Here we also made use of \eqref{hatHexpanded} and that $\hat{h}\geq 0$. We now pick a $0<\kappa\leq\kappa^\prime$
such that $3\kappa C_B^\prime\leq 1$. Inserting $1 = f_{E,\kappa}^2 + 2f_{E,\kappa}f^\perp_{E,\kappa} + (f^\perp_{E,\kappa})^2$
into \eqref{eq:MourreHe3} yields the bound
\[
\left [ \one \otimes \hM_\delta
+ R_\delta \otimes \one \right ]
\one \otimes \bar{P}_\Omega \geq \left[1 -  C_B^\prime(1+E^\prime)f^\perp_{E,\kappa}(\hH_v^\rome)^2\la \hH_v^\rome\ra\right] \one \otimes \bar{P}_\Omega,
\]
where $E^\prime = \max_{E\in J} |E|$. This estimate
together with \eqref{eq:MourreH^e_1} and \eqref{eq:MourreH^e_2}
lead to the statement of the theorem with $C_4 = \min\{C_4^\prime,C_B^\prime(1+E^\prime)\}$ and $K_0 = \cU [K_0^\prime \otimes P_\Omega ]\cU^{-1}$.
\end{proof}


\subsection{Checking the Abstract Assumptions}\label{Subsec-AbstractAss}

The purpose of this subsection is to complete the proof of Theorem~\ref{Thm-MainPFE}.
We do this by running through the abstract assumptions in
Section~\ref{Sec-AbstractResults} pertaining to
Theorems~\ref{thm:mainresult} and \ref{thm:mainresultk=1},
from which Theorem~\ref{Thm-MainPFE} then follows. In accordance with Remark~\ref{remark:assumpt-stat-regul}~\ref{item:1bll}), we ensure that all constants can be
chosen locally uniformly in energy $E$ and form factor $v$. This ensures local uniformity in Theorem~\ref{Thm-MainPFE}.

We fix $v_0\in\cI_\PF(d)$ and $E_0\in\sigma(H_{v_0}^\PF)$. Observe that there exists $e_0$ such that
$e_0 < \inf\sigma(H_v^\PF)$ for all $v\in \cB_{\gamma_0}(v_0)$, where $\gamma_0$ comes from Lemma~\ref{Lem-Unifinv}.
Put $J= [e_0, E_0]$.
Let $\gamma$ and $\delta^\prime_0$ be fixed by Theorem~\ref{ME-PFE} and choose a $\delta<\delta_0^\prime$,
which from now on is fixed.

We begin by postulating the objects for which the abstract assumptions in
Conditions~\ref{cond:condition} should hold. We take
\begin{align}
\nonumber \cH & = \cH^\rome\\
\nonumber H & = H^\rome_v\\
A & = A_\delta^\rome \\
\nonumber N & = K^\rho\otimes \one_{\Gamma(\gothh^\rome)} +
 \one_\cK\otimes \d\Gamma(h^\prime) + \one_{\cH^\rome}, \ \ \max\{2\tau,\tfrac12\} < \rho<1\\
\nonumber H^\prime & = \left[M_\delta^\rome + R_\delta^\rome\right]_{|\cD(N)}.
\end{align}
The constant $\tau$ appearing above is the one from {\bf (I1)}.
Observe that $R_\delta^\rome$ and $M_\delta^\rome$ are $N$-bounded. See Lemma~\ref{Lem-Nbound} just below.

We make use of the following dense subspace of $\cH$
\[
\cS = \cD(K)\otimes \Gamma_\fin\big(C^\infty_0(\RR)\otimes L^2(S^{d-1})\big) \subseteq\cH^\rome.
\]
The tensor product is algebraic. Observe that $\cS$ is a core for $H$, $N$, and $A$.
We recall that we can construct the group $\e^{\i tA}$ explicitly.
Let $\psi_t$ denote the (global) flow for the $1$-dimensional ODE $\dot{\psi}_t(\omega) = m_\delta^\rome(\psi_t(\omega))$.
Then, for continuous compactly supported supported $f$,
\[
(\e^{\i t a_\delta^\rome}f)(\omega) = \e^{\frac12 \int_0^t
(m_\delta^\rome)^\prime(\psi_s(\omega))\d s}f(\psi_t(\omega)).
\]
This in particular implies that 
\begin{equation}\label{Eq-SInvUnderA}
\e^{\i t A_\delta^\rome} = \Gamma(\e^{\i t a_\delta^\rome}):\cS \to \cS.
\end{equation}

We begin with the following lemma which implies that $R_\delta^\rome$ is $N$-bounded.

\begin{lemma}\label{Lem-Nbound} Let $v\in \cO_\tau$ and $\kappa = 1/4-\tau/(2\rho)$. Then
$\cD(N^{1-2\kappa})\subseteq\cD(\phi(v))$, and for $f\in \cD(N)$ we have
\[
\|\phi(v^\rome)f\| \leq  C \|v\|_\tau\| N^{1-2\kappa} f\|,
\]
where $C$ does not depend on $v$ nor on $f$.
\end{lemma}

\begin{proof}
Adopting notation from \cite{GGM} we put $C_0(v) = \| v (K+1)^{-\tau}\|^2$ and $C_2(v) = \| [(K+1)^{-\tau}\otimes \one_{\gothh}] v\|^2$.
We estimate for $f\in \cS$, repeating the argument for \cite[(3.14) and (3.16)]{GGM}, and  get
\[
\| a^*(v^\rome) f\|^2\leq C_0(v) \|(K+1)^\tau\otimes \one_{\Gamma(\gothh^\rome)} f\|^2+ C_0(v) \la f, (K+1)^{2\tau}\otimes \cN^\rome f\ra
\]
and
\[
\| a(v^\rome) f\|^2\leq C_2(v)  \la f, (K+1)^{2\tau}\otimes \cN^\rome f\ra.
\]
Observing the bound, with $2\kappa = 1/2-\tau/\rho$ and some $C'>0$,
\begin{align*}
(K+1)^{2\tau}\otimes\cN^\rome & \leq \frac{\tau}{\rho (1-2\kappa)}(K+1)^{2\rho(1-2\kappa)}\otimes\one_{\Gamma(\gothh^\rome)} \\
 & \ \ \ + \frac1{2(1-2\kappa)}(\cN^\rome)^{2(1-2\kappa)}
\leq C'N^{2(1-2\kappa)},
\end{align*}
yields
\begin{equation}\label{Cond2Term2}
\|\Phi(v^\rome)f\|\leq C \|v\|_\tau \|N^{1-2\kappa}f\|
\end{equation}
a priori as a bound for elements of $\cS$. The lemma now follows since $\cS$ is a core for $N$.
\end{proof}

\vspace{5mm}

\noindent\emph{Condition~\ref{cond:condition} (\ref{item:2i0a}):}
We make use of the fact (given the invariance of $\cS$ mentioned
in \eqref{Eq-SInvUnderA}) that our Condition~\ref{cond:condition} (\ref{item:2i0a})
is equivalent to Mourre's conditions, $\e^{\i t
  A}\cD(N)\subseteq\cD(N)$ (i.e. $\cD(N)$ is invariant) and
that $\i [N,A]$ extends from a form on $\cS$ to an element of
$\cB(N^{-1}\cH;\cH)$. See \cite[Proposition II.1]{Mo}.

From the computation
\[
\i [h^\prime,a_\delta^\rome] =  m_\delta^\rome h^{\dprime}
\]
it follows that the following identity holds in the sense of forms on $\cS$
\begin{equation}\label{Nprime}
N' = \i [N,A_\delta^\rome] = \one_\cK\otimes \d \Gamma\big( m_\delta^\rome h^\dprime\big).
\end{equation}
Since $m_\delta^\rome$ is bounded
and $\sup_{\omega\in \RR}|h^\dprime(\omega)|/h^\prime(\omega) <\infty$, we find that
$N'$ extends from $\cS$ to a bounded operator on $\cD(N)$, and the extension
is in fact an element of $\cB(N^{-1}\cH;\cH)$ as required.

It remains to check that $\cD(N)$ is invariant under $\e^{\i t A_\delta^\rome}$.
For this we compute strongly on $\cS$
\[
N \e^{\i t A_\delta^\rome} = \e^{\i t A_\delta^\rome}\big(K^\rho\otimes\one_{\Gamma(\gothh^\rome)}
+ \one_{\cK}\otimes \d\Gamma(h^\prime \circ \psi_{-t})\big).
\]
Since $t\to \psi_t(\omega)$ is increasing and $\omega\to h^\prime(\omega)$ is decreasing (and positive)
we find for $t\leq 0$
\[
0\leq h^\prime\circ \psi_{-t} \leq h^\prime.
\]
For positive $t$ we estimate
$\omega-Ct\leq \psi_{-t}(\omega)\leq\omega$, for some $C>0$, where we used that $m_\delta^\rome$ was a bounded function.
This gives for $t>0$
\[
0\leq h^\prime\circ \psi_{-t}(\omega) = \max\{1, \e^{-\psi_{-t}(\omega)} + \psi_{-t}(\omega)\} \leq \max\{1, \e^{-\omega+Ct} + \omega - Ct\}.
\]
Using that $\e^{-\omega + \alpha} + \omega \leq C_\alpha (e^{-\omega} + \omega)$, we get for any $t$ a $C^\prime = C^\prime(t)$ such that
$(h'\circ \psi_{-t})^2 \leq C' (h')^2$ and hence by \cite[Proposition~3.4]{GGM} we arrive at
\[
\d\Gamma(h^\prime\circ\psi_{-t})^2 \leq C^\prime \d\Gamma(h^\prime)^2.
\]
Since $\cS$ was a core for $N$ we now conclude that $\e^{\i t
  A_\delta^\rome}\cD(N)\subseteq \cD(N)$.
This completes the verification of Condition~\ref{cond:condition} (\ref{item:2i0a}).

\noindent\emph{Condition~\ref{cond:condition} (\ref{item:2}):}
We begin by observing that $N$ and $H_0^\rome$ commute.
In particular we can compute as a form on $\cS$
\[
\i [N^{-1},H^\rome_v]  = \i N^{-1} \phi(v^\rome) - \i \phi(v^\rome) N^{-1}.
\]
This computation in conjunction with Lemma~\ref{Lem-Nbound} implies that $\i [N^{-1},H_v^\rome]$ extends from a form on $\cD(H_v^\rome)$ to a bounded operator
and hence $N$ is of class $C^1(H)$.

Since the commutator form $\i [N,H]$ extends from $\cD(N)\cap \cD(H)$ to a bounded form on $\cD(N)$
it suffices to compute it on a core for $N$. Here we take again $\cS$ and compute
\begin{align}\label{Cond2Comp}
\nonumber \i [N,H] & = \big[K^\rho\otimes\one_{\Gamma(\gothh^\rome)},\phi(v^\rome)\big] + \phi(\i h^\prime  v^\rome)\\
& =\phi\big([K^\rho\otimes\one_{\gothh^\rome}]v^\rome-v^\rome K^\rho\big) + \phi(\i v^\rome).
\end{align}
That the second term extends by continuity to a bounded form on $\cD(N^{\frac12-\kappa})$ follows from
Lemma~\ref{Lem-Nbound} (applied with $\i v^\rome$ instead of $v^\rome$) and interpolation.

In order to deal with the first term in \eqref{Cond2Comp} we write
\begin{align*}
\phi\big([K^\rho\otimes\one_{\gothh^\rome}]v^\rome-v^\rome K^\rho\big)
= \cU\Big(\phi\big([K^\rho\otimes\one_{\tilde{\gothh}}]\tv-\tv K^\rho\big)
\otimes\one_{\Gamma(\tilde{\gothh})}\Big) \cU^{-1}.
\end{align*}
Here we need the new assumption {\bf (I4)}. We will immediately verify that the above expression
extends to a bounded form on $\cD(N^{1/2-\kappa})$ for some $\kappa > 0$.
This implies the required property for $\i [H,N]^0$.

We employ the representation formula \eqref{Eq-Thomas}
with $K$ instead of $N$.
Compute as a form on $\cD(K\otimes\one_{\tilde{\gothh}})\times \cD(K)$
\begin{align*}
& (K^\rho\otimes\one_{\tilde{\gothh}})\tv-\tv K^\rho  = -c_\rho \int_0^\infty t^\rho\big[\big((K+t)^{-1}
\otimes \one_{\tilde{\gothh}}\big) \tv
- \tv (K+t)^{-1}\big] \d t\\
& = B -  c_\rho \int_1^\infty t^\rho\big((K+t)^{-1}\otimes \one_{\tilde{\gothh}}\big)\big[ \tv K
- (K\otimes\one_{\tilde{\gothh}}) \tv \big] (K+t)^{-1} \d t,
\end{align*}
where $B$ is the contribution from the integral between $0$ and $1$, which  due to {\bf (I1)} is a bounded operator.

By {\bf (I4)} we have
\begin{align*}
c_1 & := \big\|\big(  \tv K
- (K\otimes\one_{\tilde{\gothh}}) \tv \big)(K+1)^{ -\frac12}\big\| <\infty, \\
c_2 & := \big\|(K+1)^{ -\frac12}\otimes \one_{\tilde{\gothh}}\big(  \tv K
- (K\otimes\one_{\tilde{\gothh}}) \tv \big)\big\| <\infty.
\end{align*}
Let $\tau'<1/2$ be chosen such that $\rho/2 > \tau'> \rho-1/2$. This is possible due to the choice of $\rho$.
We estimate
for $\psi\in \cD(K\otimes\one_{\tilde{\gothh}})$ and $\varphi\in \cD(K)$
\begin{equation*}
\big\la\psi,\big((K^\rho\otimes\one_{\tilde{\gothh}})\tv-\tv K^\rho\big)\varphi\big\ra
 \leq \|B\|\|\psi\|\varphi\| + \frac{c_1 c_\rho}{\frac12+\tau'-\rho} \|\psi\| \|(K+1)^{\tau'}\varphi\|.
\end{equation*}
Similarly we get
\begin{equation*}
\big\la\psi,\big((K^\rho\otimes\one_{\tilde{\gothh}})\tv-\tv K^\rho\big)\varphi\big\ra
 \leq \|B\|\|\psi\|\varphi\| + \frac{c_2 c_\rho}{\frac12+\tau'-\rho} \|(K\otimes\one_{\tilde{\gothh}}+1)^{\tau'}\psi\| \|\varphi\|.
\end{equation*}
We have thus established that the first term in \eqref{Cond2Comp} is the (expanded) field operator associated to
an operator in $\cO_{\tau'}$.
We can thus employ Lemma~\ref{Lem-Nbound} again, this time with $v^\rome$ replaced by
$[K^\rho\otimes\one_{\gothh^\rome}]v^\rome-v^\rome K^\rho$ and $\kappa$ replaced by
$0<\kappa' = 1/4 - \tau'/(2\rho)< 1/4$.
Together with an interpolation argument this ensures that $\phi((K^\rho\otimes\one_{\gothh^\rome})v^\rome-v^\rome K^\rho)$ extends by continuity to a bounded form on $\cD(N^{1/2-\kappa'})$.

We have thus verified Condition~\ref{cond:condition}~(\ref{item:2}) with the smallest of the two kappa's.
In addition we observe that the $\cB(N^{-1/2+\kappa}\cH;N^{1/2 - \kappa}\cH)$-norm of $\i [N,H]^0$ is bounded
by a constant times $\|v\|_\PF$, cf. Remark~\ref{remark:assumpt-stat-regul}~\ref{item:1bll}).
\vspace{5mm}

\begin{remark}\label{Rem-RelaxI4} We observe from the discussion above that we could relax {\bf (I4)} and
require instead that $[K^\rho\otimes\one_{\tilde{\gothh}}]\tv-\tv K^\rho$ extends to an element of
$\cB(\cD(K^{\eta});\cK\otimes\tilde{\gothh})\cap \cB(\cK;\cD(K^\eta)^*\otimes\tilde{\gothh})$, for some
$1/2 \leq \eta < 1-\tau$, where $\tau$ is coming from ${\bf (I1)}$. This would still leave room
to choose $\rho$ and $\tau'$ (in the argument above) such that $1> \rho>2\tau$ and $\rho/2>\tau'> \rho+\eta-1$.
\end{remark}

While we do not know the domain of $H$, it turns out that we can indeed compute the intersection domain
$\cD(H)\cap \cD(N)$. This is done in the following lemma.

\begin{lemma}\label{Lemma-ExplDomain} We have the identity
\begin{equation}\label{DNcapDH}
 \cD(H)\cap \cD(N) = \cD\big(K\otimes\one_{\Gamma(\gothh^\rome)}\big)\cap 
\cD\big(\one_\cK\otimes\d\Gamma(\max\{h',\omega\})\big)
\end{equation}
and $\cS$ is dense in $\cD(H)\cap\cD(N)$ with respect to the intersection topology.
\end{lemma}

\begin{proof} Let for the purpose of this proof $H_0 = H_0^\e$, the unperturbed expanded Hamiltonian,
and denote by $\cD$ the right-hand side of \eqref{DNcapDH}.
Since $N$ controls the unphysical part of $\d\Gamma(h)$, due to the choice of extension of $\omega$ by an exponential,
we observe that the identity  \eqref{DNcapDH} holds if $H$ is replaced by $H_0$.
Since $H_0$ and $N$ commute we find that $T_0=N+\i H_0$ is a  closed operator on $\cD$
and it clearly generates a contraction semigroup.

We now construct the formal operator sum $N+\i H$ in two different ways.
By Lemma~\ref{Lem-Nbound} $\cD(\phi(v^\rome)) \subset \cD(N^{1-2\kappa})$ and hence for $u\in \cD$
\[
\|\phi(v^\rome) u\|\leq c\|N^{1-2\kappa}u\| + c'\|u\| \leq \frac14 \|Nu\| + c''\|u\|\leq \frac14 \|T_{0}u\| +c''\|u\|.
\]
From this estimate we deduce that  $T_{1} = T_{0}+\i \phi(v^\rome) =: N+\i G$  is a closed  operator on $\cD$ 
and it  generates a contraction semigroup. 
See \cite[Lemma preceding Theorem X.50]{RS}. Here $G$ is implicitly defined as the operator sum $G=H_0+\phi(v^\rome)$
with domain $\cD$.

On the other hand, since we have just established Condition~\ref{cond:condition}~(\ref{item:2}),
 we conclude from \cite[Theorem~2.25]{GGM1}
that $T_{2,\pm} = N\pm\i H$ are  closed operators on $\cD(H)\cap \cD(N)$. In addition
we have $T_{2,\pm}^* = T_{2,\mp}$ and since $T_{2,\pm}$ are both accretive we conclude that $T_{2,+}$ generates a contraction semigroup.
See \cite[Corollary to Theorem X.48]{RS}.

We proceed to argue that $T_2=T_{2,+}$ is an extension of $T_1$, i.e. $T_1\subset T_2$.
 Since $\cS\subseteq\cD$, $G$ is a symmetric extension of $H_{|\cS}$ and 
$\cS$ is a core for $H$ we deduce that $H$ is an extension of $G$. Hence
 indeed $T_1\subset T_2$.

We now argue that in fact $T_1= T_2$, or more poignantly that their domains
coincide. This will follow 
if the intersection of the resolvent sets is non-empty. Indeed, let
$\zeta\in \rho(T_1)\cap \rho(T_2)$. 
Then
\[
(T_2-\zeta)(T_1-\zeta)^{-1} = (T_1-\zeta)(T_1-\zeta)^{-1} = \one,
\]
and hence $(T_2-\zeta)^{-1} = (T_1-\zeta)^{-1}$ and the domains must coincide.
But by the Hille-Yosida theorem \cite[Theorem~X.47a]{RS} we have $(-\infty,0)\subset \rho(T_1)\cap \rho(T_2)$.
Here we used that both $T_1$ and $T_2$ generate contraction semigroups.

It remains to ascertain that $\cS$ is dense in $\cD$ with respect to the
intersection topology of $\cD(H)\cap\cD(N)$. We begin by verifying that $\cS$ is dense in $\cD$ with respect to the
graph norm of $T_0$, which induces the intersection topology of $\cD(H_0)\cap \cD(N)=\cD$.

Let $\psi\in\cD$. Observe first that $\lim_{n\to\infty} \one_{\cN^\rome\leq n}\psi\to\psi$ in the graph norm of $T_0$,
since $\cN^\rome$ and $T_0$ commute. Similarly we find that
$\one_\cK\otimes\Gamma(\one_{|\omega|\leq \ell}) \psi \to \psi$ in the graph norm of $T_0$.
Hence it suffices to approximate $\psi\in\cD$ with $\Gamma(\one_{|\omega|\leq \ell})\one_{\cN^\rome\leq n}\psi=\psi$,
for some $\ell$ and $n$, by elements from $\cS$ in the graph norm of $T_0$. 
Fix now such a $\psi$, $n$ and $\ell$. 

 Since $\cS$ is a core for $K\otimes\one_{\Gamma(\gothh^\rome)}$ we can find a sequence $\{\psi_j\}\subset \cS$
with $\psi_j\to\psi$ in $\cD(K\otimes\one_{\Gamma(\gothh^\rome)})$. Put
$\tilde{\psi}_j = \one_{\cN_\rome\leq n}[\one_\cK\otimes \Gamma(f)] \psi_j\in\cS$, where $f\in C_0^\infty(\RR)$,
with  $0\leq f\leq 1$ and $f=1$ on $[-\ell,\ell]$. Then $\tilde{\psi}_j\to\psi$
in $\cD(K\otimes \one_{\Gamma(\gothh^\rome)})$ as well. We now observe that
$T_0\tilde{\psi}_j  = (\i K\otimes\one_{\Gamma(\gothh^\rome)} + B_{n,\ell})\tilde{\psi}_j$,
for some bounded operator $B_{n,\ell}$. This implies density of $\cS$ in $\cD$ in the graph norm of $T_0$.

By the closed graph theorem $H(T_0-\zeta)^{-1}$ and
$N(T_0-\zeta)^{-1}$ are  bounded, and hence $\cS$ is also dense in $\cD(H)\cap\cD(N)=\cD$
with respect to the indicated intersection topology.
\end{proof}

\vspace{0.5cm}

\noindent\emph{Condition~\ref{cond:condition} (\ref{item:2i0ue}):}
Let $\sigma$ be such that $R(\eta)$ preserves $\cD(N)$ for $\eta$ with $|\Im\eta|\geq \sigma$.
It suffices to establish the identity
\[
R(\eta)H -H R(\eta) = -\i R(\eta)H^\prime R(\eta),
\]
for $\eta$ with $|\Im\eta|\geq \sigma$, 
as a form on $\cD(H)\cap \cD(N)$, since this set is dense in $\cD(H)\cap \cD(N^{1/2})$ by Remark~\ref{IntersectionDense}.

By Lemma~\ref{Lemma-ExplDomain}, we can on the set $\cD(H)\cap \cD(N)$ espress $H$ and $H'$ as sums of
operators $H = H_0^\rome + \phi(v^\rome)$ and $H' = \d\Gamma(h') -\phi(\i a_\delta^\rome v^\rome)$. 

We are thus reduced to verifying the following two form identities on $\cD(H)\cap \cD(N)$
\begin{align}\label{Eq-FreePart}
R(\eta)H^\rome_0 -H^\rome_0 R(\eta) &= -\i R(\eta)\one_\cK\otimes \d\Gamma(m_\delta^\rome h') R(\eta)\\
\label{Eq-IntPart} R(\eta)\phi(v^\rome) - \phi(v^\rome)R(\eta) &= \i R(\eta) \phi(\i a_\delta^\rome v^\rome) R(\eta). 
\end{align}
Since all operators appearing in \eqref{Eq-FreePart} commute with $\cN^\rome$ it suffices to verify this
identity on each fixed expanded particle sector with $\cN^\rome=n$.
Introduce for $\ell$ a positive integer the semibounded dispersion $h_\ell(\omega) = \max\{-\ell,h(\omega)\}$
and a cutoff expanded free Hamiltonian $H_{0,\ell} = K\otimes\one_{\Gamma(\gothh^\rome)} + \one_\cK\otimes \d\Gamma(h_\ell)$.
Then on a particle sector $H_{0,\ell}$ is of class $C^1_\Mo(A)$ such that we can compute for $|\Im\eta|\geq \sigma_{n,\ell}$
\[
R(\eta)H_{0,\ell} - H_{0,\ell}R(\eta) = -\i R(\eta)\one_\cK\otimes \d\Gamma(m_\delta^\rome h_\ell')R(\eta).
\]
as a form on $\one_{[\cN^\rome=n]}\cD$. Here $\sigma_{n,\ell}$ is some positive constant.
Since both sides are analytic in $\eta$ for $|\Im\eta|\geq \sigma$ we conclude the above identity for all such $\eta$.
Appealing to the explicit form of the domain $\cD$ we find that we can remove the cutoff $\ell\to\infty$
by the dominated convergence theorem. This yields \eqref{Eq-FreePart} for $|\Im\eta|\geq \sigma$.

As for \eqref{Eq-IntPart} we recall that we have already established that $N$ is of class $C^1_\Mo(A)$.
It is  a consequence of the proof of \cite[Proposition~II.1]{Mo}, that 
$\i[\phi(v^\rome),A]$ read as a form on $\cD(N)\cap\cD(A)$ can be represented by an extension
from the form computed on $\cS$. Here we used \eqref{Eq-SInvUnderA}.  As a form on $\cS$ we clearly have
$\i[\phi(v^\rome),A] = -\phi(\i a_\delta^\rome v^\rome)$, which extends to an $N$-bounded operator
by Lemma~\ref{Lem-Nbound}. The computation $R(\eta)\phi(v^\rome) - \phi(v^\rome)R(\eta) = R(\eta)[\phi(v^\rome),A]R(\eta)$
as forms on $\cD(N)$ now concludes the verification of \eqref{Eq-IntPart}, and hence
of Condition~\ref{cond:condition}~(\ref{item:2i0ue}).

\vspace{5mm}

\noindent\emph{Condition~\ref{cond:condition} (\ref{item:1}):}
We compute first as a form on $\cS$
\[
\i [H',A] = H^\dprime = \one_\cK\otimes \d\Gamma\big(m_\delta^\rome \frac{\d m_\delta^\rome}{\d \omega}h^\prime
+ (m_\delta^\rome)^2h^\dprime\big) - \phi\big((a_\delta^\rome)^2v^\rome\big)
\]
and observe that the right-hand side extends by continuity to an  $N$-bounded operator, cf. Lemma~\ref{Lem-Nbound}.
Again, by the proof of  \cite[Proposition~II.1]{Mo}, cf. \eqref{Eq-SInvUnderA}, we conclude that the operator on the right-hand side of the
formula also represents the commutator form $\i[H',A]$ on $\cD(N)\cap \cD(A)$.
\vspace{5mm}

\noindent\emph{Condition~\ref{cond:virial}:} By Lemma~\ref{Lemma-ExplDomain} and Remark~\ref{IntersectionDense}, 
it suffices to check the form bound
in the virial condition on $\cS$. In addition, since $K^\rho\leq \one + K$,
it suffices to check the estimate with $\rho=1$.

Recalling \eqref{Eq-hath} and \eqref{Eq-h-and-ve} we
observe that $\hat{h}\leq \hat{h}'$, and hence $h+h'\geq 0$. Making use of this observation we find that
\begin{align*}
K\otimes\one_{\Gamma(\gothh^\rome)} + \one_\cK\otimes \d\Gamma(h') & \leq K\otimes\one_{\Gamma(\gothh^\rome)}  + \one_\cK\otimes \big(\d\Gamma(h)  + 2 \d\Gamma (h')\big) \\
& \leq K\otimes \one_{\Gamma(\gothh^\rome)}+ \one_\cK\otimes \d\Gamma(h) + 2 M_\delta^\rome.
\end{align*}
We now add and subtract $\Phi(v^\rome) + 2R_\delta^\rome$ to obtain
\[
K\otimes\one_{\Gamma(\gothh^\rome)} + \one_\cK\otimes  \d\Gamma(h') \leq H^\rome_v + 2H' -\Phi(v^\rome) - 2R_\delta^\rome.
\]
We now make use of the fact that
\[
C = \|(\Phi(v^\rome) + 2R_\delta^\rome)(K\otimes\one_{\Gamma(\gothh^\rome)} + \one_\cK\otimes \d\Gamma(h') + 1)^{-\frac12}\| < \infty
\]
to conclude the form estimate
\[
K\otimes\one_{\Gamma(\gothh^\rome)} + \one_\cK \otimes \d\Gamma(h') \leq  H^\rome_v + 2H' + \tfrac12(K\otimes\one_{\Gamma(\gothh^\rome)} + \one_\cK\otimes \d\Gamma(h') + 1) + \tfrac12 C^2.
\]
This completes the verification of the virial bound. We again observe that the constants involved can be chosen
independent of $E$ in a bounded set and $v\in\cB_\gamma(v_0)$.

\vspace{5mm}

\noindent\emph{Condition~\ref{cond:mourre}:}
This condition has already been essentially verified in the form of Theorem~\ref{ME-PFE}.
We only need to observe that the form bound extends by continuity from
$\cD(H)\cap \cD(N)$ to $\cD(H)\cap\cD(N^{1/2})$, cf. Remark~\ref{IntersectionDense}.

\vspace{5mm}

\noindent\emph{The condition \eqref{eq:34}:} Let $\psi^\rome$ be a bound state for
$H=H^\rome_v$. That is $\psi^\rome \in \cD(H^\rome_v)$ and $H^\rome_v\psi^\rome = E\psi^\rome$, for some $E\in\RR$.
Recall that $\psi^\rome = \cU(\psi\otimes\Omega)$, where $\psi\in \cD(\tH^\PF_v)$ and $\tH^\PF_v\psi = E\psi$.
From \cite[Proposition~6.5]{GGM} we conclude that $\psi\in\cD(\cN^{1/2})$.
Hence we conclude that $\psi^\rome \in \cD(\d\Gamma(h')^{1/2})\cap \cU\cD(\tH_v^\PF\otimes
\one_{\Gamma(\tilde{\gothh})})$.
In particular we find that $\psi^\rome \in\cD(H)\cap \cD(N^{1/2})$ and the result follows from
the virial estimate in Condition~\ref{cond:virial}.
Observe again that $\|N^{1/2}\psi\|$ can be bounded uniformly in  $v\in \cB_{\gamma_0}(v_0)$ and $E\in [e_0,E_0]$.

\vspace{5mm}

\noindent\emph{Condition~\ref{cond:vir21} $k_0 = 1$:} This merely amounts to checking
the statement in \eqref{eq:5oiooa} with $\ell=0$. But this is trivially satisfied since
$[N,N'] = 0$. See \eqref{Nprime}.

This completes the verification of the conditions needed to conclude
Theorem~\ref{Thm-MainPFE} 
from Theorems~\ref{thm:mainresult} and \ref{thm:mainresultk=1}.

\section{AC-Stark type models}\label{AC--Stark model}

\subsection{The Model and the Result}

We will work in the framework of generalized $N$-body systems,
which we review briefly. Let $\cA$ be a finite index set and
$X$ a finite dimensional real vector-space with inner product.
There is an injective map from $\cA$ into the subspaces of
$X$, $\cA\ni a\to X^a\subseteq X$, and we write $X_a =
(X^a)^\perp$. We introduce a partial ordering on $\cA$:
$$
a\subset b \Leftrightarrow  X^a\subseteq X^b
$$
 and assume the following
\begin{enumerate}
\item There exist $a_{\rm min},a_{\rm max}\in\cA$ with $X^{a_{\rm min}}=\{0\}$ and $X^{a_{\rm max}}= X$.
\item For each $a,b\in \cA$ there exists $c=a\cup b\in\cA$ with
$X_a\cap X_b = X_c$.
\end{enumerate}

We will write $x^a$ and $x_a$ for the orthogonal projection of a vector
$x$ onto the subspaces $X^a$ and $X_a$ respectively.

We will work with a generalized potential
$$
V=V(t,x) = \sum_{a\in\cA \setminus\{a_{\rm min}\}} V_a(t,x^a),
$$
where $V_a$ is a real-valued  function on $\R\times X^a$. In the conditions below 
$\alpha$ denotes multiindices.

\begin{conds}\label{cond:ac-stark-model} Let $k_0\in \N$ be given.
  For each
$a\neq a_{\mathrm{min}}$ the following holds. The pair-potential $\R\times
X^a \ni(t,y)\to V_a(t,y) \in \R$ is a continuous  function
satisfying
\begin{enumerate}[\quad\normalfont (1)]
\item \label{item:3}Periodicity: $V_a(t+1,y)=V_a(t,y)$, $t\in\RR\mand y\in X^a$.
\item \label{item:4}Differentiability  in $y$: For all $\alpha$ with  $|\alpha|\leq k_0+1$ there exist
  $\partial^\alpha_y V_a\in C(\R\times X^a)$.
\item \label{item:5} Global bounds: For all $\alpha$ and $k\in \N\cup\{0\}$ with
  $|\alpha|+k\leq k_0+1$ there are global bounds
  $|\partial^\alpha_y(y\cdot \nabla_y)^kV_a(t,y)|\leq C$.
\item \label{item:6} Decay at infinity: $|V_a(t,y)| + |y\cdot\nabla_y
  V_a(t,y)|=o(1)$ uniformly in
$t$.
\item \label{item:7}Regularity in $t$: There exists
  $\partial_t V_a\in C(\R\times X^a)$ and there is a   global bound  $|\partial_tV_a(t,y)|\leq C$.
\end{enumerate}
\end{conds}

We consider under Condition~\ref{cond:ac-stark-model} the Hamiltonian $h=h(t)=p^2+V$, $p=-\i\nabla$,  on the
Hilbert space $L^2(X)$. The corresponding propagator $U$ satisfies:
It is two-parameter strongly continuous family of unitary
operators  which solves the time-dependent
Schr{\"o}dinger equation
$$
\i\frac{d}{dt}{U}(t,s)\phi = {h}(t){U}(t,s)\phi\; \mfor \phi\in \cD (p^2).
$$
The family  satisfies the
Chapman Kolmogorov equations
$$
{U}(s,r){U}(r,t)={U}(s,t),\qquad r,s,t\in\R,
$$ the initial condition ${U}(s,s)=\one$ for any $s\in\R$ and
 the  periodicity equation
$$
{U}(t+1,s+1) = {U}(t,s),\qquad s,t\in\R.
$$

The operator  $U(1,0)$ is called the {\it monodromy operator}. For each $a\neq a_{\rm max}$ the
  sub-Hamiltonian  monodromy operator is $U^a(1,0)$; it  is defined as the
  monodromy operator on $\cH^a=L^2(X^a)$ constructed for $a\neq
    a_{\rm min}$
  from $h^a=(p^a)^2+V^a$, $V^a=\sum_{a_{\rm min}\neq b\subset a}
  V_b(t,x^b)$. If $a=
    a_{\rm min}$ we define $U^a(1,0)=\one$ (implying  $\sigma_{\rm pp}(U^{a_{\rm min}}(1,0)) = \{1\}$).
The set of {\it thresholds} is then
\begin{equation}
  \label{eq:52}
  \cF(U(1,0)) = \bigcup_{a\neq  a_{\rm max}}
\sigma_{\rm pp}(U^a(1,0)),
\end{equation}

We recall from \cite{MS} that the set of thresholds  is closed and countable, and
non-threshold eigenvalues, i.e. points  in $\sigma_{\rm pp}(U(1,0))\backslash \cF(U(1,0))$, have finite multiplicity and can only accumulate at
the set of thresholds. Moreover any corresponding bound state is
exponentially decaying, the singular continuous spectrum $\sigma_{\rm
  sc}(U(1,0)) = \emptyset$ and there are integral propagation
estimates for states localized away from the set of eigenvalues  and away from
$\cF(U(1,0))$. It should be remarked that the weakest condition,
Condition~\ref{cond:ac-stark-model} with  $k_0=1$, corresponds to
\cite[Condition 1.1]{MS} (more precisely Condition~\ref{cond:ac-stark-model} with  $k_0=1$  is slightly weaker than
\cite[Condition 1.1]{MS}, and  we also
remark that \cite{MS} goes
through with this modification). All of the above properties  are proven in
\cite{MS}
either under \cite[Condition 1.1]{MS}  or under weaker conditions
allowing local singularities. In particular  local singularities up to
 the Coulomb singularity  are covered in  \cite{MS}.
 See Subsection~\ref{Regularity of non-threshold atomic type bound states} for a new
 result for Coulomb systems.

In the following subsection we establish the theorem below, which implies Theorem~\ref{thm:ac-stark-type2dd}~(\ref{item:9}).

\begin{thm}\label{thm:ac-stark-type} Suppose Conditions~\ref{cond:ac-stark-model},
for some $k_0\in \N$. Let $\phi$ be an bound state for $U(1,0)$ pertaining to an eigenvalue
$\e^{-\i\lambda}\notin \cF(U(1,0))$. Then $\phi\in \cD(|p|^{k_0+1})$.
  \end{thm}

\subsection{Regularity of Non-threshold Bound States}

 The principal tool in  the proof of Theorem~\ref{thm:ac-stark-type} will be Floquet
 theory (in common with \cite{MS} and other papers) which we  briefly review.
 The Floquet Hamiltonian associated with $h(t)$ is
 \begin{equation}
   \label{eq:54}
  H = \tau + h(t)=H_0+V,\qquad\textup{on } \cH = L^2\left([0,1];L^2(X)\right).
 \end{equation}

Here $\tau$ is the self-adjoint realization of $-\i\frac{\d}{\d t}$,
with periodic boundary conditions. The
spectral properties of the monodromy operator and the Floquet
Hamiltonian are equivalent. We have the following relations
$$
\sigma_{\mathrm{pp}}(U(1,0)) = \e^{-\i\sigma_{\mathrm{pp}}(H)},\ \sigma_{\mathrm{ac}}(U(1,0)) =
\e^{-\i\sigma_{\mathrm{ac}}(H)},\ \sigma_{\mathrm{sc}}(U(1,0)) = \e^{-\i\sigma_{\mathrm{sc}}(H)},
$$
and the multiplicity of an eigenvalue $z=\e^{-\i \lambda}$ of
$U(1,0)$ is equal to the multiplicity of $\lambda$ as an
eigenvalue of $H$ (regardless of the choice of $\lambda$). We also
recall  that the Floquet Hamiltonian is
the self-adjoint generator of the strongly continuous unitary
one-parameter group on $\cH$ given by
\begin{equation}\label{eq:50}
  (\e^{-is H}\psi)(t) = U(t,t-s)\psi(t-s-[t-s]),
\end{equation}
where $[r]$ is the integer part of $r$. In particular any bound state
of the monodromy operator, $U(1,0)\phi=\e^{-\i \lambda}\phi$, gives
rise to a bound state of the Floquet Hamiltonian,
$H\psi=\lambda\psi$, by the formula
\begin{equation}\label{eq:27}
  \psi(t)=\e^{\i t\lambda}U(t,0)\phi.
\end{equation}

\begin{prop}\label{prop:ac-stark-type} Suppose Conditions~\ref{cond:ac-stark-model}
for some $k_0\in \N$ and suppose  $H\psi=\lambda\psi$  for  $\e^{-\i\lambda}\notin \cF(U(1,0))$. Then $\psi\in \cD(|p|^{k_0+1})$.
\end{prop}

\begin{proof}
  We shall use Corollary \ref{cor:more-n-regularity}  with $H$
  being the  Floquet Hamiltonian and $N=p^2+1$. This amounts to
  checking the assumptions given in terms of Conditions~\ref{cond:condition}--\ref{cond:mourre},
  Condition~\ref{cond:conditioni}, Condition~\ref{cond:vir21} and (for $k_0\geq
  2$ only) Condition~\ref{cond:more-n-regularity} (same $k_0$). We take $A=\tfrac12 (x\cdot
  p+p\cdot x)$ and compute with direct reference to Conditions~\ref{cond:condition}, Condition~\ref{cond:conditioni} and  Condition~\ref{cond:vir21}
  \begin{subequations}
  \begin{align}
    \label{eq:35}H' &= 2p^2-x\cdot \nabla V,\\
\label{eq:36}\i[N,H]^0 &= p\cdot
  \nabla V+\nabla V\cdot p = \sum_{j=1}^{\dim X}\big( (p_j\partial_j
  V+(\partial_j V) p_j\big),\\
\label{eq:37}N' &= 2p^2,\\
\label{eq:38}\i^\ell\ad ^\ell_A(N')&=2^{\ell+1}p^2,\;
\i \,\ad _N\big (\i^\ell\ad^\ell_A(N')\big)=0;\;\ell\leq
k_0-1,\\
\label{eq:39}\i^l\ad ^l_A(H')&=2^{l+1}p^2+(-1)^{l+1}(x\cdot \nabla)^{l+1} V;\;l\leq
k_0.
  \end{align}
  \end{subequations}
   A comment on \eqref{eq:35} is due. We need to show
  Condition~\ref{cond:condition}~(\ref{item:2i0ue})  using the expression \eqref{eq:35}: First we
  remark that the operators $\tau$, $p^2$ and $H_0$ are simultaneously
  diagonalizable. Therefore $\cD(H)\cap\cD(N)=\cD(\tau)\cap\cD(N)$ is dense in $\cD(H)\cap\cD(N^{1/2})$.
  (See also Remark \ref{IntersectionDense}.)
  Moreover $p^2$, $V$ and $R(\eta)$ are obviously fibered
  (i.e. they act on the fiber space $L^2(X)$) and $R(\eta)$ preserves
  $\cD(p^2)$ and $\cD(|p|)$ for $|\eta|$ large enough. Whence as a form on $\cD(\tau)\cap\cD(N)$
  \begin{equation*}
    \i[H, R(\eta)]=\i[p^2+V, R(\eta)]=-R(\eta)\i[p^2+V,A]R(\eta)=-R(\eta)H'R(\eta).
  \end{equation*} The last identity for fiber operators is well-known in standard Mourre
  theory for Schr{\"o}dinger operators. Finally we extend the shown
  version of  \eqref{eq:1commu} by continuity to a form
  identity on $\cD(H)\cap\cD(N^{1/2})$ yielding
  Condition~\ref{cond:condition}~(\ref{item:2i0ue}).

Clearly \eqref{eq:virial} holds with $C_1=0$,
  $C_2=1/2$ and $C_3=1+\sup x\cdot \nabla V(t,x)/2$. As for \eqref{eq:Mourre} a stronger version follows
  from \cite[Theorem~4.2]{MS}
\begin{equation} \label{eq:MourreStark}
  H'\geq c_0\one - C_4f^{\bot}_\lambda(H)^2 - K_0.
\end{equation}
Finally it follows from \cite[Proposition~4.1]{MS}
that indeed the condition of Corollary~\ref{cor:more-n-regularity},
$\psi\in\cD(N^{1/2})=\cD(|p|)$, is fulfilled. This shows the
proposition in the case $k_0=1$.

For $k_0\geq 2$ it remains to verify Condition~\ref{cond:more-n-regularity}. For this purpose it is helpful  to
  notice that
  \begin{subequations}
    \begin{align}\label{eq:31}
    \i\, \ad _A(p_j)&=p_j,\\\label{eq:31k}
 \i\, \ad _A\big( (N+t_j)^{-1}\big) &= -2(N+t_j)^{-1}(N-1)(N+t_j)^{-1}.
  \end{align}
  \end{subequations}
   Moreover all computations are in terms of  fiber operators (in
  particular $M_1$, $M_2$ and $M_3$ are all fibered operators), and
  recalling \cite[Proposition~II.1]{Mo} and using the fact that
  $N^{1/2}\in C^1_{\mo}( A)$ it suffices to do the  computations in
  terms of forms on
  the Schwartz space $\cS(X)$.

\noindent {\bf Re $M_1$:} We shall apply \eqref{eq:31} in combination with
  \eqref{eq:36} to verify the part of Condition~\ref{cond:more-n-regularity}
  that involves $M_1$. Let us first look at the
  particular choice in \eqref{eq:2b} for  $M_1$ given by taking all
  the $T$'s equal $N^{1/2}$. That is we will demonstrate that for  $m=1,\dots,k_0-1$
\begin{equation}
  \label{eq:2baa}   \i^m\ad_{N^{1/2}}^m(M_1) \textup{ is }|p|\text{--bounded}.
\end{equation}
We compute
\begin{align*}
 &\i^m\ad_{N^{\frac12}}^m(M_1)=- (\i c_{\frac12})^{m+1}\int_0^\infty
 \d t_{m+1}\,t_{m+1}^{\frac12}\cdots\int_0^\infty \d t_{2}\,t_{2}^{\frac12} \int_0^\infty t_{1}^{\frac12} \\
&(N+t_1)^{-1}\cdots (N+t_{m+1})^{-1}\ad _{p^2}^{m+1}(V)(N+t_{m+1})^{-1}\cdots (N+t_1)^{-1}\d t_1,
\end{align*} and in turn,
\begin{align*}
  \ad
  _{p^2}^{m+1}(V)&=\sum_{|\alpha+\beta|=m+1}c_{\alpha,\beta}\,p^\alpha\big (\partial^{\alpha+\beta}V\big )
  p^\beta =T_1+T_2+T_3;\\
T_1&=\sum_{|\alpha+\beta|=m+1,\;|\beta|\geq 1}c_{\alpha,\beta}\,p^\alpha\big (\partial^{\alpha+\beta}V\big )
  p^\beta,\\
  T_2&=\sum_{|\alpha+\beta|=m+1,\;|\beta|=1}-\i \,c_{\alpha+\beta,0}\,p^\alpha\big (\partial^{\alpha+2\beta}V\big ),\\
T_3&=\sum_{|\alpha+\beta|=m+1,\;|\beta|=1}c_{\alpha+\beta,0}\,p^\alpha\big (\partial^{\alpha+\beta}V\big )
  p^\beta.
\end{align*}
 Now in front of the bounded derivative of any of the terms of the expressions
 $T_1$, $T_2$ and $T_3$ we move the factor $p^\alpha$ to the left in the
 integral representation and use the bound
\begin{equation}\label{eq:41}
   \|N^s(N+t)^{-1}\|\leq C_{s}(1+t)^{s-1};\;s\in [0,1].
 \end{equation}
 We obtain
 \begin{equation*}
   \|p^\alpha(N+t_{m+1})^{-1}\cdots (N+t_1)^{-1}\|\leq C_{s}^{m+1}\prod_{j=1}^{m+1}\,(1+t_j)^{s-1};\;s=\tfrac{|\alpha|}{2(m+1)}.
 \end{equation*}
 Using \eqref{eq:41}  for the factors of $p^\beta$ to
 the right (in case of $T_1$ and $T_3$) combined with the resolvents to the right and an additional
 factor $N^{-1/2}$ we obtain
\begin{equation*}
   \|p^\beta(N+t_{m+1})^{-1}\cdots (N+t_1)^{-1}N^{-\frac12}\|\leq C_{\sigma}^{m+1}\prod_{j=1}^{m+1}\,(1+t_j)^{\sigma-1};\;\sigma=\tfrac{|\beta|-1}{2(m+1)}.
 \end{equation*} To treat $T_2$ we notice that
\begin{equation}\label{eq:40}
   \|(N+t_{m+1})^{-1}\cdots (N+t_1)^{-1}\|\leq \prod_{j=1}^{m+1}\,(1+t_j)^{-1}.
 \end{equation}

Now the
 integrand with an additional
 factor $N^{-1/2}$ to the right is a sum of terms either bounded
  (up to a constant)  by
 \begin{equation*}
   \prod_{j=1}^{m+1}\,t_{j}^{\frac12}(1+t_j)^{\tfrac{|\alpha|}{2(m+1)}-1}(1+t_j)^{\tfrac{|\beta|-1}{2(m+1)}-1}=
   \prod_{j=1}^{m+1}\,t_{j}^{\frac12}(1+t_j)^{-\tfrac32-\tfrac{1}{2(m+1)}}
 \end{equation*} (these terms come from $T_1$ and $T_3$), or  (for  any
 term of $T_2$) by
 \begin{equation*}
   \prod_{j=1}^{m+1}\,t_{j}^{\frac12}(1+t_j)^{\tfrac{|\alpha|}{2(m+1)}-1}(1+t_j)^{-1}=
   \prod_{j=1}^{m+1}\,t_{j}^{\frac12}(1+t_j)^{\tfrac{m}{2(m+1)}-2}.
 \end{equation*} Whence in  all cases the integral with an additional
 factor $N^{-1/2}$ to the right is  convergent in
 norm, which finishes the proof of the special case where  all of
  the $T$'s are equal to $N^{1/2}$. The general case follows  the same
  scheme. Some of the commutators with $A$ ``hit''  the potential part
  introducing a change $W(t,x)\to -x\cdot \nabla W(t,x)$. Other
  commutators with $A$ hit a factor $p_j$ in which case we apply
  \eqref{eq:31}. Finally yet other commutators with $A$ hit a factor
  $(N+t_j)^{-1}$ in which case we apply \eqref{eq:31k} and
  \eqref{eq:40}.

\noindent {\bf Re $M_2$ and  $M_3$:} The contributions to \eqref{eq:2b} from the first term
of \eqref{eq:35}, i.e. contributions from the expression $2p^2N^{-1/2}$, vanish
except for the case where all of
  the $T$'s are equal to $A$. In this case we compute
\begin{equation}\label{eq:2bx}
\i^m\ad_{A}^m \big(2p^2N^{-\frac12}\big)=\big(2t\tfrac{\d}{\d t}\big)^m f(t)_{\big |t=p^2};\;f(t)=2t(t+1)^{-\frac12}.
\end{equation} Obviously the right hand side of \eqref{eq:2bx} is
$N^{1/2}\text{--bounded}$.

The contributions to \eqref{eq:2b} from the expressions   $-x\cdot
\nabla VN^{-1/2}$ and $-N^{-1/2}x\cdot
\nabla V$ are treated like the term  $M_1$ in fact slightly
simpler. The iterated commutators are all bounded in this case. We
leave out the details.
  \end{proof}
\begin{remark*}\label{remark:ac-stark-type} Since $H$ is not elliptic
  (more precisely $|p|(H_0-\i)^{-1}$ is unbounded) we do not see an
  ``easy way''  to get the conclusion of Proposition~\ref{prop:ac-stark-type}.
  For instance we need to use the assumption
  that $\e^{-\i\lambda}$ is non-threshold. See \cite{KY} for a related result for the one-body AC-Stark problem.
  \end{remark*}
\begin{proof}[Proof of Theorem \ref{thm:ac-stark-type}:] We mimic the
  proof of \cite[Theorem~1.8]{MS}. Recall the notation
   $I_{\i n}(N)=n(N+n)^{-1}$ and $N_{\i n}=NI_{\i n}(N)$.
   Due to Proposition~\ref{prop:ac-stark-type} and
   the representation \eqref{eq:27} there exists $t_0\in [0,1[$ such that
  \begin{equation}\label{eq:42}
    U(t_0,0)\phi\in \cD(N^{(k_0+1)/2}).
  \end{equation}
  In particular $\psi(t)=\e^{\i t\lambda}U(t,0)\phi\in \cD(p^2)$ for all $t$. Next we compute
  \begin{subequations}
   \begin{align}\label{eq:44}
    \tfrac {\d}{\d t}\inp{\psi(t),N_{\i n}^{k_0+1}\psi(t)}&=\inp{\psi(t),\i [V,N_{\i n}^{k_0+1}]\psi(t)},\\ \label{eq:45}
    \i [V,N_{\i n}^{k_0+1}]&=\sum_{0\leq p\leq k_0}N_{\i n}^{p}
    \i [V,N_{\i n}]N_{\i n}^{k_0-p},\\\label{eq:46}
    \i [V,N_{\i n}]&=-I_{\i n}(N)\sum_{j=1}^{\dim X}\big((p_j\partial_j V+(\partial_j V) p_j\big)I_{\i n}(N).
  \end{align}
  \end{subequations}
  We plug \eqref{eq:46} into \eqref{eq:45} and then in turn
  \eqref{eq:45} into the
  right hand side of \eqref{eq:44}. We expand the sum and redistribute
  for each term at most $k_0$ derivatives by pulling
  through the factor $\partial_j V$ obtaining terms on a  more
  symmetric form, more precisely on the form
  \begin{equation}\label{eq:43}
    \big\la N^{\frac{k_0+1}2}\psi(t),B_n N^{\frac{k_0+1}2}\psi(t)\big\ra\textup{ where }\sup_n\|B_n\|<\infty.
  \end{equation} Notice that for all terms the  operator $B_n$ involves at most
  $k_0+1$ derivatives of $V$. Thanks  to the Cauchy-Schwarz inequality
  and Proposition
  \ref{prop:ac-stark-type}  any expression like
  \eqref{eq:43} can be integrated on $[t_0,1]$ and the integral is bounded
  uniformly in $n$. In combination with \eqref{eq:42} we
  conclude that
  \begin{equation*}
    \sup_n\,\big\la \psi(1),N_{\i
    n}^{k_0+1}\psi(1)\big\ra <\infty,
  \end{equation*} whence  $\phi=\psi(1)\in \cD(N^{(k_0+1)/2})$.
\end{proof}

\subsection{Regularity of Non-threshold Atomic Type Bound States}\label{Regularity of non-threshold atomic type bound states}
The generator of the evolution of the a system of $N$ particles in a
time-periodic Stark-field with zero mean (AC-Stark field) is of the form
$$h_{\phy}(t)=p^2-\cE(t)\cdot x+V_{\phy}$$ on $L^2(X)$. Assuming that the field is
$1$-periodic the condition $\int_0^1\cE(t)\d t=0$ leads to the
existence of unique $1$-periodic functions $b$ and $c$ such that
\begin{equation*}
  \tfrac{\d}{\d t}b(t)=\cE(t),\;\tfrac{\d}{\d t}c(t)=2b(t))\mand \int_0^1\,c(t)\d t=0;
\end{equation*}
see \cite{MS} for details. For simplicity let us here  assume that
$\cE\in C([0,1];X)$, see Remark \ref{remark:regul-non-threshL1}
for an extension.  The potential $V_{\phy}$ is a
sum of time-independent real-valued ``pair-potentials''
$$
V_{\phy}=V_{\phy}(x) = \sum_{a\in\cA \setminus\{a_{\rm min}\}} V_a(x^a).
$$ In terms of these quantities we
introduce Hamiltonians
\begin{align*}
  h_{\aux}(t)&=p^2+2b(t)\cdot p+V_{\phy},\\
 h(t)&=p^2+V_{\phy}(\cdot+c(t)).
\end{align*} The propagators $U_{\phy}$, $U_{\aux}$ and $U$ of
$h_{\phy}$, $h_{\aux}$ and $h$, respectively,  are
linked by Galileo type transformations. Define
\begin{equation*}
  S_1(t)=\e^{\i c(t) \cdot p} \mand S_2(t)=\e^{\i (b(t) \cdot
    x-\alpha(t))};\;\alpha(t)=\int_0^t\,|b(s)|^2\,\d s.
\end{equation*} 
Then
\begin{subequations}
\begin{align}\label{eq:47}
  U_{\phy}(t,0)&=S_2(t)U_{\aux}(t,0)S_2(0)^{-1},\\\label{eq:48}
U(t,0)&= S_1(t)U_{\aux}(t,0)S_1(0)^{-1},\\
U_{\phy}(t,0)&=S_2(t)S_1(t)^{-1}U(t,0)S_1(0)S_2(0)^{-1}.\label{eq:49}
\end{align}
\end{subequations} 
The bulk of \cite{MS} is a study of the Floquet Hamiltonian of
$h$. Spectral information is consequently   deduced for
the monodromy operator $U(1,0)$. Finally the formula \eqref{eq:49} then
gives spectral information for the physical monodromy operator
$U_{\phy}(1,0)$. The part of \cite{MS} concerning potentials with
local singularities contains  an incorrect  reference in that it is
referred to \cite{Ya} for the existence of the propagator $U$ (see
\cite[Remark~1.4]{MS}). However although the issue of Yajima's paper
is the existence of an appropriate  dynamics for singular time-dependent
potentials the paper as well as the  method of proof  is for the one-body problem only.
This point is easily fixed as follows, see Remark~\ref{remark:regul-non-threshL1}
for a more complicated procedure for
$\cE\in L^1([0,1];X)\setminus C([0,1];X)$: We use Yosida's
theorem which is in fact  also alluded to in \cite[Remark~1.4]{MS}
(see \cite[Theorem~II.21]{SiBook} for a statement of the theorem).
If $V_{\phy}$ is $\epsilon$-bounded
relatively to $p^2$ (which is the case under the conditions considered
in \cite{MS}) then indeed the propagator $U_{\aux}$ exists
and we can use \eqref{eq:47} and \eqref{eq:48} to define  $U_{\phy}$
and $U$. In particular we can use
\eqref{eq:49}
and obtain not only the existence of $U_{\phy}$ but various spectral
information of the corresponding monodromy operator
$U_{\phy}(1,0)$ (see  the introduction of \cite{MS} for details). We
remark that the construction of  the Floquet Hamiltonian of
$h$ is done independently of $U$ although of course
\eqref{eq:50} may be taken as a definition.

Let us for completeness
note the following  by-product of  Yosida's theorem
(intimately related to its proof): Pick
$\lambda_0\in \R$ such that $ h_{\aux}(t)\geq \lambda_0+1$ for all
$t$. The crucial assumption in the
theorem is the boundedness of the function
\begin{equation}\label{eq:51}
 t\to \| ( h_{\aux}(t)-\lambda_0)^{-1}  \tfrac {\d}{\d t}( h_{\aux}(t)-\lambda_0)^{-1} \|.
\end{equation}
Since, by assumption  $\cE\in C([0,1];X)$, clearly the following
constant is a bound of \eqref{eq:51},
\begin{equation*}
  C:=2\sup_t|\cE(t)|\,\sup_t\|\,|p|( h_{\aux}(t)-\lambda_0)^{-1} \|.
\end{equation*} 
We have the explicit  bound of the dynamics restricted to $\cD (p^2)$.
\begin{equation*}
  \|( h_{\aux}(t)-\lambda_0)U_{\aux}(t,0)\phi\|^2\leq \e^{2C|t|}\|( h_{\aux}(0)-\lambda_0)\phi\|^2\; \mfor \phi\in \cD (p^2).
\end{equation*} 
Let us also note the following property of the dynamics restricted
to $\cD (|p|)$, cf. \cite[Theorems~II.23~and~II.27]{SiBook},
\begin{align}
  &\|( h_{\aux}(t)-\lambda_0)^{1/2}U_{\aux}(t,0)\phi\|^2\nonumber\\&\leq
  \e^{\tC |t|}\|( h_{\aux}(0)-\lambda_0)^{1/2}\phi\|^2\; \mfor \phi\in \cD (|p|);\label{eq:56}
\end{align}
here
\begin{equation*}
  \tC:=2\sup_t|\cE(t)|\,\sup_t\|\,|p|^{1/2}( h_{\aux}(t)-\lambda_0)^{-1/2} \|^2.
\end{equation*}

\begin{remark}\label{remark:regul-non-threshL1}
 If $\cE\in L^1([0,1];X)$ but possibly $\cE\not \in C([0,1];X)$ we can
 still show that there exists an appropriate dynamics $U$  under the conditions considered
in \cite{MS}, although  possibly not one that preserves $\cD(p^2)$. We can use \cite[Theorem~II.27]{SiBook} directly on $h$.
For the borderline case, the Coulomb singularity, Hardy's inequality
\cite[(6.2)]{MS} is needed to verify the assumptions of this theorem;
the details are not discussed here. This yields a
dynamics $U$ preserving $\cD(|p|)$ which is good enough for getting
the conclusions of \cite{MS} related to the condition  $\cE\in
L^1([0,1];X)$. The results presented below can similarly be extended
to $\cE\in L^1([0,1];X)$.
\end{remark}

The following condition is an extension of \cite[Condition~1.3]{MS}
(which corresponds to $k_0=1$ below). The Coulomb potential
commonly used to describe atomic and molecular systems (here with moving
nuclei) is included.

\begin{conds}\label{cond:ac-stark-model2} Let $k_0\in \N$ be given.
  For each
$a\neq a_{\rm min}$ the following holds. The pair-potential $
X^a \ni y\to V_a(y) \in \R$ splits into a sum $V_a=V_a^1+V_a^2$ where
\begin{enumerate}[\quad\normalfont (1)]
\item \label{item:4i}Differentiability: $V_a^1\in C^{k_0+1}(X^a)$ and $V_a^2\in C^{k_0+1}(X^a\setminus\{0\})$.
\item \label{item:5i} Global bounds: For all $\alpha$ with
  $|\alpha|\leq k_0+1$ there are  bounds
  $|y|^{|\alpha|}\,|\partial^\alpha_y V_a^1(y)|\leq C$.
\item \label{item:6i} Decay at infinity: $|V_a^1(y)| + |y\cdot\nabla_y
  V_a^1(y)|=o(1)$.
\item \label{item:3i}Dimensionality: $V_a^2=0$ if $\dim X^a<3$.
\item \label{item:7i}Local singularity: $V_a^2$ is compactly supported
  and for all $\alpha$ with
  $|\alpha|\leq k_0+1$ there are  bounds
  $|y|^{|\alpha|+1}\,|\partial^\alpha_yV_a^2(y)|\leq C$; $y\neq 0$.
\end{enumerate}
\end{conds}

We note that the part of time-dependent potential
$V_{\phy}(\cdot+c(t))$ coming from  the first term $V_a^1$ of the
splitting of $V_a$ in Condition~\ref{cond:ac-stark-model2} conforms
with Condition~\ref{cond:ac-stark-model}. The part from $V_a^2$ does
not, and we do not in general expect there to be an analogue of
Theorem~\ref{thm:ac-stark-type} in this case for $k_0>1$. It is an open
problem to determine  whether there is an   analogue statement of
Theorem~\ref{thm:ac-stark-type} for $k_0=1$. Notice that the lowest degree of
regularity, $\phi\in\cD (|p|)$, holds even without the non-threshold
condition, cf. \cite[Theorem~1.8]{MS}. On the other hand  since the
singularity is located at  $x=-c(t)$ we would expect and we will
indeed prove regularity with
respect to  the observable
\begin{equation}\label{eq:53}
  A=A(t)=\tfrac 12 \parb{(x+c(t))\cdot p+p\cdot (x+c(t))}=S_1(t)\tfrac 12 \parb{x\cdot p+p\cdot x}S_1(t)^{-1}.
\end{equation} 
This regularity is the content of Theorem~\ref{thm:ac-stark-type2}
stated below; see \cite[Proposition~8.7~(ii)]{MS} for a related
result in the case $k_0=1$ at the level of Floquet theory,
cf. Proposition~\ref{prop:ac-stark-type2} stated below. The $A$-regularity
statement of  the theorem  for  $k_0>1$ is
new. The set of thresholds is defined as before, see \eqref{eq:52}.

\begin{thm}\label{thm:ac-stark-type2} Suppose Conditions~\ref{cond:ac-stark-model2}
  for some $k_0\in \N$. Let $\phi$ be a bound state for $U(1,0)$ pertaining to an eigenvalue
  $\e^{-\i\lambda}\notin \cF(U(1,0))$. Then $\phi\in \cD(A(1)^{k_0})$
  where $A(1)$ is given by taking $t=1$ in \eqref{eq:53}.
  \end{thm}

The above theorem implies Theorem~\ref{thm:ac-stark-type2dd}~(\ref{item:8}).
We shall prove  Theorem~\ref{thm:ac-stark-type2} along the same lines as that of the proof of
Theorem \ref{thm:ac-stark-type}. Whence we introduce the Floquet
Hamiltonian by the expression \eqref{eq:54} (with $V=V_{\phy}(\cdot+c(t))$).
  By \cite[Theorem~6.2]{MS} $V$ is
  $\epsilon$-bounded relatively to $H_0$ whence $H$ is self-adjoint.

\begin{prop}\label{prop:ac-stark-type2} Suppose Conditions~\ref{cond:ac-stark-model2}
  for some $k_0\in \N$ and suppose  $H\psi=\lambda\psi$  for
  $\e^{-\i\lambda}\notin \cF(U(1,0))$. Then for any $k,\ell\geq0$, with $k+\ell \leq  k_0$,
  we have   $\psi\in  \cD(A^k \inp{p}A^\ell)$  where $A$ is given by \eqref{eq:53}.
\end{prop}

\begin{proof}
  It is tempting to try to apply Corollary~\ref{cor:assumpt-stat-regul}  with $H$
  being the   Floquet Hamiltonian, $A$ being as stated
  and $N=p^2+1$. In fact all of the conditions of  Corollary~\ref{cor:assumpt-stat-regul} can be verified except
  for Condition~\ref{cond:condition}~(\ref{item:2})  (notice that the
  formal analogue of \eqref{eq:36} might be too singular). This deficiency
  will be discussed at the end of the  proof. All other conditions can
  be verified with
\begin{subequations}
  \begin{align}
    \label{eq:35k}H'&=2p^2-(x+c)\cdot \nabla V+2b\cdot p,\\
\label{eq:37k}N'&=2p^2,\\
\label{eq:38k}\i^\ell\ad ^\ell_A(N')&=2^{\ell+1}p^2,\;
\i \,\ad _N\big (\i^\ell\ad ^\ell_A(N')\big )=0;\;\ell\leq
k_0-1,\\
\label{eq:39k}\i^\ell\ad ^\ell_A(H')&=2^{\ell+1}p^2+(-1)^{\ell+1}((x+c)\cdot \nabla)^{\ell+1} V+2b\cdot p;\;\ell\leq
k_0.
  \end{align}
\end{subequations}
 Comments are due. First,  the second and the third terms of \eqref{eq:35k} are
  bounded relatively to $|p|$ uniformly in $t$, cf.  the Hardy
  inequality \cite[(6.2)]{MS}, and whence indeed \eqref{eq:35k} is
  $N$-bounded. We need to verify
  Condition~\ref{cond:condition}~(\ref{item:2i0ue})  using the
  expression \eqref{eq:35k}: The operators $p^2$, $V$ and $R(\eta)$ are
  fibered
   and $R(\eta)$ preserves
  $\cD(p^2)$ and $\cD(|p|)$ for $|\eta|$ large enough (uniformly in $t$). Whence as a form on $\cD(\tau)\cap\cD(N)$
  \begin{align*}
    \i[h, R(\eta)]&=-R(\eta)\i[p^2+V,A]R(\eta)
=-R(\eta)\parb{2p^2-(x+c)\cdot \nabla V}R(\eta),\\
\i[\tau, R(\eta)]&=-R(\eta)2b\cdot pR(\eta),
  \end{align*} and therefore
  \begin{equation*}
    \i[H, R(\eta)]=-R(\eta)H'R(\eta).
  \end{equation*}
 Using again  that $\cD(H)\cap\cD(N)=\cD(\tau)\cap\cD(N)$
 is dense in $\cD(H)\cap\cD(N^{1/2})$, cf. Remark~\ref{IntersectionDense}, the  latter form identity  can be extended by continuity to a form
  identity on $\cD(H)\cap\cD(N^{1/2})$ yielding
  Condition~\ref{cond:condition}~(\ref{item:2i0ue}).

As for \eqref{eq:37k}, \eqref{eq:38k}, \eqref{eq:39k},  Conditions~\ref{cond:condition}~(\ref{item:2i0a}) and~(\ref{item:1}),
 Condition~\ref{cond:conditioni} and Condition~\ref{cond:vir21} the verification is straightforward (omitted here).

To show  \eqref{eq:virial} we first  introduce the natural notation
$V=V^1+V^2$ reflecting the splitting of Conditions~\ref{cond:ac-stark-model2}.
Then  we introduce
\begin{equation*}
  C=\big\||p|^{-\frac12}\parb{(x+c)\cdot \nabla V^2-2b\cdot
    p}|p|^{-\frac12}\big\|\mand \tC=\big\|(x+c)\cdot \nabla V^1\big\|;
\end{equation*}
the norm is the operator norm on $\cH$. Then we note that
\begin{equation*}
  N\leq \tfrac12 H'+\tfrac12 C|p|+1+ \tfrac12 \tC,
\end{equation*}
yielding \eqref{eq:virial} with $C_1=0$,
  $C_2=1$ and $C_3=1+C^2/4+\tC$ understood as  a form  on
  $\cD(N^{1/2})$. We have verified Condition \ref{cond:virial}.

As for \eqref{eq:Mourre} a stronger version follows
  from \cite[Proposition~6.4]{MS}
\begin{equation} \label{eq:MourreStarkk}
  H'\geq c_0\one - C_4f^{\bot}_\lambda(H)^2 - K_0.
\end{equation} Here we use the condition that $\e^{-\i\lambda}\notin
\cF(U(1,0))$. The estimate \eqref{eq:MourreStarkk} is valid as a form  on
  $\cD(N^{1/2})$. Finally it follows from \cite[Theorem~6.3]{MS}
that indeed the condition of Corollary~\ref{cor:assumpt-stat-regul},
$\psi\in\cD(N^{1/2})=\cD(|p|)$, is fulfilled.

Now to the deficiency given by the lack of Condition~\ref{cond:condition}~(\ref{item:2}).
Checking the proof of Corollary~\ref{cor:assumpt-stat-regul}
it is realized that Condition~\ref{cond:condition}~(\ref{item:2}) is used only to assure boundedness of
$N^{1/2}BN^{-1/2}$, where under the assumption \eqref{eq:Mourre}
we have $B=C_4f^{\bot}_\lambda(H)^2\inp{H}(H-\lambda)^{-1}$. In our case we have
 a slightly stronger version  of the Mourre
estimate, \eqref{eq:MourreStarkk}, so what  we really need  is
\begin{equation}\label{eq:55}
  N^{\frac12}BN^{-\frac12}\in \cB(\cH)\textup{ where }B=g(H);\;g(E)=f^{\bot}_\lambda(E)^2(E-\lambda)^{-1}.
\end{equation}

So let us show \eqref{eq:55} without invoking a
condition like  Condition~\ref{cond:condition}~(\ref{item:2}). Clearly it suffices to show that
the commutator
\begin{equation}\label{eq:55e}
  [N^{\frac12},g(H)]\in \cB(\cH).
\end{equation}
But
\begin{align*}
  [N^{\frac12},g(H)]&=c_{\frac12}\int_0^\infty t^{\frac12}
  (N+t)^{-1}[N,g(H)](N+t)^{-1}\d t,\\
[N,g(H)]&=[H-V+I-\tau,g(H)]=-[\tau,g(H)]+T,\\
-[\tau,g(H)]&=\frac1{\pi}\int_{\C}(\dbar\tilde{g})(\eta)(H-\eta)^{-1}[\tau,V](H-\eta)^{-1}\d u\,\d v,\\
-[\tau,V]&=\i 2b\cdot \nabla V.
\end{align*}
Here the term $T$ is bounded since $V$ is bounded
relatively to $H$; whence indeed $T$ gives a bounded contribution to
the commutator in \eqref{eq:55e}. As for the contribution from the
term $-[\tau,g(H)]$ only the part from $V^2$ is non-trivial. For that
part we use \cite[(6.6)]{MS} to obtain
\begin{equation*}
  \|(H-\eta)^{-1}2b\cdot \nabla V^2(H-\eta)^{-1}\|\leq
  C\max \parb{|\Im \eta|^{-2},|\Im \eta|^{-\frac12}}.
\end{equation*}
Whence we can bound the integral
\begin{align*}
   & \Big\|\int_{\C}(\dbar\tilde{g})(\eta)(H-\eta)^{-1}2b\cdot \nabla V^2(H-\eta)^{-1}\d u\,\d v\Big\|\\
   &\leq C\int_{\C} | (\dbar\tilde{g})(\eta)| \,\max \parb{|\Im \eta|^{-2},|\Im \eta|^{-\frac12}}
\,\d u\,\d v<\infty.
\end{align*}
This means that also the first term $-[\tau,g(H)]$ is bounded and
whence in turn its contribution to the commutator in  \eqref{eq:55e}
agrees with the statement of \eqref{eq:55e}. We
have proven  \eqref{eq:55e}.
\end{proof}

\begin{proof}[Proof of Theorem~\ref{thm:ac-stark-type2}:] We mimic the
  proof of Theorem~\ref{thm:ac-stark-type}. Recall the notation
  $I_{n}(A)=-\i n(A-\i n)^{-1}$ and $A_{n}=AI_{n}(A)$. Due to Proposition~\ref{prop:ac-stark-type2} and
  the representation \eqref{eq:27} there exists $t_0\in [0,1[$ such that
  \begin{equation}\label{eq:42m}
    U(t_0,0)\phi\in \cD(|p|)\cap\cD(A(t_0)^{k_0}).
  \end{equation}
  In particular $\psi(t)=\e^{\i t\lambda}U(t,0)\phi\in \cD(|p|)$ for all $t$,
 cf. \eqref{eq:48} and \eqref{eq:56}. Moreover $\psi(\cdot)$ is differentiable as a
 $\cD(|p|)^*$--valued function, and in this sense
 \begin{equation*}
    \i \tfrac {\d}{\d t}\psi(t)=(h(t)-\lambda)\psi(t).
 \end{equation*}
 Whence  we can compute
  \begin{subequations}
   \begin{align}\label{eq:44m}
    \tfrac {\d}{\d t}\|A_{n}^{k_0}\psi(t)\|^2
    &=2\Re\inp{A_{n}^{k_0}\psi(t),\parb{\i [h(t),A_{n}^{k_0}]+\tfrac {\d}{\d t}A_{n}^{k_0}}\psi(t)},\\\label{eq:45m}
    \i [h(t),A_{n}^{k_0}]+\tfrac {\d}{\d t}A_{n}^{k_0}
    &=\sum_{0\leq p\leq k_0-1}A_{n}^{p}\parb{\i [h(t),A_{n}]+\tfrac {\d}{\d t}A_{n}}A_{n}^{k_0-p-1},\\\label{eq:46m}
    \i [h(t),A_{n}]+\tfrac {\d}{\d t}A_{n}
    &=I_{n}(A)\big (2p^2+2b\cdot p-(x+c)\cdot \nabla V\big)I_{n}(A).
  \end{align}
  \end{subequations}
  We plug \eqref{eq:46m} into \eqref{eq:45m} and then in turn
  \eqref{eq:45m} into the
  right hand side of \eqref{eq:44m}. We expand the sum and redistribute
  for each term at most $k_0-1$ factors of $A$ obtaining terms on a  more
  symmetric form, more precisely on the form
  \begin{equation}\label{eq:43m}
    \Re \big\la\inp{p}A^{k_0}\psi(t),B\inp{p}A^{k}\psi(t)\big\ra\textup{
      where }k\leq k_0-1\mand \sup_{n,t}\|B\|<\infty.
  \end{equation} Thanks  to the Cauchy-Schwarz inequality
  and Proposition~\ref{prop:ac-stark-type2}  any expression like
  \eqref{eq:43m} can be integrated on $[t_0,1]$ and the integral is bounded
  uniformly in $n$. In combination with \eqref{eq:42m} we
  conclude that
  \begin{equation*}
    \sup_n\,\|A(1)_{n}^{k_0}\psi(1)\|^2<\infty,
  \end{equation*}
  whence  $\phi=\psi(1)\in \cD(A(1)^{k_0})$.
\end{proof}


\begin{thebibliography}{FGSch2}





\bibitem[AHS]{AHS} S. Agmon, I. Herbst and E. Skibsted,
  \emph{Perturbation of embedded eigenvalues in the generalized
    $N$-body problem},
 Comm. Math. Phys. \textbf{122}  (1989), 411--438.
 
 \bibitem[AC]{AC} J. Aguilar and J.-M. Combes,
\emph{A class of analytic perturbations for one-body Schr{\"o}dinger Hamiltonians},
Comm. Math. Phys., \textbf{22} (1971), 269.
 
\bibitem[ABG]{ABG} W. Amrein, A. Boutet de Monvel and V. Georgescu, \emph{
  $C_0$-groups, commutator methods and spectral theory of $N$-body Hamiltonians},
Basel--Boston--Berlin,  Birkh\"auser, 1996.

\bibitem[BFS]{BFS} V. Bach, J. Fr{\"o}hlich and I.~M. Sigal,
\emph{Quantum electrodynamics of confined non-relativistic particles}, Adv. Math., \textbf{137} (1998), 299--395.

\bibitem[BFSS]{BFSS} V. Bach, J. Fr{\"o}hlich, I.~M. Sigal and A. Soffer,
\emph{Positive commutators and the spectrum of Pauli-Fierz Hamiltonian of atoms and molecules}, Comm. Math. Phys.,
\textbf{207} (1999), 557--587.

\bibitem[BC]{BC} E. Balslev and J.-M. Combes, \emph{Spectral properties of many-body
Schr{\"o}dinger operators with dilation analytic interactions}, Comm. Math. Phys., \textbf{22} (1971), 280--294.

\bibitem[BD]{BD} L.~Bruneau and J.~Derezi\'nski \emph{Pauli-Fierz Hamiltonians defined as quadratic forms},
Rep. On Math. Phys. \textbf{54} (2004), 169--199.

\bibitem[Ca]{Ca} L. Cattaneo, \emph{Mourre{'}s inequality and embedded boundstates},
Bull. Sci. Math. \textbf{129} (2005), 591--614.

\bibitem[CGH]{CGH}
L. Cattaneo, G.~M. Graf and W. Hunziker, \emph{A general resonance theory
  based on Mourre's inequality}, Ann. Henri Poincar{\'e}  \textbf{7}  (2006), 583--601.


\bibitem[DG]{DG} J.~Derezi\'nski and C. G{\'e}rard.
\emph{Asymptotic completeness in quantum field theory. Massive Pauli-Fierz Hamiltonians.}
Rev. Math. Phys. \textbf{11},  (1999),  383--450


\bibitem[DJ1]{DJ}
J. Derezi{\'n}ski and V. Jak$\check {\rm s}$i{\'c},
\emph{Spectral theory of Pauli-Fierz operators}, J.~Funct. Anal. \textbf{180},  (2001), 243--327.

\bibitem[DJ2]{DJ2}
J. Derezi{\'n}ski and V. Jak$\check {\rm s}$i{\'c},
\emph{Return to equilibrium for Pauli-Fierz systems}, Ann. Henri. Poincar{\'e} \textbf{4}  (2003), 739--793.

\bibitem[FMS]{FMS} J.~Faupin, J.~S.~M{\o}ller and E.~Skibsted,
\emph{Second order perturbation theory for embedded eigenvalues}. 

\bibitem[FGSch1]{FGSch1} J.~Fr{\"o}hlich, M.~Griesemer and B.~Schlein,
  \emph{Asymptotic completeness for {R}ayleigh scattering},
  Ann. Henri Poincar{\'e}, \textbf{3} (2002), 107--170.

\bibitem[FGSch2]{FGSch2} J.~Fr{\"o}hlich, M.~Griesemer and B.~Schlein,
  \emph{Asymptotic completeness for {C}ompton scattering},
  Comm. Math. Phys., \textbf{252} (2004), 415--476.

\bibitem[FGSi]{FGS} J. Fr{\"o}hlich, M. Griesemer and I.~M. Sigal,
\emph{Spectral theory for the standard model of non-relativistic QED},
Comm. Math. Phys., \textbf{283} (2008), 613--646.

\bibitem[FM]{FM} J.~Fr{\"o}hlich and M.~Merkli, \emph{Another return of ``return to equilibrium''},
 Comm. Math. Phys., \textbf{251} (2004), 235--262.


\bibitem[GG]{GG}
V. Georgescu and C. G{\'e}rard, \emph{On the virial theorem in quantum mechanics}, Comm. Math. Phys. \textbf{208}
(1999), 275--281.


\bibitem[GGM1]{GGM1}
V. Georgescu, C. G{\'e}rard and J.~S. M{\o}ller, \emph{Commutators,
  $C_0$--semigroups and resolvent estimates}, J. Funct. Anal. \textbf{216}  (2004), 303--361.

\bibitem[GGM2]{GGM}
V. Georgescu, C. G{\'e}rard and J.~S. M{\o}ller, \emph{Spectral theory of
  massless Pauli-Fierz models}, Comm. Math. Phys. \textbf{249}  (2004), 29--78.

\bibitem[G{\'e}]{G} C. G{\'e}rard, \emph{On the scattering theory of massless Nelson models},
Rev. Math. Phys. \textbf{14} (2002), 1165--1280.

\bibitem[Go]{Go} S. Gol{\'e}nia,
\emph{Positive commutators, Fermi Golden Rule and the spectrum of zero temperature Pauli-Fierz Hamiltonians},
J.~Funct. Anal. \textbf{256} (2009), 2587--2620.

\bibitem[GJ]{GJ} S. Gol{\'e}nia and T. Jecko,
\emph{A new look at Mourre's commutator theory}, Compl. Anal. Oper. Theory \textbf{1} (2007), 399--422.

\bibitem[H{\"u}Sp]{HuSp} M.~H{\"u}bner and H.~Spohn,
\emph{Spectral properties of the spin-boson {H}amiltonian},
Ann. Inst. Henri Poincar{\'e}, \textbf{62} (1995), 289--323.

\bibitem[HuSi]{HS} W. Hunziker and I.~M. Sigal,
\emph{The quantum $N$-body problem}, J. Math. Phys. \textbf{41} (2000), 3448-3510.


\bibitem[JP1]{JP1}
V. Jak$\check {\rm s}$i{\'c} and C.-A. Pillet,
\emph{On a model for quantum friction II: Fermi's golden rule and dynamics at positive temperature},
Comm. Math. Phys. \textbf{176}  (1996), 619--644.

\bibitem[JP2]{JP}
V. Jak$\check {\rm s}$i{\'c} and C.-A. Pillet,
\emph{Spectral theory of Pauli-Fierz operators}, J.~Funct. Anal. \textbf{180}  (2001), 243--327.

\bibitem[KY]{KY} Y.~Kuwabara and K.~Yajima,
\emph{The limiting absorption principle for Schr{\"o}dinger operators with long-range time-periodic potentials},
J.~Fac. Sci. Univ. Tokyo Sect. IA Math. \textbf{34} (1987), 833--851.

\bibitem[Mo]{Mo}
{\'E}. Mourre, \emph{Absence of singular continuous spectrum
for certain  selfadjoint operators}, Comm. Math. Phys. \textbf{78} (1980/81),  391--408.


\bibitem[MS]{MS}
J.~S. M{\o}ller and E. Skibsted, \emph{Spectral theory of time-periodic
  many-body systems}, Advances in Math.  \textbf{188}
 (2004), 137--221.

 \bibitem[MW]{MW} J.~S. M{\o}ller and M.~Westrich, \emph{Regularity of eigenstates in regular Mourre theory}, arXiv:1006.0410.

\bibitem[RS]{RS}
M.~Reed and B.~Simon, \emph{Methods of modern mathematical physics {I}-{IV}},
  New York, Academic Press 1972-78.

\bibitem[Si]{SiBook} B.~Simon, \emph{Quantum mechanics for Hamiltonians defined as quadratic forms},
Princeton Series in Physics. Princeton University Press, Princeton, N. J., 1971. xv+244 pp.

\bibitem[Sk]{Sk}
 E. Skibsted, \emph{Spectral analysis of $N$-body systems coupled to
   a bosonic field}, Rev. Math. Phys.  \textbf{10}  (1998), 989--1026.

\bibitem[Ya]{Ya} K.~Yajima,
\emph{Existence of solutions for Schr{\"o}dinger evolution equations},
Comm. Math. Phys. \textbf{110} (1987), 415--426.

\end{thebibliography}
\end{document}